\documentclass[a4paper,USenglish,cleveref,nameinlink,autoref,thm-restate]{lipics-v2021}
\hideLIPIcs
\usepackage{multirow}
\usepackage[table]{xcolor}
\usepackage{soul}
\nolinenumbers

\usepackage[table]{xcolor}
\usepackage{graphicx}
\usepackage{hhline}

\usepackage{makecell}

\usepackage{float}
\usepackage[most]{tcolorbox}
\usepackage[inline]{enumitem}
\usepackage{amsfonts,amsthm,amsmath}
\usepackage{amstext}
\usepackage{algorithmic}
\usepackage{graphics}
\usepackage[all]{nowidow}
\usepackage{microtype}
\usepackage[bottom]{footmisc}
\usepackage{hyperref}
\usepackage[forwardlinking=yes]{apxproof} 
\usepackage{complexity}

\definecolor{lightblue}{HTML}{cceeff}

\definecolor{lipicsblue}{rgb}{0.08235294118,0.3098039216,0.537254902}
\hypersetup{
    colorlinks=true,
    linkcolor=lipicsblue,
    filecolor=lipicsblue,      
    urlcolor=lipicsblue,
    citecolor=lipicsblue,
    pdfpagemode=FullScreen,
    }

\usepackage{amsmath,amssymb,amsthm}
\usepackage{xspace}

\newcommand{\shapeDesc}[8]{\ensuremath{\langle #1,\allowbreak #2,\allowbreak #3,\allowbreak #4,\allowbreak #5,\allowbreak #6,\allowbreak #7,\allowbreak #8 \rangle}}

\newcommand{\neworrenewcommand}[1]{\providecommand{#1}{}\renewcommand{#1}}

\newcommand{\shapeDescPUBE}[9]{
    \neworrenewcommand{\ffoo}[1]{
        \ensuremath{\langle #1,\allowbreak #2,\allowbreak #3,\allowbreak #4,\allowbreak #5,\allowbreak #6,\allowbreak #7,\allowbreak #8,\allowbreak #9,\allowbreak ##1\rangle}
    }
    \ffoo
}

\title{Upward Book Embeddings of Partitioned Digraphs}

\author{Giordano {Da Lozzo}}{ICITA Department, Roma Tre University, Italy}{giordano.dalozzo@uniroma3.it}{http://orcid.org/0000-0003-2396-5174}{}

\author{Fabrizio Frati}{ICITA Department, Roma Tre University, Italy}{fabrizio.frati@uniroma3.it}{https://orcid.org/0000-0001-5987-8713}{}

\author{Ignaz Rutter}{Faculty of Computer Science and Mathematics, University of Passau, Germany}{rutter@fim.uni-passau.de}{https://orcid.org/0000-0002-3794-4406}{}

\authorrunning{G. Da Lozzo, F. Frati, and I. Rutter}
\Copyright{Giordano {Da Lozzo,} Fabrizio Frati, and Ignaz Rutter}

\usepackage{todonotes}

\newtheorem{property}{Property}

\newcommand{\blue}[1]{{{\textcolor{blue}{#1}\xspace}}}
\renewcommand{\emph}[1]{\blue{\bf \em #1}}

\Crefname{observation}{Observation}{Observations}
\Crefname{algorithm}{Algorithm}{Algorithms}
\Crefname{section}{Sect.}{Sects.}
\Crefname{observation}{Observation}{Observations}
\Crefname{lemma}{Lemma}{Lemmas}
\Crefname{corollary}{Corollary}{Corollaries}
\Crefname{claimx}{Claim}{Claims}
\Crefname{figure}{Fig.}{Figs.}
\Crefname{figure}{Fig.}{Figs.}
\Crefname{invariant}{Inv.}{Invs.}
\Crefname{enumi}{Condition}{Conditions}
\Crefname{property}{Property}{Properties}
\Crefname{assumption}{Assumption}{Assumptions}

\relatedversion{A preliminary version of the paper appears at the 42nd International Symposium on Computational Geometry (SoCG '26).}

\begin{document}
\maketitle

\keywords{upward book embeddings, partitioned digraphs, SPQ-trees, $2$-trees}

\ccsdesc[500]{Theory of computation~Computational geometry}
\ccsdesc[500]{Mathematics of computing~Graph algorithms}
\ccsdesc[500]{Theory of computation~Design and analysis of algorithms}

\begin{abstract}
In 1999, Heath, Pemmaraju, and Trenk [SIAM J. Comput. 28(4), 1999] extended the classic notion of book embeddings to digraphs, introducing the concept of {\em upward book embeddings}, in which the vertices must appear along the spine in a topological order and the edges are partitioned into pages, so that no two edges in the same page cross. For a partitioned digraph~$G=(V,\bigcup^k_{i=1} E_i)$, that is, a digraph whose edge set is partitioned into~$k$ subsets, an upward book embedding is required to assign edges to pages as prescribed by the given partition. In a companion paper, Heath and Pemmaraju [SIAM J. Comput 28(5), 1999] proved that the problem of testing the existence of an upward book embedding of a partitioned digraph is linear-time solvable for~$k=1$ and recently Akitaya, Demaine, Hesterberg, and Liu [GD, 2017] have shown the problem \NP-complete for~$k\geq 3$.
In this paper, we study upward book embeddings of partitioned digraphs and focus on the unsolved case~$k=2$. Our first main result is a novel characterization
of the {\em upward embeddings} that support an upward book embedding in two pages.
We exploit this characterization in several ways, and obtain a rich picture of the complexity landscape of the problem.
First, we show that the problem remains \NP-complete when~$k=2$, thus closing the complexity gap for the problem.
Second, we show that, for an~$n$-vertex partitioned digraph~$G$ with a prescribed planar embedding, the existence of an upward book embedding of~$G$ that respects the given planar embedding can be tested in~$O(n \log^3 n)$ time. Finally, leveraging the SPQ(R)-tree decomposition of biconnected graphs into triconnected components, we present a cubic-time testing algorithm for biconnected directed partial $2$-trees.
%
\end{abstract}

\setcounter{tocdepth}{1}


\newcommand{\putbe}{partitioned upward 2-page book embedding\xspace}
\section{Introduction}

\begin{figure}[b!]
    \centering
    \begin{subfigure}{.32\textwidth}\centering
\includegraphics[page=5,width=\textwidth]{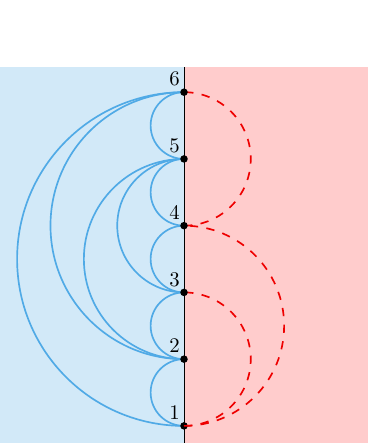}
    \subcaption{\label{fig:book-undirected}}
    \end{subfigure}
    \begin{subfigure}{.32\textwidth}\centering
    \includegraphics[page=6,width=\textwidth]{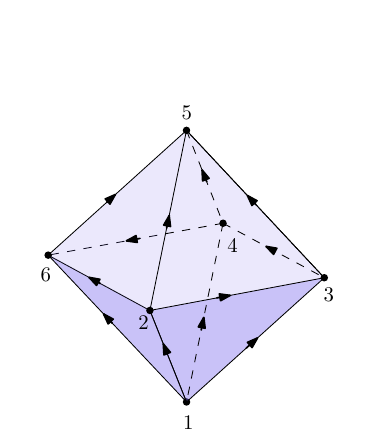}
    \subcaption{\label{fig:octahedron}}
    \end{subfigure}
    \begin{subfigure}{.32\textwidth}\centering
\includegraphics[page=7,width=\textwidth]{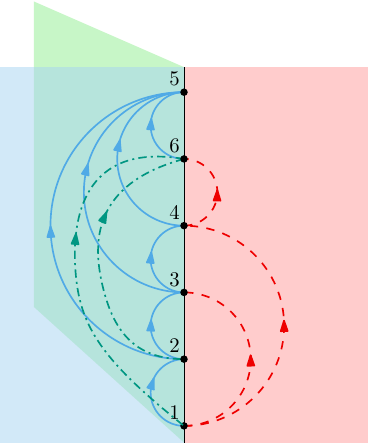}
    \subcaption{\label{fig:book-directed}}
    \end{subfigure}
    \caption{(a) A book embedding of the octahedron in $2$ pages. (b) An orientation of the octahedron. (c) An upward book embedding of the directed octahedron in (b) in $3$ pages, which is optimal.}\label{fig:bookEmbeddingExample}
\end{figure}

Book embeddings are a classic and influential topic in combinatorial and algorithmic graph theory. The notion of book as a topological space was introduced in the late 60s by Persinger~\cite{MR195077} and Atneosen~\cite{MR293592}, and later developed in its current and more popular form by the seminal work of Ollmann~\cite{Oll73}. In a \emph{book embedding} of a graph~$G=(V,E)$, all vertices lie along a line—referred to as the \emph{spine}—while edges are placed into distinct half-planes bounded by the spine, known as the \emph{pages} of the book. Therefore, constructing such an embedding for~$G$ amounts to computing a pair~$(\pi,\sigma)$, where~$\pi: V \leftrightarrow \{1,\dots,|V|\}$ is a linear ordering of the vertices and~$\sigma: E \rightarrow \{1,\dots,k\}$ is a partition of the edges into~$k$ pages so that no two edges in the same page cross according to~$\pi$, i.e., their end-vertices do not alternate in~$\pi$; see \cref{fig:book-undirected} for an example. The minimum value of~$k$ for which this is possible is the \emph{book thickness} of~$G$  (also called \emph{stack number} or \emph{page number}) and the \emph{book thickness} of a graph class~$\cal G$ is the maximum book thickness among all graphs in~$\cal G$. 

Research on book embeddings and book thickness originated from problems in VLSI circuit design~\cite{Chung87}, and has since found applications in a variety of domains. These include sorting permutations~\cite{Pratt73,Tarjan72}, fault-tolerant processing~\cite{Rosenberg83}, compact graph encodings~\cite{Jacobson89,Munro01}, graph drawing~\cite{Biedl1999,DBLP:journals/algorithmica/GiacomoDLW06,Everett10,GIORDANO201545}, computational origami~\cite{Akitaya18,Morgan2012MapFolding}, parallel process scheduling~\cite{BHATT199655}, and parallel matrix computations~\cite{heath1993sparse}, among others. For additional references and a more comprehensive overview of applications, see e.g.~\cite{Dujmović2004}. The notion of book embedding was extended to digraphs by Heath, Pemmaraju, and Trenk~\cite{HeathPT99} by introducing the natural requirement that in a book embedding~$(\pi,\sigma)$ of a digraph~$G$ the ordering~$\pi$ must be a topological ordering of~$G$; see \cref{fig:octahedron,fig:book-directed} for an example. Such book embeddings are called \emph{upward book embeddings} as they are naturally depicted with vertices placed on a vertical line and edges drawn as arcs monotonically increasing in the~$y$-direction in their page. Next, we provide an overview of the major results on book embeddings.

\subparagraph{Undirected graphs.} In 1979, Bernhart and Kainen~\cite{BERNHART1979320} showed that the graphs of book thickness~$1$ are exactly the outerplanar graphs and that the graphs of book thickness~$2$ are exactly the sub-Hamiltonian planar graphs. Whereas the former are known to be recognizable in linear time~\cite{10.1007/3-540-17218-1_57}, recognizing sub-Hamiltonian planar graphs is \NP-complete, even for planar triangulations~\cite{Wig82}. Several classes of planar graphs are known to admit a book embedding in two pages, e.g., 4-connected planar graphs~\cite{Tutte56}, planar graphs without separating triangles~\cite{DBLP:journals/appml/KainenO07}, planar graphs of maximum degree~4~\cite{DBLP:journals/algorithmica/BekosGR16}, triconnected planar graphs of maximum degree~5~\cite{DBLP:conf/esa/0001K19}, maximal planar graphs of maximum degree~6~\cite{Ewald1973}, Halin graphs~\cite{DBLP:journals/mp/CornuejolsNP83}, series-parallel graphs~\cite{DBLP:conf/cocoon/RengarajanM95}, and bipartite planar graphs~\cite{DBLP:journals/dcg/FraysseixMP95}.
Recently, Ganian et al.~\cite{DBLP:conf/icalp/GanianMOPR24} presented a~$2^{O{(\sqrt{n})}}$-time algorithm for testing the existence of a book embedding of an~$n$-vertex graph on two pages--a bound which is asymptotically tight under ETH.
Perhaps the most celebrated result concerning book thickness is the one due to Yannakakis, who showed that every planar graph has book thickness at most~$4$~\cite{DBLP:conf/stoc/Yannakakis86,DBLP:journals/jcss/Yannakakis89}. This upper bound was only recently shown to be tight independently by Yannakakis~\cite{YANNAKAKIS2020241}
and Bekos et al.~\cite{Kaufmann2020}.
For more results on book thickness see also~\cite{DBLP:journals/corr/AlamBK15,DBLP:journals/algorithmica/BekosBKR17,BEKOS2024113690,Bla03,DBLP:conf/stoc/BussS84,DBLP:journals/dcg/DujmovicW07,DBLP:journals/dam/GanleyH01,DBLP:journals/corr/GuanY2018,DBLP:conf/focs/Heath84,Istrail1988a,DBLP:journals/jal/Malitz94a,DBLP:journals/jal/Malitz94}. Finally, we remark that, for arbitrary~$k$, the problem of testing the existence of a book embedding in~$k$ pages is known to be fixed-parameter tractable (FPT) with respect to the vertex cover number~\cite{DBLP:journals/jgaa/BhoreGMN20} and the feedback edge number~\cite{DBLP:conf/icalp/GanianMOPR24}.


\subparagraph{Directed graphs.} On the combinatorial side, a large body of research has directed its focus toward establishing upper and lower bounds on the book thickness of digraphs. Tight upper bounds have long been known for directed trees and unicyclic
digraphs~\cite{HeathPT99}, for series-parallel digraphs~\cite{DBLP:conf/gd/AlzohairiR96,DBLP:journals/algorithmica/GiacomoDLW06}, and for N-free upward planar digraphs~\cite{DBLP:conf/isaac/MchedlidzeS09}.
In~\cite{HeathPT99}, Heath, Pemmaraju and Trenk conjectured a constant upper bound for the book thickness of outerplanar digraphs. The conjecture was first confirmed for several families of outerplanar digraphs by Bhore et al.~\cite{DBLP:journals/ejc/BhoreLMN23} and by N\"ollenburg and Pupyrev~\cite{DBLP:conf/gd/NollenburgP23}, and finally settled by Jungeblut, Merker, and Ueckerdt~\cite{10353199}. The major unsolved question in this area is the one posed more than 30 years ago by Nowakowski and Parker~\cite{Nowakowski89}  of whether planar posets,
and more generally upward planar digraphs, have bounded book thickeness. Recently, Jungeblut, Merker, and Ueckerdt~\cite{DBLP:journals/siamdm/JungeblutMU23} presented the first sublinear upper bound on the
page number of upward planar graphs.
A large body of research has devoted its attention to testing the existence of a book embedding of a DAG in~$k$ pages. The problem is called \textsc{Upward Book Embedding}. 
For more than two decades, the only known \NP-completeness result for the problem was the one shown by Heath and Pemmaraju~\cite{HeathP99} when~$k=6$. Recently, in two subsequent papers, Binucci et al.~\cite{BinucciLGDMP23} and Bekos et al.~\cite{DBLP:journals/tcs/BekosLFGMR23} closed the computational gap by showing \NP-completeness for~$k\geq 3$ and~$k=2$, respectively. These results, together with the linear-time algorithm for testing the existence of 1-page book embeddings of DAGs~\cite{HeathP99}, completely characterize the complexity of the \textsc{Upward Book Embedding} problem with respect to the number of available pages. For~$k=2$, efficient algorithms have been devised for outerplanar and planar triangulated~$st$-graphs~\cite{DBLP:journals/jgaa/MchedlidzeS11}, and the problem is known to be FPT for~$st$-graphs of bounded treewidth~\cite{DBLP:conf/compgeom/BinucciLGDMP19} and for~$st$-graphs whose vertices can be covered by a bounded number of directed paths~\cite{DBLP:conf/isaac/MchedlidzeS09}. For arbitrary~$k$, the \textsc{Upward Book Embedding} problem has been proved FPT with respect to the vertex cover number~\cite{DBLP:journals/ejc/BhoreLMN23,DBLP:journals/ijfcs/LiuLH24}.

\subparagraph{Partitioned book embeddings.} 

In the construction of a book embedding of a (di)graph, one is allowed to select a vertex ordering~$\pi$ and a page assignment~$\sigma$.  Since, as discussed, determining the existence of such a pair~$(\pi,\sigma)$ so to minimize the number of pages is \NP-hard, it is natural to study the complexity of the problem if $\pi$ or $\sigma$ is given as part of the input.  

Determining an assignment on~$k$ pages for a fixed vertex ordering $\pi$ naturally corresponds to a~$k$-coloring problem on circle graphs. Observe that, in this case, the undirected and directed versions retain the same complexity, as a linear-time pre-processing can be used to reject directed instances for which the prescribed vertex ordering is not a topological ordering. The problem is called \textsc{Fixed-Order Book Embedding} and is clearly polynomial-time solvable for~$k \le 2$. Unger~\cite{DBLP:conf/stacs/Unger88} showed that it is \NP-complete for~$k \ge 4$. The complexity of the case~$k=3$ is still unsolved~\cite{DBLP:conf/gd/BachmannRS23,DBLP:journals/jgaa/BachmannRS24}, although a quasi-polynomial-time algorithm has recently been proposed~\cite{AjayGanianLS}. For arbitrary~$k$, FPT algorithms for \textsc{Fixed-Order Book Embedding} have also been presented with respect to the vertex cover number~\cite{DBLP:journals/jgaa/BhoreGMN20} and the pathwidth of the vertex ordering~\cite{DBLP:journals/jgaa/BhoreGMN20,DBLP:journals/tcs/LiuCHW21}.

The complementary problem asks, for a given partition~$\sigma$ of the edges in~$k$ pages, whether there is an ordering~$\pi$ of the vertices that yields a book embedding. In the case of undirected graphs, this problem is polynomial-time solvable for~$k=1$~\cite{HeathP99},~$k=2$~\cite{DBLP:conf/gd/AngeliniBB12,hn-tpbecg-09,DBLP:journals/tcs/HongN18} and \NP-complete for~$k \ge 3$~\cite{DBLP:journals/tcs/AngeliniLN15}.  
This paper studies the complexity of this problem for directed graphs, i.e., for upward book embeddings, in which~$\pi$ is required to be a topological ordering of the input graph. This problem is called \textsc{Partitioned Upward Book Embedding}.  For~$k=1$ it coincides with the ``unpartitioned'' case already solved in~\cite{HeathP99,HeathPT99}.  For~$k>1$, it was first studied systematically by Akitaya et al.~\cite{Akitaya18}.  They connected the problem to applications in map folding~\cite{Morgan2012MapFolding} and attributed it to Edmonds, who,  already in 1997, posed the question specifically for~$k=4$ when the edges assigned to each page form a matching. They showed that the problem is \NP-complete for~$k\ge 3$, it is \NP-complete for~$k\ge 4$ even if the edges in each page form a matching, and they gave a linear-time algorithm for the case~$k=2$ when the edges in each page form a matching.
\cref{table:complexity} provides a comprehensive view of the complexity of the problem of computing upward book~embeddings~of~digraphs.

\subparagraph{Our Contributions.} 
In this paper, we study the \textsc{Partitioned Upward Book Embedding} problem and focus on the unsolved case~$k=2$.  Note that every upward book embedding is an \emph{upward planar drawing}, i.e., a planar drawing where each edge appears as a~$y$-monotone curve. The topological information in an upward planar drawing is represented by the concept of \emph{upward embedding}~\cite{DBLP:journals/algorithmica/BertolazziBLM94}.  Our first main result is a characterization of the upward embeddings that support an upward book embedding in two pages.
We exploit this characterization in several ways and obtain a rich~picture~of~the~complexity~landscape~of~the~problem. 

First, we show that the \textsc{Partitioned Upward Book Embedding} problem remains \NP-complete when~$k=2$, thus closing the complexity gap for the problem and exhibiting a sharp contrast with the undirected case, in which the problem is linear-time solvable~\cite{DBLP:journals/tcs/HongN18}.
Our proof also implies that the problem is W[1]-hard with respect to the treewidth. 


Second, we show that, for an~$n$-vertex digraph with a prescribed planar embedding, the existence of an upward book embedding that respects the given planar embedding can be tested in~$O(n \log^3 n)$ time. Our algorithm is inspired by the network-flow approach of Bertolazzi et al.~\cite{DBLP:journals/algorithmica/BertolazziBLM94} and requires the use of several non-trivial ingredients arising from our characterization.

Finally, we present a cubic-time algorithm that tests the existence of an upward book embedding for a given biconnected directed partial $2$-tree $G$. Our algorithm exploits a compact representation (called {\em descriptor pair}) of the features of an upward embedding of a subgraph of $G$ that are relevant for its extensibility to an upward embedding, satisfying our characterization, of $G$. This allows us to compute, via a bottom-up dynamic programming algorithm built on the SPQ(R)-tree decomposition of $G$, which descriptor pairs are realizable by each subgraph associated with a node of the SPQ(R)-tree.

In the description of our algorithms, we focus on the decision problem, however they can be made constructive in order to yield the desired upward book embeddings, if any.

\begin{table}[tb!]
\begin{tabular}{p{1cm}p{1.3cm}|p{1cm}p{5.5cm}|}
\cline{3-4}
\multicolumn{2}{c|}{\multirow{2}{*}{}}                                              & \multicolumn{2}{c|}{{\textbf{vertex order~$\mathbb{\pi}$}}}              \\ \cline{3-4} 
\multicolumn{2}{c|}{}                                                               & \multicolumn{1}{c|}{{\textbf{fixed}}} & \multicolumn{1}{c|}{{ \textbf{variable}}} \\ \hline
\multicolumn{1}{|c|}{\multirow{2}{*}{\multirow{2}{*}{\rotatebox[origin=c]{90}{\textbf{\shortstack{page\\ assignment~$\mathbb{\sigma}$}}}}}} & \multicolumn{1}{c|}{\textbf{fixed}}    &  \multicolumn{1}{l|}{ O(n+m) time~$\checkmark$ }               &      1 page: O(n) time~\cite{HeathP99}\newline {\setlength{\fboxrule}{1pt}
				\fcolorbox{lipicsblue}{lightblue}{2 pages: \NP-complete (\cref{th:np-complete})}}\newline~$\geq 3$ pages: \NP-complete~\cite{Akitaya18}             \\ \cline{2-4} 
\multicolumn{1}{|c|}{}                                          & \textbf{variable}  & \multicolumn{1}{p{5cm}|}{
2 pages: O(n) time~$\checkmark$ \newline
\fcolorbox{red}{pink}{3 pages: OPEN~\cite{DBLP:journals/jgaa/BachmannRS24}}  \newline
$\geq 4$ pages: \NP-complete~\cite{DBLP:conf/stacs/Unger88}
     }               &       
     2 pages: \NP-complete~\cite{DBLP:journals/tcs/BekosLFGMR23} \newline
    ~$\geq 3$ pages: \NP-complete~\cite{BinucciLGDMP23} 
     \\ \hline
\end{tabular}
\caption{\label{table:complexity}
Known results on the computational complexity of testing the existence of an upward book embedding of a digraph. Results marked with the symbol~$\checkmark$ are trivial.}
\end{table}

\section{Preliminaries} \label{se:preliminaries}

All the graphs considered in this paper are finite and \emph{simple}, i.e., they contain neither self-loops nor multiple edges. For an integer~$k>0$, let~$[k]$ denote the set~$\{1,\dots,k\}$.

\subparagraph*{Digraphs.}
A \emph{directed graph}~$G$, or \emph{digraph}, is a graph whose edges have an assigned orientation.  
We denote a directed edge as~$(u,v)$ if it is oriented from~$u$ to~$v$; then we say that~$u$ is the \emph{tail} and~$v$ is the \emph{head} of the edge. A vertex~$u$ of~$G$ is a \emph{source} (resp.\ a \emph{sink}) if it is the tail (resp.\ the head) of all its incident edges. A vertex that is a source or a sink is a \emph{switch}.
A \emph{directed acyclic graph} (for short, \emph{DAG}) is a digraph with no directed cycle. 

\subparagraph*{Drawings and embeddings.} A \emph{drawing} of a digraph maps each vertex to a distinct point in~$\mathbb{R}^2$ and each edge to a Jordan arc connecting the images of its end-vertices, so that each arc that is the image of an edge does not contain any point that is the image of a vertex, except for its end-points. A drawing of a digraph is \emph{planar} if no two edges cross. A digraph is \emph{planar} if it admits a planar drawing. A planar drawing partitions the plane into topologically connected regions, called \emph{faces}. The unique unbounded face is the \emph{outer face}, whereas the remaining faces are the \emph{internal faces}. The set of edges incident to a face forms its \emph{boundary}. Such edges determine a collection of walks or a single walk if the digraph is connected. Two planar drawings of a connected digraph are \emph{topologically equivalent} if they have the same clockwise order of the edges around each vertex and the same clockwise order of the edges along the boundary of the outer face. For disconnected digraphs, the notion of topological equivalence additionally comprises the \emph{relative positions} of connected components to one another, that is, the information about the containment of each connected component inside the regions of the plane delimited by the face boundaries
of another connected component\footnote{In the literature, topological equivalence between two drawings of a disconnected digraph sometimes only requires topological equivalence between the drawings of the connected components of the digraph, without taking into account the relative positions of the components. This relaxed notion of equivalence would make our research easier, as the existence of an upward book embedding of digraph would boil down to the existence of an upward book embedding for each connected component of the digraph.}. A \emph{planar embedding} is a class of topologically equivalent planar drawings. 
All the drawings in the same equivalence class \emph{respect} the embedding defined by the class.
A \emph{plane digraph} is a planar digraph together with a planar embedding. A planar embedding of a digraph~$G$ is \emph{bimodal} if, for every vertex~$v$, the edges of~$G$ that have their tail at~$v$ are consecutive in the clockwise order of the edges incident to~$v$. 



\subparagraph*{Upward planarity.} 
A drawing of a digraph is \emph{upward} if every edge is represented by a Jordan arc that is strictly increasing in the~$y$-direction from its tail to its head. A digraph that admits an upward planar drawing is an \emph{upward planar digraph}. Clearly, an upward planar digraph is a DAG, as a directed cycle cannot be drawn upward. A \emph{planar~$st$-graph} is a DAG with one source~$s$ and one sink~$t$ that admits a planar embedding in which~$s$ and~$t$ are on the boundary of the outer face. A planar~$st$-graph equipped with a planar embedding such that~$s$ and~$t$ are incident to the outer face is a \emph{plane~$st$-graph}. A plane~$st$-graph~$G$ always admits an upward planar drawing respecting the planar embedding of~$G$~\cite{DBLP:journals/tcs/BattistaT88}. A face~$f$ of a planar embedding~$\mathcal E$ of a plane digraph is an \emph{$st$-face} if its boundary consists of two directed paths. These are called \emph{left path} and \emph{right path} of~$f$, and are denoted by~$\ell_f$ and~$r_f$, respectively, where the edge of~$\ell_f$ incident to the source~$s_f$ of the cycle bounding~$f$ immediately precedes the edge of the right path incident to~$s_f$ in the clockwise order of the edges incident to~$s_f$ if~$f$ is an internal face, or immediately follows it if~$f$ is the outer face. The left and right path of the outer face are also called \emph{leftmost} and \emph{rightmost} path of~$G$, respectively. 



Let~$\mathcal E$ be a bimodal planar embedding of a digraph~$G=(V,E)$.  For~$v\in V$, an \emph{angle at~$v$} is an ordered pair~$(e_1,e_2)$ of edges incident to~$v$ such that~$e_2$ immediately follows~$e_1$ in the clockwise order of edges around~$v$.  An angle~$(e_1,e_2)$ is a \emph{switch angle} if~$v$ is the head or the tail of both~$e_1$ and~$e_2$.  An angle that is not a switch angle is a \emph{flat angle}. For a vertex~$v$, we denote by~$A_{\mathcal E}(v)$ the set of angles incident to~$v$. Also, we denote by~$\mathcal{A}_{\mathcal E} = \bigcup_{v \in V} A_{\mathcal E}(v)$ the set of angles of~$\mathcal E$ and, for a face~$f$, by~$A_{\mathcal E}(f)$ the set of angles incident to~$f$. 

An \emph{angle assignment for~$\mathcal E$} is a function~$\lambda \colon \mathcal{A}_{\mathcal E} \to \{-1,0,1\}$.  
An angle assignment is \emph{upward-consistent} if it satisfies the following conditions:
\begin{enumerate}[label={\bf C\arabic*}]
    \item \label{cond:flat-assignment} For each angle~$a\in \mathcal{A}_{\mathcal E}$:~$\lambda(a) = 0$ if and only if~$a$ is flat.
    \item \label{cond:vertex-assignment} For each vertex~$v \in V$:~$\sum_{a \in A_{\mathcal E}(v)} \lambda(a) = 2-\deg(v)$.
    \item \label{cond:face-assignment} For each face~$f$ of~$\mathcal E$:~$\sum_{a \in A_{\mathcal E}(f)} \lambda(a) = -2$ if~$f$ is an internal face and~$\sum_{a \in A_{\mathcal E}(f)} \lambda(a) = 2$ if~$f$ is the outer face.
\end{enumerate}

\begin{figure}
    \centering
    \begin{subfigure}{.45\textwidth}\centering
\includegraphics[page=7,width=\textwidth]{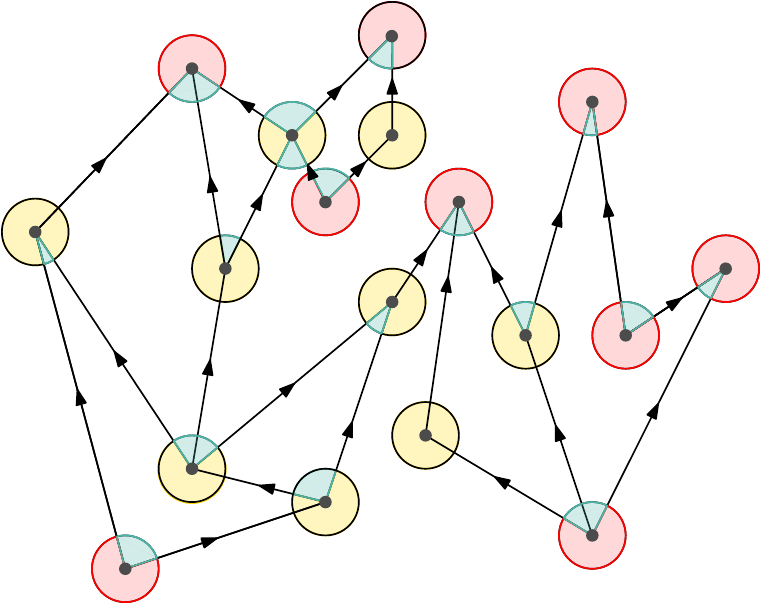}
    \subcaption{\label{fig:upward-drawing}}
    \end{subfigure}
    \hfil
    \begin{subfigure}{.45\textwidth}\centering
\includegraphics[page=8,width=\textwidth]{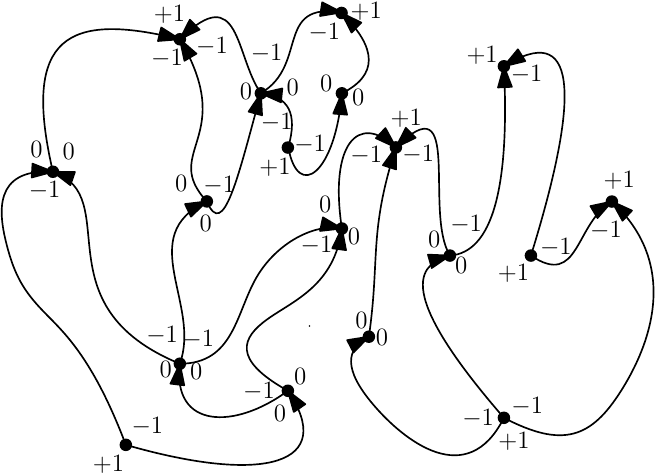}
    \subcaption{\label{fig:upward-consistent}}
    \end{subfigure}
    \caption{(a) An upward planar drawing~$\Gamma$ and its big, small, and flat angles  depicted as red, green, and yellow sectors, respectively. (b) The upward-consistent angle assignmentas~$\lambda_\Gamma$ defined~by~$\Gamma$.}\label{fig:angle-assignment}
\end{figure}

Let~$\Gamma$ be an upward planar drawing of a digraph~$G$ with planar embedding~$\mathcal E$; refer to \cref{fig:angle-assignment}.  Clearly,~$\mathcal E$ is bimodal.  Moreover, it defines an angle assignment~$\lambda_\Gamma$ as follows.  Let~$a=(e_1,e_2)$ be an angle of~$\mathcal E$. If~$a$ is flat, then we define~$\lambda_\Gamma(a) = 0$. Otherwise, consider the geometric angle~$\alpha$ in~$\Gamma$ corresponding to~$a$, i.e., lying clockwise after~$e_1$ and before~$e_2$.  We define~$\lambda_\Gamma(a) = -1$ if~$\alpha<\pi$ and~$\lambda_\Gamma(a) = 1$ if~$\alpha>\pi$.  Observe that~$\alpha\neq \pi$, since~$a$ is a switch angle.  It is not hard to see that~$\lambda_\Gamma$ is upward-consistent~\cite{DBLP:journals/algorithmica/BertolazziBLM94,DBLP:journals/siamdm/DidimoGL09}.

A pair~$(\mathcal E, \lambda)$ of a planar embedding~$\mathcal E$ and an angle assignment~$\lambda$ is an \emph{upward embedding} if there exists an upward planar drawing~$\Gamma$ with embedding~$\mathcal E$ such that~$\lambda = \lambda_\Gamma$.  
We have that angles~$a$ with~$\lambda(a) = 0$ are flat; also, we call an angle~$a$ \emph{large} if~$\lambda(a) = 1$ and \emph{small} if~$\lambda(a) = -1$.  An angle assignment is completely determined by {\em assigning} each switch to one of its incident angles, which corresponds to making that angle large; the switch angles to which no switch is assigned are small.

\begin{theorem}[\cite{DBLP:journals/algorithmica/BertolazziBLM94,DBLP:journals/siamdm/DidimoGL09}]
    \label{th:upward-conditions}
    Let~$G$ be a digraph, let~$\mathcal E$ be a planar embedding of~$G$, and let~$\lambda$ be an angle assignment for~$\mathcal E$.  Then the pair~$(\mathcal E, \lambda)$ is an upward embedding if and only if~$\mathcal E$ is bimodal and~$\lambda$ is upward-consistent.
\end{theorem}

A digraph equipped with an upward embedding is an \emph{upward plane digraph}. 

In the paper, we often say that a face or an angle is \emph{to the left} (or \emph{to the right}) of a directed path, possibly a single edge. This means that the face or angle is to the left (to the right) of the directed path when traversing the path according to its orientation.

\subparagraph*{Upward Book Embeddings.} 
A \emph{partitioned digraph} is a digraph~$G=(V,\bigcup^k_{i=1} E_i)$, whose edge set is partitioned into~$k$ sets.
An \emph{upward book embedding} (in~$k$ pages)  of an~$n$-vertex partitioned digraph~$G=(V,\bigcup^k_{i=1} E_i)$ is a bijection~$\pi \colon V \leftrightarrow \{1, \dots, n\}$ such that:
\begin{enumerate}[label={\bf(\roman*)}]
\item 
for each edge~$e=(u,v)$, it holds that~$\pi(u) < \pi(v)$, i.e., the tail of~$e$ precedes the head of~$e$ according to~$\pi$, and 
\item for any~$i \in [k]$, no two edges~$(u,v), (w,x) \in E_i$ cross, where~$(u,v)$ and~$(w,x)$ \emph{cross} if~$\pi(u) < \pi(w) < \pi(v) < \pi(x)$ or~$\pi(w) < \pi(u) < \pi(x) < \pi(v)$, i.e., their end-vertices interleave in the total order of~$V$ defined~by~$\pi$.
\end{enumerate}
In this paper, we focus on the case~$k=2$ and denote the sets~$E_1$ and~$E_2$ as~$L$ and~$R$, respectively. We often omit that our upward book embeddings are in two pages and just talk about upward book embeddings.  Also, we call the edges in~$L$ \emph{left edges} and the edges in~$R$  \emph{right edges}. This terminology is motivated by the fact that an upward book embedding~$(\pi,\sigma)$ of~$G = (V,L \cup R)$ determines an upward planar drawing of~$G$ as follows:  
The vertices of~$G$ lie along a vertical line, called \emph{spine}, so that the~$y$-coordinate of each vertex~$v$ is~$\pi(v)$, and each edge~$(u,v)$ in~$L$ (resp.\ in~$R$) is drawn as a semi-circle with diameter~$\pi(v)-\pi(u)$ to the left (resp.\ to the right) of the spine. We often implicitly refer to such a representation and say that in an upward book embedding the edges lie to the left or to the right of the spine. In all the illustrations, the edges in $L$ are blue and solid, while the edges in $R$ are red and dashed.


An upward planar drawing \emph{respects} a planar embedding~$\mathcal E$ if it belongs to the equivalence class~$\mathcal E$. An upward book embedding~$\Gamma$ \emph{respects} a planar embedding~$\mathcal E$ if the upward planar drawing associated with~$\Gamma$ respects~$\mathcal E$, and it \emph{respects} an upward embedding~$(\mathcal E,\lambda)$ if~it respects~$\mathcal E$ and~$\lambda_\Gamma = \lambda$.


Let~$G=(V,L \cup R)$ be a partitioned digraph with an upward embedding~$(\mathcal E,\lambda)$. A vertex~$v \in V$ is \emph{4-modal} if it satisfies the following condition: 
\begin{enumerate}[label={\bf(\roman*)}]
\item if~$v$ is a non-switch vertex, then in clockwise order around~$v$ in~$\mathcal E$ we have all the outgoing left edges, all the outgoing right edges, all the incoming right edges, and all the incoming left edges; one of the former two sets and/or one of the latter two sets might be empty.
\item if~$v$ is a source (resp.\ a sink), then in clockwise order around~$v$ in~$(\mathcal E,\lambda)$ we have the large angle at~$v$, all the outgoing left edges, and all the outgoing right edges (resp.\ the large angle at~$v$, all the incoming right edges, and all the incoming left edges); one of the two sets of outgoing (resp.\ incoming) edges might be empty.
\end{enumerate}
We say that~$(\mathcal E,\lambda)$ is \emph{4-modal} if all the vertices of~$G$ are 4-modal.


\begin{property}\label{pr:4-modal}
Let~$\Gamma$ be an upward book embedding of a partitioned digraph~$G=(V,L \cup R)$ and let~$(\mathcal E,\lambda)$ be the upward embedding of~$G$ defined by~$\Gamma$. Then~$(\mathcal E,\lambda)$ is 4-modal.   
\end{property}

\begin{proof}
Consider any non-switch vertex~$v$ of~$G$. Since~$\Gamma$ is an upward book embedding, the left (right) edges are to the left (resp.\ right) of the spine of~$\Gamma$. Also, since~$\Gamma$ is an upward planar drawing, the edges incoming into~$v$ (outgoing from~$v$), lie below (resp.\ above) the horizontal line through~$v$. Hence,  the outgoing left, outgoing right, incoming right, and incoming left edges are incident into~$v$ in the second, first, fourth, and third quadrant, respectively. Thus,~$v$ is 4-modal. If~$v$ is a source (a sink), we analogously have that the outgoing left and outgoing right edges (resp.\ the incoming right and incoming left edges) incident to~$v$ lie in the second and first (resp.\ fourth and third) quadrant, respectively, while the large angle at~$v$ occupies the  third and fourth (resp.\ first and second) quadrant, hence~$v$ is 4-modal.
\end{proof}

\begin{figure}[tb!]
    \centering
    \begin{subfigure}{.35\textwidth}\centering
\includegraphics[page=1,width=\textwidth]{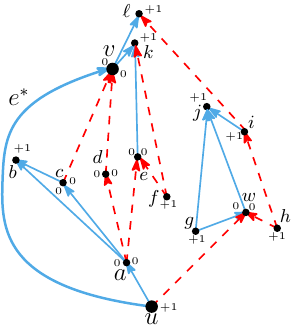}
    \subcaption{\label{fig:spq-tree-graph}}
    \end{subfigure}
    \hfill
    \begin{subfigure}{.64\textwidth}\centering
\includegraphics[page=2,height=.67\textwidth]{figures/Preliminaries.pdf}
    \subcaption{\label{fig:spq-tree-ecomposition}}
    \end{subfigure}
\caption{\label{fig:spq-tree}(a) An upward planar drawing~$\Gamma$ of a biconnected partitioned directed partial~$2$-tree $G$. The labeling~$\lambda_\Gamma$ of the angles determined by~$\Gamma$ is shown; the missing labels are equal to~$-1$.
    (b) The SPQ-tree of~$G$ rooted at the Q-node corresponding to the edge~$e^*$ of~$G$.}
\end{figure}

\subparagraph*{Partial 2-trees.} The class of (undirected) \emph{partial~$2$-trees} can be defined in several equivalent ways. Namely, a graph is a partial~$2$-tree if and only if:
\begin{itemize}
\item it has treewidth at most two;
\item it excludes~$K_4$ as a minor; or 
\item it is a subgraph of a~$2$-tree, which is a graph that can be obtained starting from an edge and repeatedly inserting a vertex of degree two adjacent to two adjacent vertices. 
\end{itemize}
Notably, the class of partial~$2$-trees includes the \emph{series-parallel graphs}, see, e.g., \cite{DBLP:journals/ijcga/BertolazziCBTT94,DBLP:journals/dcg/Biedl11,DBLP:journals/dmtcs/Frati10}.

Let~$G$ be a biconnected partial~$2$-tree and let~$e^*$ be an edge of~$G$ with end-vertices~$u^*$ and~$v^*$; refer to \cref{fig:spq-tree}. The \emph{SPQ-tree~$T$ of~$G$ with respect to~$e^*$} is a rooted tree that describes a recursive decomposition of~$G$ into smaller partial~$2$-trees; it is a specialization of the well-known \emph{SPQR-tree}, which is defined for general biconnected planar graphs~\cite{dt-opl-96,gm-lti-00}.

The root of~$T$ is a Q-node~$\rho^*$ associated with the entire graph~$G$ and has a single child~$\sigma^*$. Define~$G-e^*$ as the \emph{pertinent graph} of~$\sigma^*$, and let~$u^*$ and~$v^*$ be the \emph{poles} of both~$\sigma^*$ and~$\rho^*$. The remainder of the definition of~$T$ proceeds recursively as follows. Suppose we are given a quadruple~$\langle \mu, u, v, G_{\mu} \rangle$, where~$\mu$ is a node of~$T$ with poles~$u$ and~$v$, and~$G_{\mu}$ is its pertinent graph. Initially, this quadruple is~$\langle \sigma^*, u^*, v^*, G - e^* \rangle$. Three cases can occur:

\begin{itemize}
\item If~$G_{\mu}$ is a single edge~$(u,v)$, then~$\mu$ is a Q-node representing that edge;~$\mu$ is a leaf of~$T$.

\item If~$G_{\mu}$ is not biconnected, then~$\mu$ is an S-node. Let~$w$ be a cut-vertex of~$G_{\mu}$. Removing~$w$ splits~$G_{\mu}$ into two connected components: one,~$G^u_{\mu}$, containing~$u$, and the other,~$G^v_{\mu}$, containing~$v$.  
Then~$\mu$ has two children~$\nu_1$ and~$\nu_2$ in~$T$. The pertinent graph~$G_{\nu_1}$ (resp.~$G_{\nu_2}$) is the subgraph of~$G_{\mu}$ induced by~$\{w\} \cup V(G^u_{\mu})$ (resp. by~$\{w\} \cup V(G^v_{\mu})$). The poles of~$\nu_1$ are~$u$ and~$w$, and those of~$\nu_2$ are~$w$ and~$v$. The construction of~$T$ recurses on~$\langle \nu_1, u, w, G_{\nu_1} \rangle$ and on~$\langle \nu_2, w, v, G_{\nu_2} \rangle$.

\item If~$G_{\mu}$ is biconnected, then~$\mu$ is a P-node.
\begin{itemize}
	\item If~$(u,v)$ is not an edge of~$G_{\mu}$, then removing~$u$ and~$v$ splits~$G_{\mu}$ into~$k$ connected components~$G^1_{\mu}, \dots, G^k_{\mu}$, with~$k\geq 2$. Then~$\mu$ has~$k$ children~$\nu_1, \dots, \nu_k$, where~$G_{\nu_i}$ is the subgraph of~$G_{\mu}$ induced by~$\{u,v\}\cup V(G^i_{\mu})$; the poles of~$\nu_i$ are~$u$ and~$v$.
	\item If~$(u,v)$ is an edge of~$G_{\mu}$, then removing~$u$ and~$v$ leaves~$k-1$ components~$G^1_{\mu}, \dots, G^{k-1}_{\mu}$, with~$k\geq 2$. In this case,~$\mu$ has~$k$ children~$\nu_1, \dots, \nu_k$: for~$i=1,\dots,k-1$, the pertinent graph~$G_{\nu_i}$ is the subgraph of~$G_{\mu}$ induced by~$\{u,v\}\cup V(G^i_{\mu})$, excluding the edge~$(u,v)$, while~$G_{\nu_k}$ is the edge~$(u,v)$. Again, all nodes~$\nu_i$ have poles~$u$ and~$v$.
\end{itemize}
The construction of $T$ recurses on each quadruple~$\langle \nu_i, u, v, G_{\nu_i} \rangle$.
\end{itemize}

Observe that every S-node has two children, which may themselves be S-nodes. The SPQ-tree of~$G$ is in general not unique: Different choices for the cut-vertex~$w$ of the pertinent graph of an S-node and different choices for the reference edge might result in different SPQ-trees. In this paper, we assume that the choice of the cut-vertex~$w$ for each S-node of a rooted SPQ-tree is performed arbitrarily. On the other hand, the choice of the reference edge which serves as the root of the SPQ-tree will be done in all possible ways, as the reference edge will be forced to be incident to the outer face. If~$G$ has~$n$ vertices, then its SPQ-tree has~$O(n)$ nodes and can be computed in~$O(n)$ time~\cite{dt-opl-96}.


A \emph{directed partial~$2$-tree} is a digraph whose underlying graph is a partial~$2$-tree, where the underlying graph is the undirected graph obtained by ignoring the edge directions. An SPQ-tree of a biconnected directed partial~$2$-tree~$G$ is an SPQ-tree of its underlying graph, although the edges of the pertinent graph of each node are oriented as in~$G$.

\section{Characterization for Upward Embeddings} \label{se:characterization}


In this section we characterize the upward embeddings of a partitioned digraph that allow for the construction of an upward book embedding. We first present our characterization for plane~$st$-graphs, and then extend it to general plane digraphs.
If~$G$ is a plane~$st$-graph with planar embedding~$\mathcal E$, there is a unique upward-consistent angle-assignment~$\lambda$ which turns~$\mathcal E$ into an upward embedding~$(\mathcal E,\lambda)$. Indeed, vertices different from~$s$ and~$t$ do not have any incident large angle, and the large angles at~$s$ and~$t$ are necessarily those in the outer face of~$\mathcal E$. Thus, for a planar~$st$-graph, we can avoid talking about angle-assignments and upward embedding, and just consider planar embeddings.



\begin{figure}
    \centering
    \begin{subfigure}[t]{.49\textwidth}
    \centering
\includegraphics[page=5]{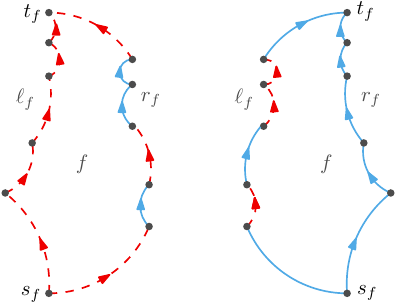}
    \subcaption{}
    \label{fig:impossible-st}
    \end{subfigure}
    \hfil
    \begin{subfigure}[t]{.49\textwidth}
    \centering
\includegraphics[page=4]{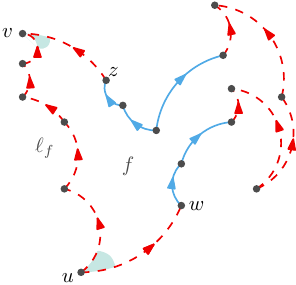}
    \subcaption{}
    \label{fig:impossible-dag}
    \end{subfigure}
\caption{Impossible faces in (a) a plane~$st$-graph and (b) an upward plane digraph.}
\end{figure}

Let~$G=(V,L \cup R)$ be a partitioned plane~$st$-graph, let~$\mathcal E$ be the planar embedding of~$G$, and let~$f$ be a face of~$\mathcal E$. We say that~$f$ is \emph{impossible} if the left path~$\ell_f$ of~$f$ only consists of edges in~$R$ and~the right path~$r_f$ of~$f$ is not a single edge, or if~$r_f$ only consists of edges in~$L$ and~$\ell_f$ is not a single edge; see \cref{fig:impossible-st}. We have the following.


\begin{theorem}\label{th:st-characterization}
Let~$G=(V,L \cup R)$ be a partitioned plane~$st$-graph and let~$\mathcal E$ be the planar embedding of~$G$. Then~$G$ admits an upward book embedding respecting~$\mathcal E$ if and only if~$\mathcal E$ is 4-modal and no face of~$\mathcal E$  is impossible. 
\end{theorem}

\begin{proof}
We first prove the necessity. \cref{pr:4-modal} ensures the necessity of the 4-modality of~$\mathcal E$. Next, consider an upward book embedding~$\Gamma$ respecting~$\mathcal E$ and suppose, for a contradiction, that there exists a face~$f$ of~$\mathcal E$ such that~$\ell_f$ only consists of edges in~$R$ and~$r_f$ contains an internal vertex~$v$. Since~$\ell_f$ only consists of edges in~$R$, its representation in~$\Gamma$ lies entirely to the right of the spine of~$\Gamma$, except at its vertices. Since~$\Gamma$ is an upward planar drawing respecting~$\mathcal E$, we have that~$v$ has to lie in the strip~$S$ delimited by the horizontal lines through the source and the sink of~$f$. However, this is not possible since~$v$ has to lie to the right of~$\ell_f$ (given that~$\ell_f$ and~$r_f$ are respectively the left and right path of~$f$) and since no point of the spine lies in the strip~$S$ and to the right of the curve representing~$\ell_f$ in~$\Gamma$. A contradiction can be achieved analogously if~$\mathcal E$ contains a face such that~$r_f$ only consists of edges in~$L$ and~$\ell_f$ contains an internal vertex.

We now prove the sufficiency. Every plane~$st$-graph can be constructed starting from its leftmost path by repeatedly adding the right path of an internal face whose left path already belongs to the graph, see, e.g., \cite{AngeliniLBF17,DBLP:journals/tcs/BattistaT88,DBLP:journals/tcs/FratiGW14,Mel}; more precisely, what needs to be added are the internal vertices and all the edges of the right path of the face. We use this in order to construct an upward book embedding~$\Gamma$ of~$G$ respecting~$\mathcal E$. We start by drawing in~$\Gamma$ the leftmost path of~$G$ so that edges in~$L$ are to the left of the spine and edges in~$R$ are to the right of the spine. When we draw the right path~$r_f$ of a face~$f$ whose left path~$\ell_f$ already belongs to the subgraph of~$G$ drawn in~$\Gamma$, we distinguish two cases.

First, if~$r_f$ is a single edge~$e$, then, since~$G$ is simple, we have that~$\ell_f$ contains an internal vertex. If~$e \in L$, we have that~$r_f$ consists only of edges in~$L$ and~$\ell_f$ is not just a single edge, i.e.,~$f$ is impossible.  Thus, it follows that~$e \in R$.
Then we just draw~$e$ as a semi-circle to the right of the spine (and hence to the right of the curve representing~$\ell_f$ in~$\Gamma$).

\begin{figure}
    \centering
    \begin{subfigure}{.32\textwidth}\centering
\includegraphics[page=1]{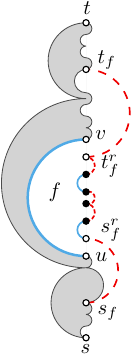}
    \subcaption{}
    \label{fig:rfNotASingleEdge}
    \end{subfigure}
    \begin{subfigure}{.32\textwidth}\centering    
    \includegraphics[page=2]{figures/characterization.pdf}
    \subcaption{}
    \label{fig:sf_srf_is_left}
    \end{subfigure}
    \begin{subfigure}{.32\textwidth}\centering
    \includegraphics[page=3]{figures/characterization.pdf}
\subcaption{}
\label{fig:sf_srf_and_tr_ft_f_are_left}
    \end{subfigure}
    \caption{Illustrations for the proof of \cref{th:st-characterization}. (a)~$(s_f,s^{r}_f),(t^{r}_f,t_f) \in R$. (b)~$(s_f,s^{r}_f) \in L$,~$(t^{r}_f,t_f) \in R$. (c)~$(s_f,s^{r}_f), (t^{r}_f,t_f) \in L$.}
    \label{fig:draw-non-impossible-face}
\end{figure}

Second, suppose that~$r_f$ is not a single edge; refer to \cref{fig:draw-non-impossible-face}. This implies that~$\ell_f$ is not entirely composed of right edges, as otherwise~$f$ would be an impossible face. Hence, let~$(u,v)$ be a left edge of~$\ell_f$. Also, let~$s_f$ and~$t_f$ be the source and the sink of the cycle delimiting~$f$, respectively. Furthermore, let~$s^{r}_f$ and~$t^{r}_f$ be the neighbors of~$s_f$ and~$t_f$, respectively, in~$r_f$. Analogously, let~$s^{\ell}_f$ and~$t^{\ell}_f$ be the neighbors of~$s_f$ and~$t_f$, respectively, in~$\ell_f$. We distinguish four cases.
\begin{itemize}
    \item Suppose first that~$(s_f,s^{r}_f)$ and~$(t^{r}_f,t_f)$ are both right edges; see \cref{fig:rfNotASingleEdge}. Then we embed all the internal vertices of~$r_f$ after~$u$ and before~$v$ on the spine, in the order in which they appear along~$r_f$. The semi-circle representing~$(s_f,s^{r}_f)$ is to the right of the portion of the curve representing~$\ell_f$ between~$s_f$ and the horizontal line through~$s^{r}_f$, in particular it is to the right of the semi-circle representing~$(u,v)$ since~$(u,v)$ is a left edge. Hence,~$(s_f,s^{r}_f)$ does not cause crossings in~$\Gamma$. Analogously,~$(t^{r}_f,t_f)$ does not cause crossings in~$\Gamma$. Finally, the curve representing the directed subpath of~$r_f$ between~$s^{r}_f$ and~$t^{r}_f$ does not cause crossings in~$\Gamma$, since it is {\em protected} by the semi-circle representing the edge~$(u,v)$ to its left. That is, the semi-circle representing~$(u,v)$ is to the left of (and does not cross) the subpath of~$r_f$ between~$s^{r}_f$ and~$t^{r}_f$, since~$(u,v)$ is a left edge, with~$u$ below~$s^{r}_f$ and~$v$ above~$t^{r}_f$. Also, the subpath of~$r_f$ between~$s^{r}_f$ and~$t^{r}_f$ does not cross any edge different from~$(u,v)$ in~$\Gamma$, as the intersection, if any, of such an edge with the strip delimited by the horizontal lines through~$u$ and~$v$, is to the left of~$(u,v)$, given that~$(u,v)$ is on the rightmost path of the graph represented in~$\Gamma$ before drawing~$r_f$. 
    \item Suppose next that~$(s_f,s^{r}_f)$ is a left edge and~$(t^{r}_f,t_f)$ is a right edge; see \cref{fig:sf_srf_is_left}. By 4-modality, the edge~$(s_f,s^{\ell}_f)$ is a left edge.  Then we embed all the internal vertices of~$r_f$ after~$s_f$ and before~$s^{\ell}_f$ on the spine, in the order in which they appear along~$r_f$. The proof that this does not cause crossings is analogous to the previous case.
    \item The case in which~$(s_f,s^{r}_f)$ is a right edge and~$(t^{r}_f,t_f)$ is a left edge can be discussed symmetrically to the previous case. 
    \item Finally, suppose that~$(s_f,s^{r}_f)$ and~$(t^{r}_f,t_f)$ are both left edges. If~$r_f$ is entirely composed of left edges, then~$\ell_f$ is a single edge, as otherwise~$f$ would be an impossible face. By 4-modality,~$\ell_f$ is a left edge. Therefore, we can  embed all the internal vertices of~$r_f$ after~$s_f$ and before~$t_f$ on the spine,  in the order in which they appear along~$r_f$. 
    
    Otherwise,~$r_f$ contains at least one right edge, call it~$(w,z)$; see \cref{fig:sf_srf_and_tr_ft_f_are_left}. We embed all the vertices of the subpath of~$r_f$ between~$s^{r}_f$ and~$w$ after~$s_f$ and before~$s^{\ell}_f$ on the spine, in the order in which they appear along~$r_f$; note that~$(s_f,s^\ell_f)$ is a left edge, since~$(s_f,s^{r}_f)$ is a left edge and by the 4-modality of~$\mathcal E$. Also, we embed all the vertices of the subpath of~$r_f$ between~$z$ and~$t^r_f$ after~$t^{\ell}_f$ and before~$t_f$ on the spine, in the order in which they appear along~$r_f$; again note that~$(t^\ell_f,t_f)$ is a left edge. The curves representing such subpaths of~$r_f$ do not cause crossings, as they are protected by the semi-circles representing the edges~$(s_f,s^\ell_f)$ and~$(t^\ell_f,t_f)$ to their left. Finally, the semi-circle representing~$(w,z)$, which lies to the right of the spine, does not cause crossings, since it lies to the right of the semi-circles representing~$(s_f,s^\ell_f)$ and~$(t^\ell_f,t_f)$, since these are left edges, and lies to the right of the curve representing the subpath of~$\ell_f$ between~$s^\ell_f$ and~$t^\ell_f$, since the vertices of this subpath all come after~$w$ and before~$z$ on the spine.
\end{itemize}
This concludes the proof of the characterization.    
\end{proof}

For plane digraphs that can have multiple sources and sinks, we generalize the notion of impossible face as follows.  Let~$G=(V,L \cup R)$ be a partitioned upward plane digraph, let~$(\mathcal E,\lambda)$ be the upward embedding of~$G$, and let~$f$ be a face of~$\mathcal E$.
Let~$\mathcal L_f$ (resp.\ $\mathcal R_f$) be the set of maximal directed paths in the boundary of~$f$ that consist of edges that have~$f$ to their right (resp.\ to their left).  Note that, if~$G$ is not biconnected, paths from~$\mathcal L_f$ and paths from~$\mathcal R_f$ are not necessarily disjoint.
The face~$f$ is \emph{impossible} if it satisfies one of the following two conditions (see \cref{fig:impossible-dag}):
\begin{enumerate}[label={\bf(\roman*)}]
\item~$\mathcal L_f$ contains a path~$\ell_f$ with the following properties. First,~$\ell_f$ consists of edges in~$R$. Second, the rest of the boundary of~$f$ is not a single edge. Third, let~$u$ and~$v$ be the extremes of~$\ell_f$, let~$e^u_f$ and~$e^v_f$ be the edges of~$\ell_f$ incident to~$u$ and~$v$, respectively; then the angle in~$f$ incident to~$u$ and to the right of~$e^u_f$ and the angle in~$f$ incident to~$v$ and to the right of~$e^v_f$ are small.
\item~$\mathcal R_f$ contains a path~$r_f$ with the following properties. First,~$r_f$ consists of edges in~$L$. Second, the rest of the boundary of~$f$ is not a single edge. Third, let~$u$ and~$v$ be the extremes of~$r_f$, let~$e^u_f$ and~$e^v_f$ be the edges of~$r_f$ incident to~$u$ and~$v$, respectively; then the angle in~$f$ incident to~$u$ and to the left of~$e^u_f$ and the angle in~$f$ incident to~$v$ and to the left of~$e^v_f$ are small.
\end{enumerate}

Note that, if~$G$ is a partitioned plane~$st$-graph, then the new definition of impossible face coincides with the previous one. We now prove our general characterization. A \emph{good embedding} of a partitioned upward plane digraph is an upward embedding~$(\mathcal E, \lambda)$ that is 4-modal and that is such that no face of~$\mathcal E$  is impossible.

\begin{theorem}
    \label{th:characterization}
Let~$G=(V,L \cup R)$ be a partitioned upward plane digraph and let~$(\mathcal E, \lambda)$ be the upward embedding of~$G$. Then~$G$ admits an upward book embedding respecting~$(\mathcal E, \lambda)$ if and only if~$(\mathcal E, \lambda)$ is a good embedding. 
\end{theorem}

\begin{proof}
We first assume that~$G$ is connected. We will get rid of this assumption later.

We start by proving the necessity. \cref{pr:4-modal} ensures the necessity of the 4-modality of~$\mathcal E$. Next, consider an upward book embedding~$\Gamma$ of~$G$ respecting~$(\mathcal E, \lambda)$ and suppose, for a contradiction, that a  face~$f$ of~$\mathcal E$  is impossible. After possibly horizontally mirroring the upward book embedding and swapping the left edges with the right edges, we may assume without loss of generality that~$\mathcal L_f$ contains a  path~$\ell_f$ with the following properties. First,~$\ell_f$ consists of edges in~$R$. Second, the rest of the boundary of~$f$, say~$p_f$, is not a single edge. Third, let~$u$ and~$v$ be the source and sink of~$\ell_f$, respectively, let~$e^u_f$ and~$e^v_f$ be the edges of~$\ell_f$ incident to~$u$ and~$v$, respectively; then the angle~$\alpha_u$ in~$f$ incident to~$u$ and to the right of~$e^u_f$ and the angle~$\alpha_v$ in~$f$ incident to~$v$ and to the right of~$e^v_f$ are small. Let~$w$ and~$z$ be the vertices adjacent to~$u$ and~$v$, respectively, such that the edge~$(u,w)$ follows the edge~$e^u_f$ in clockwise order around~$u$ and the edge~$(z,v)$ follows the edge~$e^v_f$ in counter-clockwise order around~$v$. Note that~$w=z$ might happen, however~$w\neq v$ and~$z\neq u$, given that~$p_f$ is not a single edge. Since~$\ell_f$ only consists of edges in~$R$, its representation in~$\Gamma$ lies entirely to the right of the spine of~$\Gamma$, except at its vertices; also, by the 4-modality of~$\mathcal E$, the edges~$(u,w)$ and~$(z,v)$ are in~$R$. Since~$\Gamma$ is an upward planar drawing, the order of the vertices of~$\ell_f$ along the spine is the same as their order in~$\ell_f$, with~$u$ below~$v$; also,~$w$ lies above~$u$ and~$z$ below~$v$. If any of~$w$ and~$z$ lies in the strip~$S_{uv}$ delimited by the horizontal lines through~$u$ and~$v$, then a contradiction can be reached as in the proof of \cref{th:st-characterization}, given that, because of the upward embedding~$(\mathcal E,\lambda)$ which forces~$\alpha_u$ and~$\alpha_v$ to be small, the edges~$(u,w)$ and~$(v,z)$ have to be to the right of~$e^u_f$ and~$e^v_f$, respectively. However, no point of the spine lies to the right of~$\ell_f$ and in the strip~$S_{uv}$. It follows that~$w$ has to lie above~$v$ and~$z$ below~$u$. This, however, implies that the edges~$(u,w)$ and~$(z,v)$ cross each other, given that they are both drawn to the right of the spine. This contradiction completes the proof of necessity.

We next prove the sufficiency. In order to do that, we show that it is possible to augment~$G$ and its upward embedding~$(\mathcal E,\lambda)$, respectively, to a partitioned~plane~$st$-graph and to a good embedding of it. Then \cref{th:st-characterization} implies that the augmented graph admits an upward book embedding respecting its upward embedding, from which one can obtain an upward book embedding of~$G$ respecting~$(\mathcal E,\lambda)$ by ignoring the vertices and edges added for the augmentation. The augmentation consists of two steps. In {\bf Step 1}, we augment~$G$ and~$(\mathcal E,\lambda)$ so that the outer face is an~$st$-face. In {\bf Step 2}, we adopt a modified version of the procedure described by Bertolazzi et al. \cite{DBLP:journals/siamcomp/BertolazziBMT98} to augment any upward plane digraph to a plane~$st$-graph, while maintaining the upward embedding.

\begin{figure}
    \centering
    \begin{subfigure}{.32\textwidth}\centering
    \includegraphics[page=6,width=.9\textwidth]{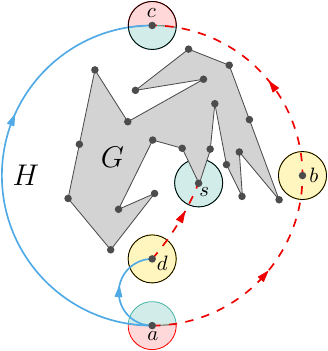}
    \subcaption{\label{fig:augmentation-outer-st}}
    \end{subfigure}
    \hfil
    \begin{subfigure}{.32\textwidth}\centering
\includegraphics[page=9,width=1\textwidth]{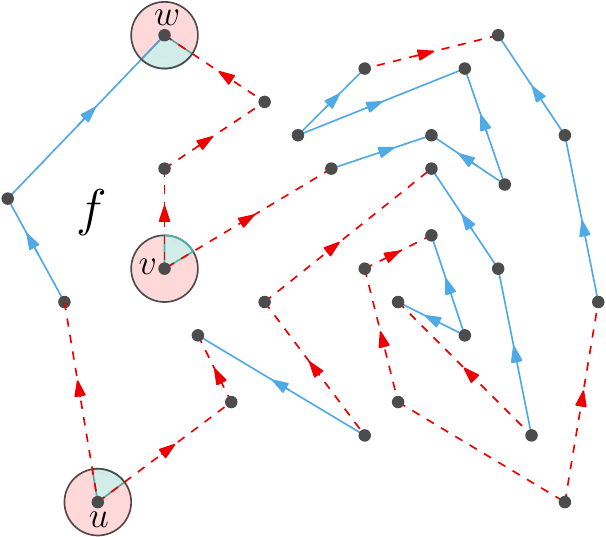}
    \subcaption{\label{fig:bertolazzi-before}}
    \end{subfigure}
    \hfil
    \begin{subfigure}{.32\textwidth}\centering
\includegraphics[page=10,width=1\textwidth]{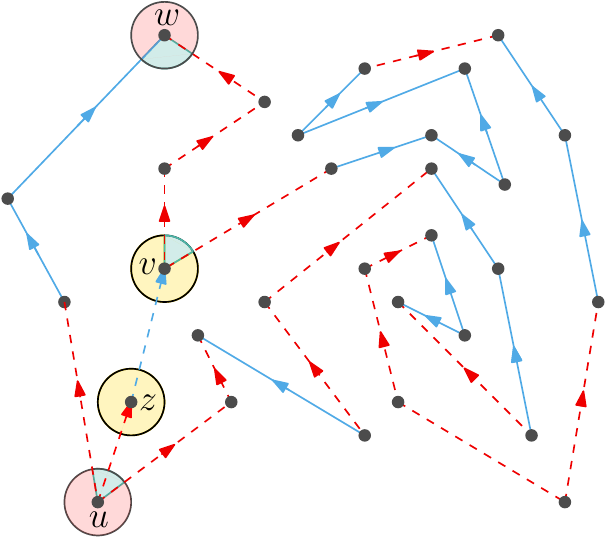}
    \subcaption{\label{fig:bertolazzi-after}}
    \end{subfigure}
    \caption{Illustrations for the proof of \cref{th:characterization}. (a) Upward embedding~$(\mathcal E_H,\lambda_H)$ of the partitioned upward plane digraph~$H$ obtained from~$G$ in {\bf Step 1}. (b)-(c) Augmentation of a face~$f$ with a large angle in an upward embedding~$(\mathcal E,\lambda)$ to obtain the upward embedding~$(\mathcal E',\lambda')$  in~{\bf Step~2}.
    }
\end{figure}

\subparagraph*{Step 1.} We first augment~$G=(V,L \cup R)$ and its upward embedding~$(\mathcal E,\lambda)$ to a partitioned upward plane digraph~$H=(V\cup \{a,b,c,d\},L\cup \{(a,c),(a,d)\},R\cup \{(a,b),(b,c),(d,s)\}$, where~$s$ is any source of~$G$ that is incident to the outer face of~$\mathcal E$ and that has a large angle in the outer face (refer to \cref{fig:augmentation-outer-st}); since~$\lambda$ is upward-consistent, such a source~$s$ exists, as otherwise every switch angle at a vertex $v$ delimited by two edges outgoing from $v$ would be small, and the sum of the values assigned by~$\lambda$ to the angles incident to the outer face could not be~$+2$. We also augment~$(\mathcal E, \lambda)$ to an upward embedding~$(\mathcal E_H, \lambda_H)$ of~$H$ as follows. 

First, the planar embedding~$\mathcal E_H$ of~$H$ has the following properties: 
\begin{enumerate}[label={\bf(\roman*)}]
\item the outer face of~$\mathcal E_H$ is delimited by the cycle~$(a,b,c)$; 
\item the restriction of~$\mathcal E_H$ to~$G$ coincides with~$\mathcal E$; 
\item the clockwise order of the edges incident to~$a$ is~$(a,c),(a,d),(a,b)$; also, the edge~$(d,s)$ is incident to~$s$ in the outer face of~$\mathcal E$. 
\end{enumerate}

Second, the angle assignment~$\lambda_H$ has the following properties: 
\begin{enumerate}[label={\bf(\roman*)}]
\item~$\lambda_H$ assigns the same value as~$\lambda$ to every angle of~$\mathcal E_H$ that is also an angle of~$\mathcal E$; 
\item~$\lambda_H$ assigns~$0$ to all angles incident to~$b$ and~$d$; 
\item~$\lambda_H$ assigns~$1$ to the angles at~$a$ and~$c$ incident to the outer face of~$\mathcal E_H$ and~$-1$ to all other angles at~$a$ and~$c$; and 
\item~$\lambda_H$ assigns~$0$ to the angles at~$s$ incident to the edge~$(d,s)$.  
\end{enumerate}

We first show that~$(\mathcal E_H,\lambda_H)$ is an upward embedding of~$H$. To this end, we start by observing that~$(\mathcal E_H, \lambda_H)$ is 4-modal and in particular it is bimodal. Then, by \cref{th:upward-conditions}, it remains to prove that~$\lambda_H$ is upward-consistent. 

We prove that~$\lambda_H$ satisfies \cref{cond:flat-assignment}. This condition is verified for every angle that is also an angle of~$\mathcal E$, since~$\lambda$ is upward-consistent and~$\lambda_H$ assigns the same value as~$\lambda$ to every angle of~$\mathcal E_H$ that is also an angle of~$\mathcal E$. The angles at~$a$ and~$c$ are all assigned with a value different from~$0$ and indeed none of them is flat. Each angle at~$b$ or~$d$ is assigned the value~$0$ and indeed it is flat. Finally, the angles at~$s$ incident to~$(d,s)$ are both assigned the value~$0$ and each of them is flat, since~$s$ is a source in~$G$ and~$(d,s)$ is incoming into~$s$.

We prove that~$\lambda_H$ satisfies \cref{cond:vertex-assignment}. This condition is verified for every vertex not in~$\{a,b,c,d,s\}$, since~$\lambda$ is upward-consistent and~$\lambda_H$ assigns the same value as~$\lambda$ to every angle of~$\mathcal E_H$ that is also an angle of~$\mathcal E$. The two angles at~$c$ are assigned a~$-1$ and a~$+1$, hence their sum is~$0$, which equals~$2-\deg(c)$. The two angles at each of~$b$ and~$d$ are assigned value~$0$, which is equal to~$2$ minus their degree. The two angles at~$a$ incident to internal faces of~$\mathcal E_H$ are both assigned value~$-1$, whereas the angle at~$a$ incident to the outer face of~$\mathcal E_H$ is assigned value~$+1$. Hence, the sum of these values is~$-1$, which is equal to~$2-\deg(a)$. Finally, the angles at~$s$ incident to~$(d,s)$ are both assigned the value~$0$, and all the other angles are assigned the value~$-1$; the sum of these values is hence equal to~$2-\deg(s)$.  

Finally, we prove that~$\lambda_H$ satisfies \cref{cond:face-assignment}. This condition is verified for every face of~$\mathcal E_H$ that is also a face of~$\mathcal E$, since~$\lambda$ is upward-consistent and~$\lambda_H$ assigns the same value as~$\lambda$ to every angle of~$\mathcal E_H$ that is also an angle of~$\mathcal E$. Also, the angles incident to the outer face of~$\mathcal E_H$ at~$a$ and~$c$ are assigned the value~$+1$, and the angle incident to the outer face of~$\mathcal E_H$ at~$b$ is assigned with the value~$0$, hence the sum of these values is~$2$, as required by \cref{cond:face-assignment}. It remains to discuss the condition for the internal face~$f_{\mathrm{in}}$ of~$\mathcal E_H$ incident to the vertex~$d$. We distinguish three sets of angles incident to~$f_{\mathrm{in}}$. The set~$A_1$ contains the angles of~$f_{\mathrm{in}}$ that are also angles in the outer face~$f_o$ of~$\mathcal E$; by construction, the angles in~$A_1$ are assigned by~$\lambda_H$ the same value as~$\lambda$. Notice that~$A_1$ contains all the angles of~$\mathcal E$ incident to~$f_o$, except for the angle~$\sigma$ at~$s$, which by assumption is assigned value~$+1$ by~$\lambda$. The set~$A_2$ contains the five angles of~$f_{in}$ incident to~$b$,~$d$, and~$s$, each of which is assigned the value~$0$. Finally, the set~$A_3$ contains the three angles of~$f_{\mathrm{in}}$ incident to~$a$ and~$c$, each of which is assigned the value~$-1$. It follows that~$\sum_{a \in A_{\mathcal E_H}(f_{\mathrm{in}})} \lambda_H(a) = \big(\sum_{a \in A_1}\lambda(a)\big) - \lambda(\sigma) + 0\cdot|A_2| - 1 \cdot |A_3|   = 2 - 1 +0 - 3 = -2$.
This concludes the proof that~$(\mathcal E_H,\lambda_H)$ is an upward embedding of~$H$.

In order to conclude the discussion of Step 1, since~$(\mathcal E_H,\lambda_H)$ is 4-modal, it remains to prove that no face of~$\mathcal E_H$ is impossible. Every face of~$\mathcal E_H$ that is also a face of~$\mathcal E$ is not impossible, since~$(\mathcal E,\lambda)$ is a good embedding. By construction, the outer face of~$\mathcal E_H$ is an~$st$-face and it is not impossible since its left path is an edge in~$L$. Finally, consider the face~$f_{\mathrm{in}}$. We distinguish four types of maximal directed paths on the boundary of~$f_{\mathrm{in}}$. 
\begin{itemize}
\item First, the maximal directed paths that comprise the path~$(a,d,s)$ contain at least one left edge, namely~$(a,d)$, and at least one right edge, namely~$(d,s)$, hence they cannot make~$f_{\mathrm{in}}$ impossible. 
\item Second, the maximal directed path~$(a,c)$ has~$f_{\mathrm{in}}$ to its right and it consists of a left edge, hence it cannot make~$f_{\mathrm{in}}$ impossible. 
\item Similarly, the maximal directed path~$(a,b,c)$ has~$f_{\mathrm{in}}$ to its left and it consists of two right edges, hence it cannot make~$f_{\mathrm{in}}$ impossible.  
\item Finally, every remaining maximal directed path incident to~$f_{\mathrm{in}}$ is composed entirely of edges of~$G$. Then such a path does not make~$f_{\mathrm{in}}$ impossible, since the outer face~$f_o$ of~$\mathcal E$ is not impossible. 
\end{itemize}



\subparagraph*{Step 2.} After \textbf{Step 1}, we rename~$H$ and~$(\mathcal E_H,\lambda_H)$ to~$G$ and~$(\mathcal E, \lambda)$, respectively, where the outer face of~$G$ in~$\mathcal E$ is an~$st$-face. We now show that~$G$ and~$(\mathcal E, \lambda)$ can be augmented, by adding vertices and edges, to a partitioned plane~$st$-graph~$G^+$ with a good embedding~$(\mathcal E^+, \lambda^+)$. This concludes the proof (for connected digraphs), as then, by~\cref{th:st-characterization}, we have that~$G^+$ admits an upward book embedding respecting~$(\mathcal E^+, \lambda^+)$, and the restriction of such an upward book embedding to~$G$ is an upward book embedding of~$G$  respecting~$(\mathcal E, \lambda)$. We prove that~$G$ and~$(\mathcal E, \lambda)$ have the claimed augmentation by induction on the number~$S$ of switches of~$G$. 

For the base case, we have~$S=2$, hence~$G$ is a partitioned plane~$st$-graph and thus it suffices to set~$G^+=G$ and~$(\mathcal E^+, \lambda^+)=(\mathcal E, \lambda)$, since~$(\mathcal E, \lambda)$ is a good embedding,~by~assumption.

If~$S > 2$, then there are at least three large angles in~$(\mathcal E, \lambda)$.  Since the outer face is an~$st$-face, it has exactly two large angles and hence there exists an internal face~$f$ of~$\mathcal E$ that contains a large angle.  Then Bertolazzi et al.~\cite{DBLP:journals/algorithmica/BertolazziBLM94} showed that~$f$ contains three switch angles that are consecutive in the clockwise order of the switch angles around~$f$ and that are small, small, and large, respectively (see \cref{fig:bertolazzi-before}). Let~$u$,~$w$, and~$v$ be the vertices the three angles are incident to, respectively, and note that~$v$ is a switch of~$G$.
We assume that~$v$ is a source, the case in which it is a sink is analogous.  

Let~$\tilde G$ be the upward plane digraph obtained by adding the edge~$(u,v)$ to~$G$; also, let~$\tilde{\mathcal E}$ be the planar embedding obtained from~$\mathcal E$ by adding the edge~$(u,v)$ inside~$f$; finally, let~$\tilde \lambda$ be the upward-consistent angle assignment obtained from~$\lambda$ by defining the two angles incident to~$(u,v)$ at~$u$ to be small and the two angles incident to~$(u,v)$ at~$v$ to be flat. Bertolazzi et al.~\cite{DBLP:journals/algorithmica/BertolazziBLM94} showed that~$(\tilde{\mathcal E},\tilde \lambda)$ is indeed an upward embedding of~$\tilde G$. Note that~$\tilde G$ has one less switch than~$G$, as~$v$ is not a source in~$\tilde G$, which would allow induction to be applied. Unfortunately, adding~$(u,v)$ to either of the parts~$L$ or~$R$ of the edge set of~$G$ may result in an impossible face or in a violation of the 4-modality.  

To remedy this, we create a partitioned upward plane digraph~$G'$ with upward embedding~$(\mathcal E', \lambda')$ from~$\tilde G$ and~$(\tilde{\mathcal E},\tilde \lambda)$ by additionally subdividing the edge~$(u,v)$ with a new vertex~$z$ whose incident angles are both flat (see  \cref{fig:bertolazzi-after}); we assign the edge~$(u,z)$ to the same partition ($L$ or~$R$) as the edge that follows it in clockwise order around~$u$ and we assign the edge~$(z,v)$ to the other partition.  Clearly, that~$(\mathcal E', \lambda')$ is an upward embedding of~$G'$ still comes from the result by Bertolazzi et al.; moreover, the choice of the partition for~$(u,z)$ ensures that~$u$ remains 4-modal and, since~$(z,v)$ is the sole incoming edge at~$v$, the 4-modality of~$v$ is also preserved. It remains to prove that~$(\mathcal E',\lambda')$ has no impossible face.  Assume, for the sake of contradiction, that there exists an impossible face~$g$ in~$(\mathcal E',\lambda')$.  Since~$(\mathcal E,\lambda)$ contains no impossible face, it follows that~$g$ must be one of the two faces that are incident to the new vertex~$z$. Let~$p_g$ be a maximal directed path on the boundary of~$g$ that satisfies the conditions in the definition of impossible face. Note that~$p_g$ cannot contain the path~$(u,v,z)$, which contains both an edge in~$L$ and an edge in~$R$, whereas~$p_g$ is either composed entirely of edges in~$L$ or of edges in~$R$.  Hence,~$p_g$ is part of the boundary of~$f$.  Since the rest of the boundary of~$f$ contains at least as many vertices as the rest of the boundary of~$g$ and since the augmentation of~$(\mathcal E,\lambda)$ to~$(\mathcal E',\lambda')$ has only split a small angle (at~$u$) into two small angles and a large angle (at~$v$) into two flat angles, it follows that~$f$ is an impossible face, a contradiction.


Since~$G'$ has one less switch than~$G$, this concludes the induction and hence the proof for the case in which~$G$ is connected. 


\subparagraph*{The disconnected case.} We now discuss the case in which~$G$ is not connected. Let~$G_1,\dots,G_k$ be the connected components of~$G$, for some integer~$k\geq 2$, and let~$(\mathcal E_i,\lambda_i)$ be the upward embedding of~$G_i$ which is the restriction of~$(\mathcal E,\lambda)$ to~$G_i$, for~$i=1,\dots,k$. 

The proof of the necessity of the characterization uses two facts additional to the arguments presented for the connected case. On the one hand, the 4-modality of the upward embedding~$(\mathcal E,\lambda)$ boils down to the 4-modality of the upward embeddings~$(\mathcal E_i,\lambda_i)$, whose necessity we already proved. On the other hand, it might be the case that~$(\mathcal E_i,\lambda_i)$ does not contain any impossible face, for~$i=1,\dots,k$, and yet~$(\mathcal E,\lambda)$ does. This happens if and only if an internal face~$f$ of the planar embedding~$\mathcal E_i$ of a connected component~$G_i$ of~$G$ satisfies the following properties. First, the planar embedding~$\mathcal E$ places a connected component~$G_j$ of~$G$ with~$j\neq i$ in~$f$. Second,~$f$ is an~$st$-face such that its left boundary~$\ell_f$ is composed of edges in~$R$ and its right boundary is a single edge (also in~$R$), or its right boundary~$r_f$ is composed of edges in~$L$ and its left boundary is a single edge (also in~$L$). Note that~$f$ is not an impossible face in~$(\mathcal E_i,\lambda_i)$, but the face~$f_{\mathcal E}$ of~$\mathcal E$ corresponding to~$f$ is impossible in~$(\mathcal E,\lambda)$, as the rest (with respect to~$\ell_f$ or~$r_f$, respectively) of the boundary of~$f_{\mathcal E}$ is not a single edge, given that it comprises the boundary of the outer face of~$\mathcal E_j$. The characterization correctly handles this situation, since in any upward book embedding of~$G_i$ respecting~$(\mathcal E_i,\lambda_i)$, no part of the spine lies inside~$f$, hence the placement of~$G_j$ inside~$f$ demanded by~$(\mathcal E,\lambda)$ is not possible. The necessity of not having any impossible face follows.

The sufficiency of the characterization can be proved as follows. We start by constructing an upward book embedding~$\Gamma_i$ of~$G_i$ respecting~$(\mathcal E_i,\lambda_i)$, for~$i=1,\dots,k$, as described in the connected case. We now insert the upward book embeddings~$\Gamma_i$ into an (initially empty) upward book embedding~$\Gamma$ one by one. In particular,  we insert into~$\Gamma$ an upward book embedding~$\Gamma_i$ of a component~$G_i$ once every component that has to contain~$G_i$ in an internal face is already part of~$\Gamma$. When~$\Gamma_i$ is inserted in~$\Gamma$, the face~$f$ of the current embedding in which it needs to be inserted contains in its interior in~$\Gamma$ a portion of the spine, since the face of~$\mathcal E$ corresponding to~$f$ is not impossible. Then the vertices of~$G_i$ can be placed consecutively in a portion of the spine inside~$f$, thus inserting~$\Gamma_i$ into~$\Gamma$. Eventually, this results in an upward book embedding~$\Gamma$ of~$G$ respecting~$(\mathcal E,\lambda)$. This concludes the proof of the sufficiency and of the characterization.
\end{proof}

\section{Computational Complexity with Variable Embedding} \label{se:complexity}

In this section we study the complexity of testing whether a partitioned digraph admits an upward book embedding. We show that the problem is \NP-complete. Our hardness proof exploits the characterization of \cref{th:characterization} and the \NP-hardness of {\sc Upward Planarity Testing}.

\begin{theorem}\label{th:np-complete}
    It is \NP-complete to decide whether a partitioned planar digraph~$G=(V,L \cup R)$ admits an upward book embedding.
\end{theorem}

\begin{proof}
    Clearly, the problem is in \NP. In order to prove \NP-hardness, we give a reduction from {\sc Upward Planarity Testing}, which was proved to be \NP-hard by Garg and Tamassia~\cite{DBLP:journals/siamcomp/GargT01}. Given a planar digraph~$G=(V,E)$, we construct a partitioned digraph~$G'=(V',L \cup R)$ as follows. Subdivide each edge~$e\in E$ with a new vertex~$v_e$; these subdivision vertices, together with the vertices in~$V$, form~$V'$. The edge set~$L$ consists of the edges outgoing from vertices in~$V$ (and incoming into vertices in~$V'\setminus V$), and the edge set~$R$ consists of the remaining edges.  Clearly,~$G'$ can be constructed from~$G$ in polynomial time.  It remains to show that~$G'$ admits an upward book embedding if and only if~$G$ admits an upward planar drawing. 

    For the necessity, observe that an upward book embedding~$\Gamma'$ of~$G'$ is an upward planar drawing of~$G'$. Then an upward planar drawing~$\Gamma$ of~$G$ is obtained from~$\Gamma'$  by (i) placing each vertex of~$G$ in~$\Gamma$ as in~$\Gamma'$ and by (ii) drawing in~$\Gamma$ each edge~$e=(u,w)$ of~$G$ as the Jordan arc formed by the union of the drawings of~$(u,v_e)$ and of~$(v_e,w)$  in~$\Gamma'$.
    
    For the sufficiency, suppose that~$G$ admits an upward planar drawing~$\Gamma$. Then an upward planar drawing~$\Gamma'$ of~$G'$ can be constructed from~$\Gamma$ by placing each subdivision vertex~$v_e$ at any internal point of the Jordan arc representing the edge~$e$. Let~$(\mathcal E',\lambda')$ be the upward embedding corresponding to~$\Gamma'$. 
    We prove that~$(\mathcal E',\lambda')$ is 4-modal. Consider any vertex~$v$ of~$G'$. If~$v\notin V$, then~$v$ has one incoming and one outgoing edge in~$G'$, hence it is trivially 4-modal. If~$v\in V$ then, by construction, all the edges outgoing from~$v$, if any, are in~$L$ and are consecutive in the clockwise order of the edges incident to~$v$, due to the bimodality of the planar embedding of~$G$ corresponding to~$\Gamma$. Similarly, all the edges incoming into~$v$, if any, are in~$R$ and are consecutive in the clockwise order of the edges incident to~$v$. It follows that~$v$ is 4-modal.
    Finally,~$(\mathcal E',\lambda')$ has no impossible face, since by construction any maximal directed path  in the boundary of any face contains both an edge in~$L$ and an edge in~$R$.  Hence,~$(\mathcal E',\lambda')$  is a good embedding. By \Cref{th:characterization}, we have that~$G'$ admits an upward book embedding respecting~$(\mathcal E',\lambda')$.
\end{proof}

{\sc Upward Planarity Testing} is known to be W[1]-hard with respect to the treewidth~\cite{DBLP:conf/gd/JansenKKLMS23}. Also, the reduction shown in \cref{th:np-complete} constructs a graph which is a subdivision of the original instance of {\sc Upward Planarity Testing}. Since any two graphs, one of which is a subdivision of the other one, have the same treewidth, we get the following.

\begin{corollary}
It is W[1]-hard with respect to the treewidth to decide whether a partitioned digraph~$G=(V,L \cup R)$ admits an upward book embedding.
\end{corollary}

Furthermore, by subdividing the edges twice, rather than once, so that edges incident to vertices in~$V$ are in~$L$ and the other edges in~$R$, the reduction gives an instance in which one color induces a matching and the other one a forest of stars. This is in sharp contrast with the fact that the problem is solvable in linear time when both edge parts are matchings~\cite{Akitaya18}.


\section{Test for Graphs with a Fixed Planar Embedding} \label{se:fixed}

In this section we show how to exploit the characterization of~\cref{th:characterization} in order to prove that, for an~$n$-vertex partitioned plane digraph~$G$ with a given planar embedding~$\mathcal E$, it can be tested in~$O(n\log^3 n)$ time whether~$G$ admits an upward book embedding respecting~$\mathcal E$. 

We start by reviewing a tool for testing whether~$G$ admits an upward planar drawing~$\Gamma$ respecting~$\mathcal E$, assuming that~$G$ is connected; we will remove this assumption later. By~\cref{th:upward-conditions}, the bimodality of~$\mathcal E$ is a necessary condition for the existence of~$\Gamma$. Since the bimodality of~$\mathcal E$ can be easily tested in~$O(n)$ time, in the following we assume that~$\mathcal E$ is indeed bimodal. Then, again by~\cref{th:upward-conditions},  the existence of~$\Gamma$ is equivalent to the existence of an upward-consistent angle assignment~$\lambda$ for~$\mathcal E$. In order to test for the existence of~$\lambda$, Bertolazzi et al.~\cite{DBLP:journals/algorithmica/BertolazziBLM94} proposed the following strategy\footnote{Our description of the strategy differs slightly from the one by Bertolazzi et al.~\cite{DBLP:journals/algorithmica/BertolazziBLM94}, as they assume~$G$ to be triconnected (and thus, that each vertex is incident to a face at most once in~$\mathcal E$), while we do not.}. 

A \emph{flow network}~$\mathcal N$ is a directed graph such that each source is associated with a non-negative value, called \emph{supply}, each sink is associated with a non-negative value, called \emph{demand}, and each edge~$a$ is associated with a non-negative value~$c_a$, called \emph{capacity}. Vertices and edges of a flow network are usually called \emph{nodes} and \emph{arcs}, respectively. A \emph{flow} is an assignment of a value~$\phi_a$ to each arc~$a$ of~$\mathcal N$; the value~$\phi_a$ is called the \emph{flow assigned to}~$a$. A flow is \emph{feasible} if: (i) the flow assigned to each arc~$a$ of~$\mathcal N$ is at most its capacity:~$\phi_a\leq c_a$; (ii) the sum of the flows assigned to the arcs outgoing each source~$s$ of~$\mathcal N$ does not exceed the supply of~$s$; (iii) the sum of the flows assigned to the arcs incoming into each sink~$t$ of~$\mathcal N$ does not exceed the demand of~$t$; and (iv) the sum of the flows assigned to the arcs incoming into each non-switch node of~$\mathcal N$ is equal to the sum of the flows assigned to the arcs outgoing from the same node. The \emph{value} of a flow is the sum of the flows assigned to the arcs incoming into the sinks (or, equivalently, outgoing from the sources).

Starting from the plane digraph~$G$ with planar embedding~$\mathcal E$, one can construct a planar flow network~$\mathcal N$, as illustrated in \cref{fig:flow-network}, where:
\begin{itemize}
    \item for each switch~$v$ of~$G$, the network~$\mathcal N$ contains a source~$s_v$ that supplies a single unit of flow;
    \item for each face~$f$ of~$\mathcal E$, the network~$\mathcal N$ contains a sink~$t_f$ that demands a number of units of flow equal to~$n_f/2 -1$ or~$n_f/2 +1$, depending on whether~$f$ is an internal face or the outer face of~$\mathcal E$, respectively (where~$n_f$ is the number of switch angles incident to~$f$);
    \item for each angle~$\alpha$ in~$\mathcal E$ at a switch of~$G$, the network~$\mathcal N$ contains a node~$w_\alpha$; and
    \item the network~$\mathcal N$ contains an arc from each source~$s_v$ to each node~$w_\alpha$ such that the angle~$\alpha$ is incident to~$v$ in~$\mathcal E$ and an arc from each node~$w_\alpha$ to each sink~$t_f$ such that the angle~$\alpha$ is incident to~$f$ in~$\mathcal E$; all such arcs have a capacity of a single unit of flow. 
\end{itemize}

\begin{figure}[tb!]
    \begin{subfigure}
    {.49\textwidth}\centering
\includegraphics[page=1,width=\linewidth]{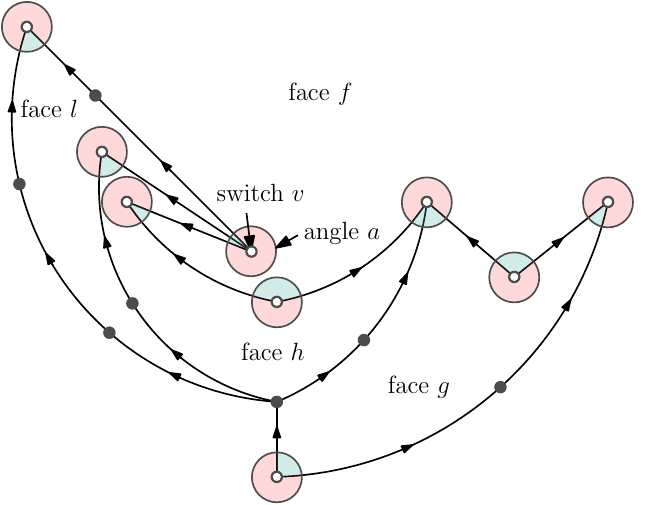}
    \subcaption{}
    \label{fig:instance}
    \end{subfigure}
    \hfil
    \begin{subfigure}
    {.49\textwidth}\centering    
    \includegraphics[page=2,width=\linewidth]{figures/network.pdf}
    \subcaption{}
    \label{fig:network}
    \end{subfigure}
    \caption{\label{fig:flow-network}(a)~An upward planar drawing~$\Gamma$ of an upward plane digraph~$G$; only large and small switch angles at switches of~$G$ are depicted (as  red and green sectors of discs centered at switch vertices, respectively). 
    (b)~The network~$\cal N$ constructed from~$G$ and its planar embedding; for convenience, the edges of~$G$, which are not part of~$\mathcal N$, are drawn as gray dashed curves. Arcs of~$\cal N$ traversed by the flow are red, whereas arcs with no flow are green.}
\end{figure}

It was proved in~\cite{DBLP:journals/algorithmica/BertolazziBLM94} that there exists an upward-consistent angle assignment~$\lambda$ for~$\mathcal E$ if and only if~$\mathcal N$ admits a feasible flow whose value is the sum~$d_T$ of the demands of the sinks in~$\mathcal N$ (or, equivalently, the sum of the supplies of the sources in~$\mathcal N$). More precisely, upward-consistent angle assignments and (integral) feasible flows with value~$d_T$ are in bijection\footnote{This statement assumes that the flow assigned to each arc is integer. This is not a loss of generality, since, in a network with integer supplies, demands, and capacities, a non-integer feasible flow with integer value can always be transformed into an integer feasible flow with the same value.}:
\begin{enumerate}
\item If an angle assignment~$\lambda$ is upward-consistent then one can get a feasible flow for~$\mathcal N$ with value~$d_T$ by assigning flow~$1$ to each arc incident to a node~$w_\alpha$ such that the angle assigned to~$\alpha$ by~$\lambda$ is large, and by assigning flow~$0$ to all other arcs.
\item If a feasible flow for~$\mathcal N$ has value~$d_T$, then one can get an upward-consistent angle assignment~$\lambda$ by assigning a large angle to each switch angle~$\alpha$ such that the arcs incident to node~$w_\alpha$ are assigned one unit of flow, and by assigning a small angle to all other switch angles.
\end{enumerate}



Testing whether~$G$ admits an upward planar drawing~$\Gamma$ respecting~$\mathcal E$ then becomes equivalent to testing whether~$\mathcal N$ admits a flow whose value is~$d_T$. An algorithm solving this problem in~$O(n\log^3 n)$ time is known \cite{DBLP:journals/siamcomp/BorradaileKMNW17}, which gives the running time of the best known upward planarity testing algorithm with fixed planar embedding.
    
In a nutshell, our idea is to modify~$\mathcal N$ so that there exists an upward-consistent angle assignment~$\lambda$ for~$\mathcal E$ such that~$(\mathcal E,\lambda)$ is a good embedding (i.e., such that~$G$ admits an upward book embedding respecting~$(\mathcal E,\lambda)$, by \cref{th:characterization}) if and only if~$\mathcal N$ admits a feasible flow whose value is~$d_{T}$. As in the result by Bertolazzi et al.~\cite{DBLP:journals/algorithmica/BertolazziBLM94}, the correspondence is actually stronger, as we show in the proof of \cref{th:planar-embedding} that, if $G$ is connected, the integral feasible flows with value~$d_T$ for the modified network are in bijection with the upward-consistent angle assignments~$\lambda$ for~$\mathcal E$ such that~$(\mathcal E,\lambda)$ is a good embedding.
We show an algorithm, called {\sc $\mathcal N$-modifier}, that performs a sequence of modifications to~$\mathcal N$. Along the way, {\sc $\mathcal N$-modifier} might stop and conclude that~$G$ admits no upward book embedding respecting~$\mathcal E$.

\smallskip
\noindent\textsc{The {\sc $\mathcal N$-modifier} algorithm.} We start by describing the modifications to~$\mathcal N$ that {\sc $\mathcal N$-modifier} performs in order to ensure 4-modality. Consider any vertex~$v$ of~$G$. 
\begin{itemize}
\item If~$v$ is not a switch, then its 4-modality does not depend on the angle assignment. Then {\sc $\mathcal N$-modifier} checks whether in~$\mathcal E$, in clockwise order around~$v$, we have all the outgoing left edges, then all the outgoing right edges, then all the incoming right edges, and finally all the incoming left edges, where one of the former two sets and/or one of the latter two sets might be empty. If the test is negative, {\sc $\mathcal N$-modifier} concludes that~$G$ admits no upward book embedding respecting~$\mathcal E$. 
Otherwise, the processing of~$v$ is concluded.
\item 
\begin{figure}[htb]
    \begin{subfigure}
    {.32\textwidth}\centering
\includegraphics[page=5,width=.8\linewidth]{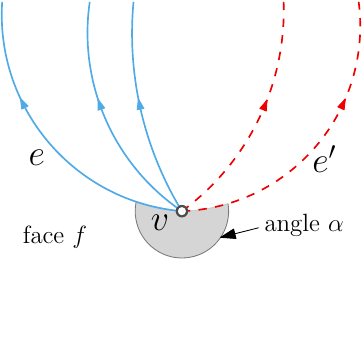}
    \subcaption{}
    \label{fig:4modality-vertex}
    \end{subfigure}
    \hfil
    \begin{subfigure}
    {.32\textwidth}\centering
\includegraphics[page=3,width=.8\linewidth]{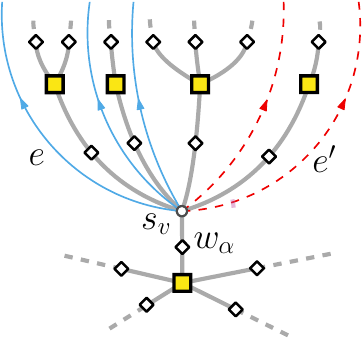}
    \subcaption{}
    \label{fig:4modality-before}
    \end{subfigure}
    \hfil
    \begin{subfigure}
    {.32\textwidth}\centering    
    \includegraphics[page=4,width=.8\linewidth]{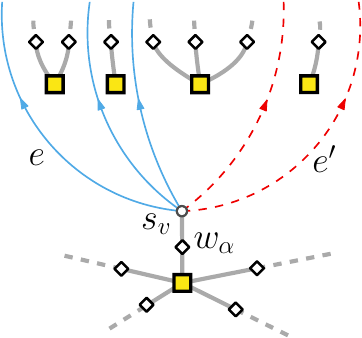}
    \subcaption{}
    \label{fig:4modality-after}
    \end{subfigure}
    \caption{\label{fig:flow-network-4-modality}Modification of~$\mathcal N$ to ensure~$4$-modality. (a) A vertex~$v$ that has both outgoing left edges and outgoing right edges.
    (b) The part of~$\mathcal N$ close to~$v$ (the edges of~$G$ do not belong to~$\mathcal N$, but they are shown to maintain a visual reference with (a)).
    (c) Modification of~$\mathcal N$ that removes all neighbors of~$s_v$ different from~$w_\alpha$.}
\end{figure}
If~$v$ is a switch, assume it is a source; if~$v$ is a sink its processing is analogous. Then {\sc $\mathcal N$-modifier} checks whether in~$\mathcal E$, the outgoing left edges and thus also the outgoing right edges are consecutive, where one of the two sets might be empty. If the test is negative, {\sc $\mathcal N$-modifier} concludes that~$G$ admits no upward book embedding respecting~$\mathcal E$. 
If the test is positive and~$v$ does not have any outgoing left edges or any outgoing right edges, the processing of~$v$ is concluded. If the test is positive, and~$v$ has both outgoing left edges and outgoing right edges, we modify~$\mathcal N$ as follows; refer to \cref{fig:flow-network-4-modality}. Let~$e$ be the left edge outgoing from~$v$ such that the edge~$e'$ preceding~$e$ in clockwise order around~$v$ is a right edge, let~$\alpha$ be the angle~$(e,e')$, and let~$f$ be the face to the left of~$e$. For each angle~$\beta\neq \alpha$ incident to~$v$ in~$\mathcal E$, the algorithm {\sc $\mathcal N$-modifier} removes~$w_\beta$ and its incident arcs from~$\mathcal N$. This modification corresponds to forcing the arc~$(s_v,w_\alpha)$ to be assigned one unit of flow, hence making~$\alpha$ large. 
\end{itemize}

We next describe the modifications {\sc $\mathcal N$-modifier} applies to~$\mathcal N$ in order to ensure that the upward embeddings corresponding to the feasible flows of~$\mathcal N$ with value~$d_T$ do not contain any impossible face. Consider any face~$f$ of~$\mathcal E$. Also, consider any maximal directed path~$\ell_f$ in the boundary of~$f$
that has~$f$ to its right (the treatment of the maximal directed paths that have~$f$ to their left is analogous). If~$\ell_f$ contains a left edge or if the rest of the boundary of~$f$ is a single edge, then the processing of~$\ell_f$ is concluded.  Otherwise, let~$u$ and~$v$ be the extremes of~$\ell_f$, let~$e^u_f$ and~$e^v_f$ be the edges of~$\ell_f$ incident to~$u$ and~$v$, let~$\alpha^u_f$ be the angle in~$f$ incident to~$u$ and to the right of~$e^u_f$, and let~$\alpha^v_f$ be the  angle in~$f$ incident to~$v$ and to the right of~$e^v_f$. By \cref{th:characterization}, the angle assignment needs to ensure that at least one of the angles~$\alpha^u_f$ and~$\alpha^v_f$ is~large. Since~$\ell_f$ is a maximal directed path, both~$\alpha^u_f$ and~$\alpha^v_f$ are switch angles. 
We say that~$\alpha^u_f$~(resp.\ $\alpha^v_f$) is \emph{enlargeable} if the node~$w_{\alpha^u_f}$ (resp.\ the node~$w_{\alpha^v_f}$) is in~$\mathcal N$. Note that, even if~$u$ (resp.~$v$) is a switch of~$G$, it might be that~$w_{\alpha^u_f}$ (resp.~$w_{\alpha^v_f}$) is not in~$\mathcal N$, by the effect of some previous modification  
of~$\mathcal N$. We thus distinguish three cases.
\begin{itemize}
\item If neither~$\alpha^u_f$ nor~$\alpha^v_f$ is enlargeable, then {\sc $\mathcal N$-modifier} concludes that~$G$ admits no upward book embedding respecting~$\mathcal E$. 
\item If exactly one of~$\alpha^u_f$ and~$\alpha^v_f$ is enlargeable, say~$\alpha^u_f$ is enlargeable and~$\alpha^v_f$ is not, then {\sc $\mathcal N$-modifier} modifies~$\mathcal N$ as follows: For each angle~$\beta\neq \alpha^u_f$ incident to~$u$ in~$\mathcal E$, the algorithm {\sc $\mathcal N$-modifier} removes~$w_\beta$ and its incident arcs from~$\mathcal N$. This modification corresponds to forcing the arc~$(s_u,w_{\alpha^u_f})$ to be assigned one unit of flow, hence making~$\alpha^u_f$~large. 
\item Finally, consider the situation in which both~$\alpha^u_f$ and~$\alpha^v_f$ are enlargeable; refer to \cref{fig:modification-avoid-impossible-faces}.

If the degree of $s_u$ or $s_v$ in~$\mathcal N$ is one, that is, if there exists no angle~$\beta\neq \alpha^u_f$ incident to~$u$ such that the corresponding node~$w_\beta$ is in~$\mathcal N$, or there exists no angle~$\beta\neq \alpha^v_f$ incident to~$v$ such that the corresponding node~$w_\beta$ is in~$\mathcal N$, then the processing of~$\ell_f$ is concluded. Indeed, it is already guaranteed that one of~$\alpha^u_f$ and~$\alpha^v_f$ will be a large angle. 

\begin{figure}[htb!]
\centering
  \begin{subfigure}{.49\textwidth}\centering
\includegraphics[page=6,scale=.7]{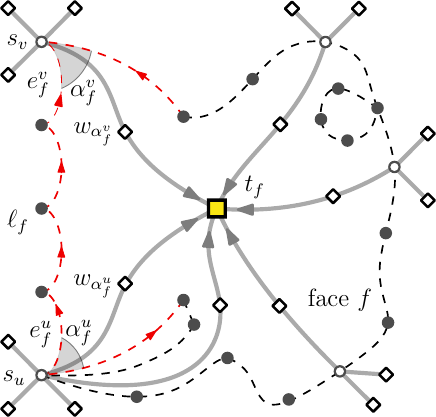}
    \label{fig:before-modification-avoid-impossible-faces}
    \subcaption{}
    \end{subfigure}
    \hfil
  \begin{subfigure}{.49\textwidth}\centering
\includegraphics[page=7,scale=.7]{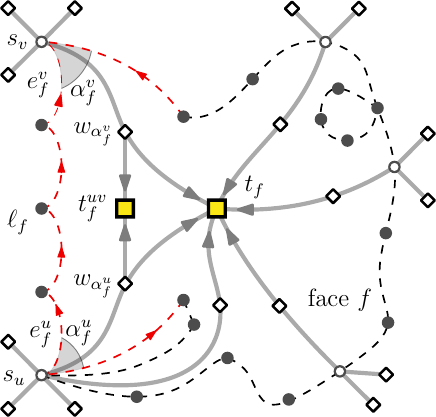}
    \subcaption{}
    \label{fig:after-modification-avoid-impossible-faces}
    \end{subfigure}
\caption{\label{fig:modification-avoid-impossible-faces}Modification to~$\mathcal N$ to avoid impossible faces determined by a maximal directed path~$\ell_f$ composed of right edges that has the face~$f$ on its right, when both angles~$\alpha^u_f$ and~$\alpha^v_f$ at the extremes~$u$ and~$v$ of~$\ell_f$, respectively, are enlargeable and the degree of both $s_u$ and $s_v$ in~$\mathcal N$ is greater than one. The part of~$\mathcal N$ associated with~$f$ before (a) and after (b) the modification.}
\end{figure}

Otherwise, {\sc $\mathcal N$-modifier} adds to~$\mathcal N$ a sink~$t^{uv}_f$ with demand~$1$, as well as arcs~$(w_{\alpha^u_f},t^{uv}_f)$ and~$(w_{\alpha^v_f},t^{uv}_f)$, each with capacity~$1$, and decreases by~$1$ the demand of~$t_f$. This modification to~$\mathcal N$ corresponds to forcing one of~$(w_{\alpha^u_f},t^{uv}_f)$ and~$(w_{\alpha^v_f},t^{uv}_f)$ (and consequently one of~$(s_u,w_{\alpha^u_f})$ and~$(s_v,w_{\alpha^v_f})$) to be assigned one unit of flow, hence making large the angle~$\alpha^u_f$ or the angle~$\alpha^v_f$, respectively. Note that~$\mathcal N$ remains a planar flow network. Also, if~$(w_{\alpha^u_f},t^{uv}_f)$ (resp.\ $(w_{\alpha^v_f},t^{uv}_f)$) is not assigned one unit of flow, then~$(w_{\alpha^u_f},t_f)$ (resp.~$(w_{\alpha^v_f},t_f)$) might be assigned one unit of flow, thus making large both~$\alpha^u_f$ and~$\alpha^v_f$. 
\end{itemize}

We are now ready to state our result for plane digraphs.

\begin{theorem} \label{th:planar-embedding}
Let~$G=(V,L \cup R)$ be an~$n$-vertex partitioned plane digraph with planar embedding~$\mathcal E$. It is possible to test in~$O(n\log^3 n)$ time whether~$G$ admits an upward book embedding respecting~$\mathcal E$. 
\end{theorem}

\begin{proof}
We start by considering instances such that~$G$ is connected. The core of the proof consists of proving the aforementioned bijection between the feasible flows with value~$d_T$ of the network constructed by the algorithm {\sc $\mathcal N$-modifier} and the upward-consistent angle assignments~$\lambda$ such that~$(\mathcal E,\lambda)$ is a good embedding (below we state such a bijection precisely). Then \cref{th:characterization} implies that the existence of an upward book embedding of~$G$ respecting~$\mathcal E$ is equivalent to the fact that {\sc $\mathcal N$-modifier} both 1) did not conclude that~$G$ admits no upward book embedding respecting~$\mathcal E$ and 2) constructed a planar flow network~$\mathcal N$ that has a feasible flow with value~$d_T$. Since {\sc $\mathcal N$-modifier} can be easily implemented to run in~$O(n)$ time and since the algorithm by Borradaile et al.~\cite{DBLP:journals/siamcomp/BorradaileKMNW17} to test whether~$\mathcal N$ admits a feasible flow with the required value runs in~$O(n \log^3 n)$ time, the theorem follows.

{\bf Notation.} We introduce some notation. Suppose first that {\sc $\mathcal N$-modifier} did not conclude that~$G$ admits no upward book embedding respecting~$\mathcal E$. Let~$\mathcal N_0$ be the flow network constructed by Bertolazzi et al.~\cite{DBLP:journals/algorithmica/BertolazziBLM94}. Recall that, in a first phase, {\sc $\mathcal N$-modifier} considers the vertices of~$G$ in some order, performs some checks, and possibly modifies the flow network so to ensure the 4-modality of the vertices. For~$i=1,\dots,n$, denote by~$\mathcal N_i$ the flow network constructed by {\sc $\mathcal N$-modifier} after considering the~$i$-th vertex of~$G$ so to ensure its 4-modality. If {\sc $\mathcal N$-modifier} did not modify the flow network when considering the~$i$-th vertex of~$G$, then~$\mathcal N_i$ is the same network as~$\mathcal N_{i-1}$. In a second phase, {\sc $\mathcal N$-modifier} considers the maximal directed paths on the boundary of the faces (let~$p$ be the number of such paths), performs some checks, and possibly modifies the flow network so to ensure the absence of impossible faces. For~$i \in \{n+1,\dots,n+p\}$, denote by~$\mathcal N_{i}$ the flow network constructed by {\sc $\mathcal N$-modifier} after considering the~$(i-n)$-th maximal directed path on the boundary of some face of~$\mathcal E$. If {\sc $\mathcal N$-modifier} did not modify the flow network when considering such a path, then~$\mathcal N_{i}$ is the same network as~$\mathcal N_{i-1}$. If the algorithm did conclude that~$G$ admits no upward book embedding respecting~$\mathcal E$, then the notation for the flow networks constructed by {\sc $\mathcal N$-modifier} is restricted to the networks constructed before the termination.  Let~$\mathcal N_k$ be the last constructed network ($k=n+p$ if {\sc $\mathcal N$-modifier} did not conclude that~$G$ admits no upward book embedding respecting~$\mathcal E$, and~$k<n+p$ otherwise). For any~$i\in \{0,1,\dots,k\}$, let~$\mathcal V_i$ be the set of the first~$\min\{i,n\}$ vertices processed by {\sc $\mathcal N$-modifier} and let~$\mathcal P_i$ be the set of the first~$\max\{i-n,0\}$ maximal directed paths processed by~$\mathcal N$-modifier.

{\bf Structure.} We make some observations on the structure of the networks~$\mathcal N_0,\mathcal N_1,\dots,\mathcal N_k$. 
\begin{itemize}
	\item {\em Demand preservation.} The sum of the demands of the sinks and the sum of the demands of the sources is the same value, which we denote by~$d_T$, in all the networks~$\mathcal N_0,\mathcal N_1,\dots,\mathcal N_k$. 
	\item {\em Outgoing arcs for angle nodes.} Any node~$w_\alpha$ corresponding to a switch angle~$\alpha$ in a network~$\mathcal N_i$, for some~$i\in \{0,\dots,k\}$, has an outgoing arc to the sink~$t_f$ corresponding to the face~$f$ of $\mathcal E$ angle~$\alpha$ is incident to. Also,~$w_\alpha$ has at most one more outgoing arc; such an arc, if it exists, is directed towards a sink {\em associated with~$t_f$} that is introduced when processing a path in the boundary of~$f$.
\end{itemize}

Concerning property {\em demand preservation}, note that whenever {\sc $\mathcal N$-modifier} modifies the sinks and their demands, it decreases the demand of a sink~$t_f$ by~$1$ and introduces a sink~$t^{uv}_f$ associated with~$t_f$ with demand~$1$. Thus the sum of the demands of the sinks remains the same. Also, the algorithm never modifies the sources and their supplies.

In order to prove property {\em outgoing arcs for angle nodes}, consider a node~$w_\alpha$ corresponding to a switch angle~$\alpha$ in a network~$\mathcal N_i$, for some~$i\in \{0,\dots,k\}$. 
Let~$v$ and~$f$ be  the vertex of~$G$ and the face of~$\mathcal E$ incident to the angle~$\alpha$, respectively. 
First,~$w_\alpha$ has an outgoing arc to the sink~$t_f$ corresponding to~$f$. Indeed, if {\sc $\mathcal N$-modifier} removed the arc~$(w_\alpha,t_f)$ from a network~$\mathcal N_j$, for some~$j<i$, then it also removed the vertex~$w_\alpha$, which contradicts the assumption that~$w_\alpha$ belongs to~$\mathcal N_i$. 
Second, we prove that~$w_\alpha$ has at most one more outgoing arc and that such an arc, if it exists, is directed towards a sink associated with~$t_f$. Recall that~$v$ is a switch in~$G$. 
Assume it is a sink, the case in which it is a source is analogous. The algorithm {\sc $\mathcal N$-modifier} inserts an arc outgoing~$w_\alpha$ in a network~$\mathcal N_j$ with~$j<i$ only if:
\begin{itemize}
	\item  there exists a maximal directed path~$\ell_f$ in the boundary of~$f$ that ends at~$v$, is entirely composed of right edges, has~$f$ to its right, and the angle in~$f$ to the right of the last edge of~$\ell_f$ is~$\alpha$; or
	\item  there exists a maximal directed path~$r_f$ in the boundary of~$f$ that ends at~$v$, is entirely composed of left edges, has~$f$ to its left, and the angle in~$f$ to the left of the last edge of~$r_f$ is~$\alpha$.
\end{itemize} 
It remains to observe that if both the above paths~$\ell_f$ and~$r_f$ exist, then~$\alpha$ is delimited by a right and a left edge; refer to \cref{fig:both_paths_and_4-modality}. Hence, in order to ensure the 4-modality of~$v$, the algorithm {\sc $\mathcal N$-modifier} would have removed from the network all the nodes corresponding to angles incident to~$v$, except for~$\alpha$. Thus, it would not have introduced any new sink associated with~$t_f$ when processing paths~$\ell_f$ and~$r_f$.

\begin{figure}[h!]
\centering
  \begin{subfigure}{.49\textwidth}\centering
\includegraphics[page=9,scale=.7]{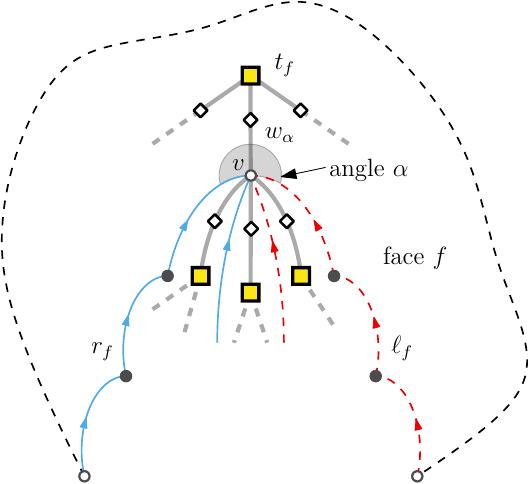}
    \label{fig:both_paths_and_4-modality_before}
    \subcaption{}
    \end{subfigure}
    \hfil
  \begin{subfigure}{.49\textwidth}\centering
\includegraphics[page=8,scale=.7]{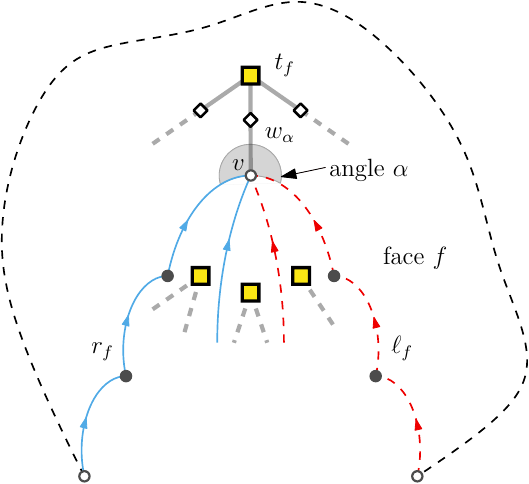}
    \subcaption{}
    \label{fig:both_paths_and_4-modality_after}
    \end{subfigure}
\caption{\label{fig:both_paths_and_4-modality} Illustration for the proof of property  {\em outgoing arcs for angle nodes} in the case in which there exists paths $\ell_f$ and $r_f$ that have a common face $f$ to their right and left, respectively, and are directed toward a sink $v$ of $G$. The network $\mathcal N$ before (a) and after (b) applying the transformation that ensures the $4$-modality of $v$.}
\end{figure}


{\bf Correspondence.} Consider an upward-consistent angle assignment~$\lambda$ for~$\mathcal E$ and consider a network~$\mathcal N_i$ with~$i\in \{0,1,\dots,k\}$. We define a corresponding flow~$\mathcal F_\lambda(\mathcal N_i)$ as follows. For each angle~$\alpha$ of~$\mathcal E$ that is assigned a large angle by~$\lambda$, consider the node~$w_{\alpha}$ (if it is in~$\mathcal N_i$). Then~$\mathcal F_\lambda(\mathcal N_i)$ assigns flow~$1$ to the unique arc incoming into~$w_{\alpha}$. By property {\em outgoing arcs for angle nodes},~$w_{\alpha}$ has either one or two outgoing arcs. If~$w_{\alpha}$ has one outgoing arc, then~$\mathcal F_\lambda(\mathcal N_i)$ assigns flow~$1$ to it. If~$w_{\alpha}$ has two outgoing arcs, then one of them is directed to the sink~$t_f$ corresponding to the face~$f$ angle~$\alpha$ is incident to, and one of them is directed to a sink~$t^{uv}_f$ associated with~$t_f$. If an arc incoming into~$t^{uv}_f$ has already been assigned flow~$1$, then~$\mathcal F_\lambda(\mathcal N_i)$ assigns flow~$1$ to the arc~$(w_{\alpha},t_f)$, otherwise it assigns flow~$1$ to the arc~$(w_{\alpha},t^{uv}_f)$. After processing all angles of~$\mathcal E$ that are assigned a large angle by~$\lambda$,~$\mathcal F_\lambda(\mathcal N_i)$ assigns flow~$0$ to all arcs of~$\mathcal N_i$ that are not assigned flow~$1$. Since all arcs of~$\mathcal N_i$ have capacity~$1$, obviously~$\mathcal F_\lambda(\mathcal N_i)$ is feasible. However, we will need to prove that the value of~$\mathcal F_\lambda(\mathcal N_i)$ is~$d_T$.

Conversely, consider a feasible flow~$\mathcal F$ for a network~$\mathcal N_i$ with~$i\in \{0,1,\dots,k\}$. We define a corresponding angle-assignment~$\lambda_\mathcal F$ for~$\mathcal E$ as follows. We have that~$\lambda_\mathcal F$ assigns~$0$ to every flat angle of~$\mathcal E$. Also, it assigns~$1$ to a switch angle~$\alpha$ at a vertex~$v$ if and only if~$v$ is a switch in~$G$, the node~$w_{\alpha}$ is in~$\mathcal N_i$, and the arc~$(s_v,w_{\alpha})$ is assigned one unit of flow by~$\mathcal F$. Finally, it assigns~$-1$ to every switch angle~$\alpha$ that is not assigned flow~$1$. 

\smallskip\noindent {\bf Bijection.} The bijection is formalized as follows. For any~$i\in \{0,1,\dots,k\}$, we have that:
\begin{enumerate}
\item If an angle assignment~$\lambda$ for~$\mathcal E$ is upward-consistent and the upward embedding~$(\mathcal E,\lambda)$ is such that the vertices in~$\mathcal V_i$ are 4-modal and the paths in~$\mathcal P_i$ do not cause an impossible face, then {\sc $\mathcal N$-modifier} does not conclude that~$G$ admits no upward book embedding respecting~$\mathcal E$ when processing these vertices and paths. Also, the flow~$\mathcal F_{\lambda}(\mathcal N_i)$ has value~$d_T$.
\item If {\sc $\mathcal N$-modifier} does not conclude that~$G$ admits no upward book embedding respecting~$\mathcal E$ when processing the vertices in~$\mathcal V_i$ and the paths in~$\mathcal P_i$, and if a feasible flow~$\mathcal F$ for~$\mathcal N_i$ has value~$d_T$, then~$(\mathcal E,\lambda_{\mathcal F})$ is an upward embedding in which the vertices in~$\mathcal V_i$ are 4-modal and the paths in~$\mathcal P_i$ do not cause an impossible face.  
\end{enumerate}

When~$i=k$, the bijection provides the desired correspondence between the feasible flows with value~$d_T$ of~$\mathcal N_k$ and the upward-consistent angle assignments~$\lambda$ such that~$(\mathcal E,\lambda)$ is 4-modal and does not contain any impossible face, since~$\mathcal V_k$ contains all the vertices of~$G$ and~$\mathcal P_k$ all the maximal directed paths on the boundary of any face.


In the base case of the induction, we have~$i=0$ and the statement follows by the results of Bertolazzi et al.~\cite{DBLP:journals/algorithmica/BertolazziBLM94}, since~$\mathcal N_0$ is the network they construct and the sets~$\mathcal V_0$ and~$\mathcal P_0$ are empty, hence the additional constraints of our bijection are vacuously satisfied. Suppose now that the bijection holds true for some value~$i-1\in \{0,1,\dots,k-1\}$, we prove that it holds true for~$i$, as well.

{\bf First implication.} Suppose that an angle assignment~$\lambda$ for~$\mathcal E$ is upward-consistent and the upward embedding~$(\mathcal E,\lambda)$ is such that the vertices in~$\mathcal V_i$ are 4-modal and the paths in~$\mathcal P_i$ do not cause an impossible face. By induction, {\sc $\mathcal N$-modifier} did not conclude that~$G$ admits no upward book embedding respecting~$\mathcal E$ when processing the vertices in~$\mathcal V_{i-1}$ and the paths in~$\mathcal P_{i-1}$. Also, the flow~$\mathcal F_{\lambda}(\mathcal N_{i-1})$ has value~$d_T$. We have to prove that the same statements hold true with~$i$ in place of~$i-1$.

Suppose first that~$i\leq n$ (which implies that~$\mathcal P_i=\emptyset$). Let~$v$ be the~$i$-th vertex processed by~$\mathcal N$-modifier:~$\mathcal V_i=\mathcal V_{i-1}\cup \{v\}$. 
We first show that {\sc $\mathcal N$-modifier} did not conclude that~$G$ admits no upward book embedding respecting~$\mathcal E$ when processing~$v$. Such a conclusion is reached by the algorithm in one of two cases, namely if: (i)~$v$ is not a switch of~$G$ and in~$\mathcal E$, in clockwise order around~$v$, we do not have outgoing left edges, outgoing right edges, incoming right edges, and incoming left edges; or if (ii)~$v$ is a switch of~$G$ and in~$\mathcal E$ the left edges (and thus also the right edges) are not consecutive. 
However, the fact that~$v$ is 4-modal in~$(\mathcal E,\lambda)$ ensures that these cases do not happen. 
It follows that {\sc $\mathcal N$-modifier} did not conclude that~$G$ admits no upward book embedding respecting~$\mathcal E$ when processing~$v$ and constructed a flow network~$\mathcal N_i$. 

In order to prove that~$\mathcal F_{\lambda}(\mathcal N_{i})$ satisfies the required properties, we distinguish two cases. 
\begin{itemize}
\item In the first case,~$\mathcal N_i$ coincides with~$\mathcal N_{i-1}$. 
This implies that~$\mathcal F_{\lambda}(\mathcal N_{i})$ coincides with~$\mathcal F_{\lambda}(\mathcal N_{i-1})$, hence it has value~$d_T$.
\item In the second case,~$\mathcal N_i$ is different from~$\mathcal N_{i-1}$. This happens if~$v$ is a switch in~$G$ and has both left and right edges incident to it. Assume that~$v$ is a source, the case in which it is a sink being analogous. Let~$e$ be the left edge outgoing from~$v$ such that the edge~$e'$ preceding~$e$ in clockwise order around~$v$ is a right edge, and let~$\alpha$ be the angle~$(e,e')$. The 4-modality of~$v$ forces~$\alpha$ to be a large angle in~$(\mathcal E,\lambda)$. Hence,~$\mathcal F_{\lambda}(\mathcal N_{i-1})$ assigns flow~$1$ to the arcs incident to~$w_\alpha$ (note that there is a single arc outgoing from~$w_\alpha$, by the assumption~$i\leq n$ and by property {\em outgoing arcs for angle nodes}). Since the supply of~$s_v$ is~$1$, we have that~$\mathcal F_{\lambda}(\mathcal N_{i-1})$ assigns flow~$0$ to the arcs incident to~$w_\beta$, for every angle~$\beta\neq \alpha$ incident to~$v$. Since~$\mathcal N_i$ coincides with~$\mathcal N_{i-1}$, apart from the fact that, for every angle~$\beta\neq \alpha$ incident to~$v$, the node~$w_\beta$ and the arcs incident to it are not present in~$\mathcal N_i$, it follows that~$\mathcal F_{\lambda}(\mathcal N_{i})$ coincides with the flow for~$\mathcal N_i$ which is obtained from~$\mathcal F_{\lambda}(\mathcal N_{i-1})$ by neglecting the arcs~incident to~$w_\beta$, for every angle~$\beta\neq \alpha$ incident to~$v$. Since such arcs are assigned flow~$0$ by~$\mathcal F_{\lambda}(\mathcal N_{i-1})$, we have that~$\mathcal F_{\lambda}(\mathcal N_{i})$ has value~$d_T$. 
\end{itemize}

Suppose next that~$n+1\leq i\leq k$ (which implies that~$\mathcal V_i$ is the entire vertex set of~$G$). Consider the~$(i-n)$-th maximal directed path processed by~$\mathcal N$-modifier. Suppose that it is in the boundary of a face~$f$, with~$f$ to its right, the case in which the path has~$f$ to its left being analogous. Denote by~$\ell_f$ such a path. 

If~$\ell_f$ contains a left edge or the rest of the boundary of~$f$ is a single edge, then~$\mathcal N_i$ coincides with~$\mathcal N_{i-1}$. It follows that~$\mathcal F_{\lambda}(\mathcal N_{i})$ coincides with~$\mathcal F_{\lambda}(\mathcal N_{i-1})$, hence it has value~$d_T$. 

If~$\ell_f$ only consists of right edges and the rest of the boundary of~$f$ is not a single edge, then let~$u$ and~$v$ be the extremes of~$\ell_f$, let~$e^u_f$ and~$e^v_f$ be the edges of~$\ell_f$ incident to~$u$ and~$v$, respectively, let~$\alpha^u_f$ be the angle in~$f$ incident to~$u$ and to the right of~$e^u_f$, and let~$\alpha^v_f$ be the  angle in~$f$ incident to~$v$ and to the right of~$e^v_f$. Since~$\ell_f$ does not cause~$f$ to be an impossible face, it follows that one of~$\alpha^u_f$ and~$\alpha^v_f$ is large in~$(\mathcal E,\lambda)$. 

Suppose that~$\alpha^u_f$ is large in~$(\mathcal E,\lambda)$, the case in which~$\alpha^v_f$ is large in~$(\mathcal E,\lambda)$ is analogous. This implies that~$u$ is a switch of~$G$, hence~$s_u$ is a source in~$\mathcal N_{i-1}$ and~$\mathcal N_{i}$. Also, the node~$w_{\alpha^u_f}$ belongs to~$\mathcal N_{i-1}$ and the arc~$(s_u,w_{\alpha^u_f})$ is assigned one unit of flow by~$\mathcal F_{\lambda}(\mathcal N_{i-1})$.
It follows that~$\alpha^u_f$ is enlargeable in~$\mathcal N_{i-1}$ and hence the algorithm {\sc $\mathcal N$-modifier} did not conclude that~$G$ admits no upward book embedding respecting~$\mathcal E$ when processing~$\ell_f$. 
Since the supply of~$s_u$ is~$1$, we have that~$\mathcal F_{\lambda}(\mathcal N_{i-1})$ assigns flow~$0$ to the arcs incident to~$w_\beta$, for every angle~$\beta\neq \alpha^u_f$ incident to~$u$. We now distinguish three cases.

\begin{itemize}
\item Suppose first that~$\mathcal N_i$ coincides with~$\mathcal N_{i-1}$. This happens if there is no angle~$\beta\neq \alpha^u_f$ incident to~$u$ in~$\mathcal E$ such that~$w_\beta$ is in~$\mathcal N_{i-1}$, or if~$\alpha^v_f$ is also enlargeable in~$\mathcal N_{i-1}$ and there is no angle~$\beta\neq \alpha^v_f$ incident to~$v$ in~$\mathcal E$ such that~$w_\beta$ is in~$\mathcal N_{i-1}$. In this case~$\mathcal F_{\lambda}(\mathcal N_{i})$ coincides with~$\mathcal F_{\lambda}(\mathcal N_{i-1})$, hence it has value~$d_T$.
\item Suppose next that~$\alpha^v_f$ is not enlargeable in~$\mathcal N_{i-1}$ and there is an angle~$\beta\neq \alpha^u_f$ incident to~$u$ in~$\mathcal E$ such that~$w_\beta$ is in~$\mathcal N_{i-1}$. By definition,~$w_{\alpha^v_f}$ does not belong to~$\mathcal N_{i-1}$. Then the network~$\mathcal N_i$ coincides with~$\mathcal N_{i-1}$, apart from the fact that, for every angle~$\beta\neq \alpha^u_f$ incident to~$u$ in~$\mathcal E$, the arcs incident to~$w_\beta$ are not present in~$\mathcal N_i$. It follows that~$\mathcal F_{\lambda}(\mathcal N_{i})$ coincides with the flow for~$\mathcal N_i$ which is obtained from~$\mathcal F_{\lambda}(\mathcal N_{i-1})$ by neglecting the arcs~incident to~$w_\beta$, for every angle~$\beta\neq \alpha$ incident to~$u$. Since such arcs are assigned flow~$0$ by~$\mathcal F_{\lambda}(\mathcal N_{i-1})$, we have that~$\mathcal F_{\lambda}(\mathcal N_{i})$ has value~$d_T$. 
\item Suppose finally that there exists an angle~$\beta\neq \alpha^u_f$ incident to~$u$ such that the corresponding node~$w_\beta$ is in~$\mathcal N_{i-1}$, that~$\alpha^v_f$ is enlargeable in~$\mathcal N_{i-1}$, and that there exists an angle~$\beta'\neq \alpha^v_f$ incident to~$v$ such that the corresponding node~$w_{\beta'}$ is in~$\mathcal N_{i-1}$. Then~$\mathcal N_i$ is obtained from~$\mathcal N_{i-1}$ by inserting a sink~$t^{uv}_f$ with demand~$1$, as well as arcs~$(w_{\alpha^u_f},t^{uv}_f)$ and~$(w_{\alpha^v_f},t^{uv}_f)$, each with capacity~$1$, and by decreasing by~$1$ the demand of~$t_f$. Note that~$w_{\alpha^u_f}$ has a single outgoing arc~$(w_{\alpha^u_f},t_f)$ in~$\mathcal N_{i-1}$, as if it had at least two outgoing arcs in~$\mathcal N_{i-1}$ it would have at least three outgoing arcs in~$\mathcal N_i$, which is not possible by property {\em outgoing arcs for angle nodes}. It follows that~$(w_{\alpha^u_f},t_f)$ (in addition to~$(s_u,w_{\alpha^u_f})$) is assigned one unit of flow by~$\mathcal F_{\lambda}(\mathcal N_{i-1})$. Analogously,~$w_{\alpha^v_f}$ has a single outgoing arc~$(w_{\alpha^v_f},t_f)$ in~$\mathcal N_{i-1}$ and, if~$\alpha^v_f$ is large, then~$(s_v,w_{\alpha^v_f})$ and~$(w_{\alpha^v_f},t_f)$ are assigned one unit of flow by~$\mathcal F_{\lambda}(\mathcal N_{i-1})$. Now~$\mathcal F_{\lambda}(\mathcal N_{i})$ coincides with~$\mathcal F_{\lambda}(\mathcal N_{i-1})$, except that a unit of flow that is assigned to~$(w_{\alpha^u_f},t_f)$ or~$(w_{\alpha^v_f},t_f)$ by~$\mathcal F_{\lambda}(\mathcal N_{i-1})$ is assigned to~$(w_{\alpha^u_f},t^{uv}_f)$ or~$(w_{\alpha^v_f},t^{uv}_f)$, respectively, by~$\mathcal F_{\lambda}(\mathcal N_{i})$. Thus~$\mathcal F_{\lambda}(\mathcal N_{i})$ has value~$d_T$ since~$\mathcal F_{\lambda}(\mathcal N_{i-1})$ has value~$d_T$.   
\end{itemize}

This completes the induction and hence the proof of the forward implication.


\smallskip
\noindent {\bf Second implication.} Suppose that {\sc $\mathcal N$-modifier} did not conclude that~$G$ admits no upward book embedding respecting~$\mathcal E$ when processing the vertices in~$\mathcal V_i$ and the paths in~$\mathcal P_i$. Suppose also that a feasible flow~$\mathcal F_i$ for~$\mathcal N_i$ has value~$d_T$. From~$\mathcal F_i$, we construct a flow~$\mathcal F_{i-1}$ for~$\mathcal N_{i-1}$ as follows. 

\begin{itemize}
	\item If~$\mathcal N_i$ coincides with~$\mathcal N_{i-1}$, then~$\mathcal F_{i-1}$ coincides with~$\mathcal F_i$.
	\item Suppose next that~$\mathcal N_i$ is constructed from~$\mathcal N_{i-1}$ by removing, for a certain angle~$\alpha$ of~$\mathcal E$ incident to a vertex~$v$, the node~$w_\beta$ and the arcs incident to it, for every angle~$\beta\neq \alpha$ incident to~$v$. Then~$\mathcal F_{i-1}$ is obtained from~$\mathcal F_i$ by assigning flow~$0$ to the arcs that belong to~$\mathcal N_{i-1}$ and not to~$\mathcal N_i$.
	\item Suppose finally that 	$\mathcal N_i$ is constructed from~$\mathcal N_{i-1}$ by adding a sink~$t^{uv}_f$ with demand~$1$, by adding two arcs~$(w_{\alpha^u_f},t^{uv}_f)$ and~$(w_{\alpha^v_f},t^{uv}_f)$, each with capacity~$1$, and by  decreasing by~$1$ the demand of a sink~$t_f$. Then~$\mathcal F_{i-1}$ is obtained from~$\mathcal F_i$ by neglecting the node~$t^{uv}_f$ and its incident arcs, and by assigning flow~$1$ to the arc~$(w_{\alpha^u_f},t_f)$ or~$(w_{\alpha^v_f},t_f)$, depending on which of~$(w_{\alpha^u_f},t^{uv}_f)$ and~$(w_{\alpha^v_f},t^{uv}_f)$ is assigned with flow~$1$, respectively.   
\end{itemize}

In all cases,~$\mathcal F_{i-1}$ is feasible and has value~$d_T$, given that the same is true for~$\mathcal F_i$. This is obvious for the first two cases, while for the third case one needs to observe that: (i) one of the arcs~$(w_{\alpha^u_f},t^{uv}_f)$ and~$(w_{\alpha^v_f},t^{uv}_f)$ is assigned with flow~$1$, given that~$t^{uv}_f$ has demand~$1$ and~$\mathcal F_i$ has value~$d_T$; (ii) the arcs~$(w_{\alpha^u_f},t_f)$ or~$(w_{\alpha^v_f},t_f)$ belong to~$\mathcal N_{i-1}$, as the modification only occurs when~$\alpha^u_f$ and~$\alpha^v_f$ are both enlargeable; and (iii) if~$(w_{\alpha^u_f},t^{uv}_f)$ is assigned flow~$1$ by~$\mathcal F_i$, then~$(w_{\alpha^u_f},t_f)$ is assigned flow~$0$ by~$\mathcal F_i$, given that~$s_u$ supplies~$1$ unit of flow; similarly, if~$(w_{\alpha^v_f},t^{uv}_f)$ is assigned flow~$1$ by~$\mathcal F_i$, then~$(w_{\alpha^v_f},t_f)$ is assigned flow~$0$ by~$\mathcal F_i$. 

By induction,~$(\mathcal E,\lambda_{\mathcal F_{i-1}})$ is an upward embedding in which the vertices in~$\mathcal V_{i-1}$ are 4-modal and the paths in~$\mathcal P_{i-1}$ do not cause an impossible face. We have to prove that the same statements hold true with~$i$ in place of~$i-1$. 

Since~$\mathcal F_{i-1}$ assigns flow~$1$ to an arc~$(s_u,w_\alpha)$ outgoing from a source~$s_u$ of~$\mathcal N_{i-1}$ if and only if the same arc in~$\mathcal N_i$ is assigned flow~$1$ by~$\mathcal F_i$, we have that~$\lambda_{\mathcal F_i}$ coincides with~$\lambda_{\mathcal F_{i-1}}$. This implies that~$\lambda_{\mathcal F_{i}}$ is upward-consistent, that each vertex~$u\in \mathcal V_{i-1}$ is 4-modal in~$(\mathcal E,\lambda_{\mathcal F_i})$, and that each path in~$\mathcal P_{i-1}$ does not cause an impossible face in~$(\mathcal E,\lambda_{\mathcal F_i})$. Thus it only remains to prove that the~$i$-th vertex processed by {\sc $\mathcal N$-modifier} is 4-modal in~$(\mathcal E,\lambda_{\mathcal F_i})$, if~$i\leq n$, or that the~$(i-n)$-th maximal directed path processed by {\sc $\mathcal N$-modifier} does not cause an impossible face, if~$n+1\leq i\leq k$.

Suppose first that~$i\leq n$. Let~$v$ be the~$i$-th vertex processed by~$\mathcal N$-modifier. We distinguish two cases.

\begin{itemize}
	\item If~$\mathcal N_i$ coincides with~$\mathcal N_{i-1}$, then either: (i)~$v$ is not a switch in~$G$, and in~$\mathcal E$, in clockwise order around~$v$, we have outgoing left edges, outgoing right edges, incoming right edges, and incoming left edges (where one of the former two sets and/or one of the latter two sets might be empty); or (ii)~$v$ is a switch in~$G$ and all its incident edges are left edges or they all are right edges. In both situations,~$v$ is 4-modal in~$(\mathcal E,\lambda_{\mathcal F_i})$, regardless of the large-angle assignment~$\lambda_{\mathcal F_{i}}$.  
	\item If~$\mathcal N_i$ is different from~$\mathcal N_{i-1}$, then~$v$ is a switch; assume it is a source, the case in which it is a sink is analogous. Also,~$v$ has both outgoing left edges and outgoing right edges, where the former (and thus also the latter) appear consecutively around~$v$. Let~$e$ be the left edge outgoing from~$v$ such that the edge~$e'$ preceding~$e$ in clockwise order around~$v$ is a right edge, let~$\alpha$ be the angle~$(e,e')$, and let~$f$ be the face to the left of~$e$. Then~$\mathcal N_i$ is obtained by removing~$w_\beta$ and its incident arcs from~$\mathcal N_{i-1}$, for each angle~$\beta\neq \alpha$ incident to~$v$. Since~$w_\alpha$ is the only neighbor of~$s_v$ in~$\mathcal N_i$ and since~$\mathcal F_i$ has value~$d_T$, it follows that the angle~$\alpha$ is large in~$(\mathcal E,\lambda_{\mathcal F_i})$ (and all other angles incident to~$v$ are small), hence~$v$ is 4-modal. 
\end{itemize}

Suppose next that~$n+1\leq i\leq k$. Consider the~$(i-n)$-th maximal directed path processed by~$\mathcal N$-modifier. Suppose that it is in the boundary of a face~$f$, with~$f$ to its right, the case in which the path has~$f$ to its left being analogous. Denote by~$\ell_f$ such a path. Let~$u$ and~$v$ be the extremes of~$\ell_f$, let~$e^u_f$ and~$e^v_f$ be the edges of~$\ell_f$ incident to~$u$ and~$v$, respectively, let~$\alpha^u_f$ be the angle in~$f$ incident to~$u$ and to the right of~$e^u_f$, and let~$\alpha^v_f$ be the  angle in~$f$ incident to~$v$ and to the right of~$e^v_f$. We distinguish three cases.

\begin{itemize}
	\item If~$\mathcal N_i$ coincides with~$\mathcal N_{i-1}$, then one of the following is true: (i)~$\ell_f$ contains a left edge; (ii) the rest of the boundary of~$f$ is a single edge; (iii) exactly one of~$\alpha^u_f$ and~$\alpha^v_f$ is enlargeable, say~$\alpha^u_f$ is enlargeable and~$\alpha^v_f$ is not, and for each angle~$\beta\neq \alpha^u_f$ incident to~$u$ in~$\mathcal E$, node~$w_\beta$ is not in~$\mathcal N_{i-1}$; or (iv) both~$\alpha^u_f$ and~$\alpha^v_f$ are enlargeable, and there is no angle~$\beta\neq \alpha^u_f$ incident to~$u$ in~$\mathcal E$ such that~$w_\beta$ is in~$\mathcal N_{i-1}$ or there is no angle~$\beta\neq \alpha^v_f$ incident to~$v$ in~$\mathcal E$ such that~$w_\beta$ is in~$\mathcal N_{i-1}$. 
    
    Cases (i) and (ii) directly imply that~$\ell_f$ does not cause~$f$ to be an impossible face in~$(\mathcal E,\lambda_{\mathcal F_i})$. In Case (iii), since~$\mathcal F_i$ has value~$d_T$ and~$(s_u,w_{\alpha^u_f})$ is the only edge outgoing from~$s_u$ in~$\mathcal N_i$, we have that~$\mathcal F_i$ assigns flow~$1$ to~$(s_u,w_{\alpha^u_f})$, hence~$\lambda_{\mathcal F_i}$ assigns a large angle to~$\alpha^u_f$; it follows that~$\ell_f$ does not cause~$f$ to be an impossible face in~$(\mathcal E,\lambda_{\mathcal F_i})$. Similarly, in case (iv), since~$\mathcal F_i$ has value~$d_T$ and since~$(s_u,w_{\alpha^u_f})$ is the only edge outgoing from~$s_u$ or~$(s_v,w_{\alpha^v_f})$ is the only edge outgoing from~$s_v$, we have that~$\mathcal F_i$ assigns flow~$1$ to~$(s_u,w_{\alpha^u_f})$ or~$(s_v,w_{\alpha^v_f})$, hence~$\lambda_{\mathcal F_i}$ assigns a large angle to~$\alpha^u_f$ or~$\alpha^v_f$; it follows that~$\ell_f$ does not cause~$f$ to be an impossible face in~$(\mathcal E,\lambda_{\mathcal F_i})$.
	\item If exactly one of~$\alpha^u_f$ and~$\alpha^v_f$ is enlargeable, say~$\alpha^u_f$ is enlargeable and~$\alpha^v_f$ is not, and if~$\mathcal N_{i-1}$ contains at least one node~$w_\beta$ such that~$\beta\neq \alpha^u_f$ is an angle incident to~$u$ in~$\mathcal E$, then~$\mathcal N_i$ is obtained by removing, for each angle~$\beta\neq \alpha^u_f$ in~$\mathcal E$ incident to~$u$, node~$w_\beta$ and its incident arcs from~$\mathcal N_{i-1}$. As in the previous case, since~$\mathcal F_i$ has value~$d_T$ and~$(s_u,w_{\alpha^u_f})$ is the only edge outgoing from~$s_u$ in~$\mathcal N_i$, we have that~$\mathcal F_i$ assigns flow~$1$ to~$(s_u,w_{\alpha^u_f})$, hence~$\lambda_{\mathcal F_i}$ assigns a large angle to~$\alpha^u_f$; it follows that~$\ell_f$ does not cause~$f$ to be an impossible face in~$(\mathcal E,\lambda_{\mathcal F_i})$.
	\item Finally, suppose that both~$\alpha^u_f$ and~$\alpha^v_f$ are enlargeable, that there exists an angle~$\beta\neq \alpha^u_f$ incident to~$u$ in~$\mathcal E$ such that~$w_\beta$ is in~$\mathcal N_{i-1}$, and that there exists an angle~$\beta\neq \alpha^v_f$ incident to~$v$ in~$\mathcal E$ such that the corresponding node~$w_\beta$ is in~$\mathcal N_{i-1}$. Then~$\mathcal N_i$ is obtained by adding to~$\mathcal N_{i-1}$ a sink~$t^{uv}_f$ with demand~$1$, as well as arcs~$(w_{\alpha^u_f},t^{uv}_f)$ and~$(w_{\alpha^v_f},t^{uv}_f)$, each with capacity~$1$, and by decreasing by~$1$ the demand of~$t_f$. Since~$\mathcal F_i$ has value~$d_T$, one of~$(w_{\alpha^u_f},t^{uv}_f)$ and~$(w_{\alpha^v_f},t^{uv}_f)$ is assigned flow~$1$ by~$\mathcal F_i$, thus one of~$(s_u,w_{\alpha^u_f})$ and~$(s_v,w_{\alpha^v_f})$ is also assigned flow~$1$, hence~$\lambda_{\mathcal F_i}$ assigns a large angle to~$\alpha^u_f$ or~$\alpha^v_f$; it follows that~$\ell_f$ does not cause~$f$ to be an impossible face in~$(\mathcal E,\lambda_{\mathcal F_i})$.
\end{itemize}
This completes the proof of the second implication and hence of the theorem in case the graph is connected.

If~$G$ is disconnected, we can apply the same technique to test in~$O(n \log^3 n)$ time whether all connected components of~$G$ admit an upward book embedding that respects their given planar embeddings.  If the test fails, then~$G$ does not admit an upward book embedding.  Otherwise, as argued at the end of the proof of \Cref{th:characterization}, an upward book embedding exists if and only if for each component~$G_j$ that is embedded in an internal face~$f$ of a component~$G_i$, we have that the boundary of~$f$ in~$G_i$ contains both edges from~$L$ and~$R$.  This can be tested in total~$O(n)$ time.
\end{proof}

\section{Test for Biconnected Partitioned Directed Partial~$2$-Trees}

\newcommand{\upmu}{$(\mathcal E_\mu,\lambda_\mu)$\xspace}
\newcommand{\true}{1\xspace}
\newcommand{\false}{0\xspace}

In this section, we show a cubic-time algorithm to test whether a biconnected partitioned directed partial~$2$-tree admits an upward book embedding in two pages. Many of the ideas presented in this section build on tools introduced in~\cite{DBLP:conf/compgeom/ChaplickGFGRS22,DBLP:conf/gd/ChaplickGFGRS22} to test efficiently whether a directed partial~$2$-tree admits an upward embedding. Note that, in the absence of a characterization such as the one in~\cref{th:characterization}, it would be prohibitive to lift such tools to work for~our~problem.

Let~$G$ be an~$n$-vertex biconnected partitioned directed partial~$2$-tree and let~$e^*$ be an edge of~$G$. We describe a test that determines, in~$O(n^2)$ time, whether~$G$ admits a good embedding in which~$e^*$ lies on the outer face. Repeating this test for all~$O(n)$ choices of~$e^*$ yields an~$O(n^3)$-time algorithm to decide whether~$G$ admits a good embedding. By \cref{th:characterization}, this is equivalent to testing whether~$G$ admits an upward book embedding. 


Let~$T$ be an SPQ-tree of the underlying graph of~$G$, rooted at the Q-node~$\rho^*$ corresponding to~$e^*$. Let~$\mu$ be a node of~$T$ with poles~$u$ and~$v$; recall that~$G_{\mu}$ denotes the pertinent graph of~$\mu$. A \emph{$uv$-external upward embedding} \upmu of~$G_{\mu}$ is an upward embedding of~$G_{\mu}$ in which~$u$ and~$v$ are incident to the outer face. The requirement that~$e^*$ is incident to the outer face of any upward embedding~$(\mathcal E,\lambda)$ of~$G$ implies that, for each node~$\mu$ of~$T$ with poles~$u$ and~$v$, the restriction of~$(\mathcal E,\lambda)$ to~$G_{\mu}$ is a~$uv$-external upward embedding~\upmu of~$G_{\mu}$. 
A \emph{$uv$-external good embedding} is a~$uv$-external upward embedding that is a good embedding. Every upward embedding \upmu of~$G_\mu$ that might be extended to a good embedding of~$G$ in which~$e^*$ is incident to the outer face is itself a good embedding. Indeed, it is obvious that the 4-modality of \upmu and the absence of impossible {\em internal} faces are necessary conditions for the extensibility of \upmu to a good embedding~$(\mathcal E,\lambda)$ of~$G$ in which~$e^*$ is incident to the outer face. The following is less immediate, given that, differently from the internal faces, the outer face~$f_{\mu}$ of~$\mathcal E_\mu$ is not a face of~$\mathcal E$.   

\begin{lemma} \label{le:outer-good}
Let~$(\mathcal E,\lambda)$ be a good embedding of~$G$ in which~$e^*$ is incident to the outer face. Let \upmu be the restriction of~$(\mathcal E,\lambda)$ to the vertices and edges of~$G_\mu$. Then the outer face $f_\mu$ of~$\mathcal E_\mu$ is not an impossible face.   
\end{lemma}

\begin{proof}
By \cref{th:characterization}, there exists an upward book embedding~$\Gamma$ that respects~$(\mathcal E, \lambda)$ with~$e^*$ on the outer face. Let~$\Gamma_\mu$ be obtained by removing from~$\Gamma$ all vertices and edges that do not belong to~$G_\mu$. Then~$\Gamma_\mu$ is an upward book embedding of~$G_\mu$ that respects~\upmu. Therefore, \cref{th:characterization} implies that the outer face~$f_\mu$ of~$\mathcal E_\mu$ is not impossible.
\end{proof}

For a~$uv$-external good embedding~\upmu of~$G_\mu$, the possible ``shapes'' of the cycle bounding the outer face~$f_{\mu}$ of~$\mathcal E_{\mu}$ can be characterized by the notion of a \emph{shape descriptor}. A shape descriptor is an~$8$-tuple~$\shapeDesc{\tau_l}{\tau_r}{\lambda^u}{\lambda^v}{\rho_l^u}{\rho_r^u}{\rho_l^v}{\rho_r^v}$, defined as follows.  Let the \emph{left outer path}~$P_l$ (resp.\ the \emph{right outer path}~$P_r$) of~$\mathcal E_{\mu}$ be the path obtained by traversing the boundary of~$f_{\mu}$ from~$u$ to~$v$ in clockwise (resp.\ counterclockwise) direction. The \emph{left-turn-number}~$\tau_l$ of~$\mathcal E_{\mu}$ is the sum of the labels assigned by~$\lambda_\mu$ to the angles at the vertices of~$P_l$ (excluding~$u$ and~$v$) in~$f_{\mu}$; the \emph{right-turn-number}~$\tau_r$ of~$\mathcal E_{\mu}$ is defined analogously for~$P_r$.  The values~$\lambda^u$ and~$\lambda^v$ are the labels of the angles at~$u$ and~$v$ in~$f_{\mu}$, respectively. Finally,~$\rho_l^u$ is set to~$\texttt{in}$ or~$\texttt{out}$ depending on whether the edge of~$P_l$ incident to~$u$ is incoming or outgoing at~$u$, respectively; the values~$\rho_r^u$,~$\rho_l^v$, and~$\rho_r^v$ are defined analogously. The values of a shape descriptor are not independent of each other~\cite{DBLP:conf/compgeom/ChaplickGFGRS22,DBLP:conf/gd/ChaplickGFGRS22}. In fact, the values of~$\tau_l$,~$\lambda^u$,~$\lambda^v$, and~$\rho_l^u$ suffice to determine the other four values of the shape descriptor.    
 
We enrich the information provided by a shape descriptor with a second tuple that contains information concerning whether \upmu might be extended to a good embedding of~$G$. The \emph{pbe descriptor} (short for partitioned book embedding descriptor) is a~$10$-tuple~$\shapeDescPUBE{p_l^u}{p_r^u}{p_l^v}{p_r^v}{\chi_l}{\chi_r}{\alpha_l^u}{\alpha_r^u}{\alpha_l^v}{\alpha_r^v}$, which is defined as follows. 
\begin{itemize}
\item The labels~$p_l^u$,~$p_r^u$,~$p_l^v$, and~$p_r^v$ have values~$L$ or~$R$, depending on whether the edge incident to~$u$ in~$P_l$, the edge incident to~$u$ in~$P_r$, the edge incident to~$v$ in~$P_l$, and the edge incident to~$v$ in~$P_r$ belong to~$L$ or~$R$, respectively. 
\item The label~$\chi_l$ is \true if~$P_l$ is a directed path from~$u$ to~$v$ and all its edges belong to~$L$ or if~$P_l$ is a directed path from~$v$ to~$u$ and all its edges belong to~$R$, it is \false otherwise. Similarly,~$\chi_r$ is \true if~$P_r$ is a directed path from~$u$ to~$v$ and all its edges belong to~$R$ or if~$P_r$ is a directed path from~$v$ to~$u$ and all its edges belong to~$L$, it is \false otherwise.
\item The label~$\alpha_l^u$ is \true if~$P_l$ contains a directed path~$P_l^{uw}$ from~$u$ to a vertex~$w\notin \{u,v\}$ such that all the edges of~$P_l^{uw}$ belong to~$L$ and~$\lambda_{\mu}$ assigns a small angle at~$w$ in~$f_{\mu}$, or if~$P_l$ contains a directed path~$P_l^{wu}$ from a vertex~$w\notin \{u,v\}$ to~$u$ such that all the edges of~$P_l^{wu}$ belong to~$R$ and~$\lambda_{\mu}$ assigns a small angle at~$w$ in~$f_{\mu}$; the label~$\alpha_l^u$ is \false otherwise. The labels~$\alpha_r^u$,~$\alpha_l^v$, and~$\alpha_r^v$ are defined analogously, with respect to~$P_r$ rather than~$P_l$ and/or with respect to~$v$ rather than~$u$.
\end{itemize}

The values of the pbe descriptor of a~$uv$-external good embedding \upmu of~$G_\mu$ might depend on each other and on the values of the shape descriptor of \upmu. For example, if~$p_l^u=L$ and~$\chi_l=\true$, then~$p_l^v=L$,~$\rho_l^u=\texttt{out}$, and~$\rho_l^v=\texttt{in}$. 

We call \emph{descriptor pair} of~\upmu the pair~$(\sigma, \omega)$ where~$\sigma$ and~$\omega$ are the shape and pbe descriptors of~\upmu, respectively. The information in a descriptor pair fully describes how a~$uv$-external good embedding of~$G_\mu$ interfaces with the rest of~$G$ for the construction of a good embedding of~$G$ with~$e^*$ on the outer face. That is, consider a good embedding~$(\mathcal E,\lambda)$ of~$G$ in which~$e^*$ is incident to the outer face, let~\upmu be the~$uv$-external good embedding of~$G_\mu$ in~$(\mathcal E,\lambda)$, and let~$(\sigma, \omega)$ be the descriptor pair of~\upmu. Replacing~\upmu with {\em any} other~$uv$-external good embedding of~$G_\mu$ with descriptor pair~$(\sigma, \omega)$ still results in a good embedding of~$G$ in which~$e^*$ is incident to the outer face. Even more, consider a~$uv$-external good embedding~\upmu of~$G_\mu$ with descriptor pair~$(\sigma, \omega)$, let~$\nu$ be a child of~$\mu$ with poles~$u'$ and~$v'$, and let~$(\sigma', \omega')$ be the descriptor pair of the~$u'v'$-external good embedding~$(\mathcal E_\nu,\lambda_\nu)$ of~$G_\nu$ in~\upmu. Replacing~$(\mathcal E_\nu,\lambda_\nu)$  with {\em any} other~$u'v'$-external good embedding of~$G_\nu$ with descriptor pair~$(\sigma', \omega')$ results in a~$uv$-external good embedding of~$G_\mu$ {\em with descriptor pair~$(\sigma, \omega)$}. This allows us, for a node~$\mu$ of~$T$, to only keep track of the descriptor pairs~$(\sigma, \omega)$ that are ``realizable'' by~$G_\mu$, rather than of the actual~$uv$-external good embeddings of~$G_\mu$. Thus, in the following, we show how to compute the \emph{feasible set~$\mathcal F_{\mu}$ of~$\mu$}. This is the set that contains all the descriptor pairs~$(\sigma, \omega)$ such that~$G_{\mu}$ admits a~$uv$-external good embedding with descriptor pair~$(\sigma, \omega)$. We have the following lemma.

\begin{lemma} \label{le:shape-data-structure}
The set~$\mathcal F_{\mu}$ has size~$O(|V(G_\mu)|)$ and can be stored in~$O(|V(G_\mu)|)$~space, so that a query on whether a descriptor pair belongs to~$\mathcal F_{\mu}$ can be answered in~$O(1)$ time.
\end{lemma}

\begin{proof}
An analogous lemma was proved in \cite{DBLP:conf/compgeom/ChaplickGFGRS22,DBLP:conf/gd/ChaplickGFGRS22} for a feasible set containing shape descriptors, rather than descriptor pairs. However, a pbe descriptor can only assume~$O(1)$ many distinct values, hence the size of~$\mathcal F_{\mu}$, as well as the time and space for storing it, and the query time, only change by a multiplicative constant.     
\end{proof}



Our algorithm traverses the SPQ-tree~$T$ of~$G$ bottom-up and computes, for each node~$\mu$ of~$T$, its feasible set~$\mathcal F_\mu$, provided that the feasible sets of its children have already been computed. Eventually, the test concludes that~$G$ admits a good embedding with~$e^*$ on the outer face if and only if the feasible set of the root~$\rho^*$ of~$T$ is non-empty. We now describe how the algorithm deals with each node~$\mu$ of~$T$, based on its type.

\subsection{Q-node} \label{sub:q}

Non-root Q-nodes have a unique upward embedding, from which we can derive the following.

\begin{lemma}\label{lem:Q_node_general}
	Let~$\mu$ be a non-root Q-node of~$T$. The feasible set~$\mathcal{F}_\mu$ of~$\mu$ can be computed in~$O(1)$ time.
\end{lemma}

\begin{proof}
Since~$G_\mu$ has a unique upward embedding~$\mathcal{E}_\mu$, it has a unique descriptor pair. Indeed, as noted in \cite{DBLP:conf/compgeom/ChaplickGFGRS22,DBLP:conf/gd/ChaplickGFGRS22}, the shape descriptor of~$\mathcal{E}_\mu$ is either~$\shapeDesc{0}{0}{1}{1}{\texttt{out}}{\texttt{out}}{\texttt{in}}{\texttt{in}}$ if the edge~$(u,v)$ of~$G$ corresponding to~$\mu$ is directed from~$u$ to~$v$ or~$\shapeDesc{0}{0}{1}{1}{\texttt{in}}{\texttt{in}}{\texttt{out}}{\texttt{out}}$ otherwise. Also, the pbe descriptor of~$\mathcal{E}_\mu$ is either~$\shapeDescPUBE{L}{L}{L}{L}{\true}{\false}{\false}{\false}{\false}{\false}$ if the edge~$(u,v)$ is directed from~$u$ to~$v$ and belongs to~$L$, or~$\shapeDescPUBE{R}{R}{R}{R}{\false}{\true}{\false}{\false}{\false}{\false}$ if the edge~$(u,v)$ is directed from~$u$ to~$v$ and belongs to~$R$, or~$\shapeDescPUBE{L}{L}{L}{L}{\false}{\true}{\false}{\false}{\false}{\false}$ if the edge~$(u,v)$ is directed from~$v$ to~$u$ and belongs to~$L$, or~$\shapeDescPUBE{R}{R}{R}{R}{\true}{\false}{\false}{\false}{\false}{\false}$ if the edge~$(u,v)$ is directed from~$v$ to~$u$ and belongs to~$R$. 
\end{proof}

\subsection{S-node} \label{sub:s}

For S-nodes, our algorithm works as follows. Let~$\nu_1$ and~$\nu_2$ be the children of an S-node~$\mu$ in~$T$, let~$n_1$ and~$n_2$ be the number of vertices of~$G_{\nu_1}$ and~$G_{\nu_2}$, and let~$w$ be the unique vertex shared by~$G_{\nu_1}$ and~$G_{\nu_2}$. We combine every descriptor pair~$(\sigma_1,\omega_1)$ in~$\mathcal{F}_{\nu_1}$ with every descriptor pair~$(\sigma_2,\omega_2)$ in~$\mathcal{F}_{\nu_2}$; for every such combination, the algorithm assigns the two angles at~$w$ in the outer face with every possible label in~$\{-1,0,1\}$. Whenever the combination and the assignment result in a descriptor pair~$(\sigma,\omega)$ of a good embedding of~$G_\mu$, the algorithm adds~$(\sigma,\omega)$ to~$\mathcal{F}_\mu$. In order to test whether the combination of descriptor pairs, together with the assignment of labels to the angles at~$w$, results in a descriptor pair~$(\sigma,\omega)$ of a good embedding of~$G_\mu$, we check whether the properties of \cref{th:upward-conditions,th:characterization} are satisfied. This can be done in~$O(1)$ time, as in the following.

\begin{lemma} \label{le:s-node-check}
For every descriptor pair~$(\sigma_1,\omega_1)$ in~$\mathcal{F}_{\nu_1}$, every descriptor pair~$(\sigma_2,\omega_2)$ in~$\mathcal{F}_{\nu_2}$, and every pair of values~$\beta_w,\gamma_w$ in~$\{-1,0,1\}$, it is possible to check in~$O(1)$ time whether there exists a~$uv$-external good embedding \upmu of~$G_\mu$ in which the upward embedding~$(\mathcal E_{\nu_i},\lambda_{\nu_i})$ of~$G_{\nu_i}$ has descriptor pair~$(\sigma_i,\omega_i)$, for~$i=1,2$, and in which the angles at~$w$ in the outer face of~$\mathcal E_{\mu}$ to the left of the left outer path of~$\mathcal E_{\mu}$ and to the right of the right outer path of~$\mathcal E_{\mu}$  are~$\beta_w$ and~$\gamma_w$, respectively. Also, in the positive case, the descriptor pair~$(\sigma,\omega)$ of \upmu can be computed in~$O(1)$ time.
\end{lemma}

\begin{proof}
Assume that~$\nu_1$ is the child with poles~$u$ and~$w$, while~$\nu_2$ is the child with poles~$w$ and~$v$. We need to check whether the combination of the descriptor pairs~$(\sigma_1,\omega_1)$ and~$(\sigma_2,\omega_2)$, together with the assignment of labels~$\beta_w,\gamma_w$ to the angles at~$w$ results in an embedding~\upmu of~$G_\mu$ such that: 
\begin{enumerate}
    \item \upmu is a~$uv$-external upward embedding, in particular we need to check whether~$\lambda_\mu$ is an upward-consistent assignment, see Properties  C1--C3 and~\cref{th:upward-conditions};
    \item \upmu is 4-modal; and
    \item \upmu does not contain any impossible face. 
\end{enumerate}
Since~$(\sigma_1,\omega_1)$ and~$(\sigma_2,\omega_2)$ correspond to a~$uw$-external good embedding of~$G_{\nu_1}$ and to a~$wv$-external good embedding of~$G_{\nu_2}$, respectively, a fail in the above checks can only occur in the elements of~\upmu that are created by joining~$(\sigma_1,\omega_1)$ and~$(\sigma_2,\omega_2)$; thus, we need to check whether the angles incident to~$w$ satisfy the conditions of an upward-consistent assignment and whether~$w$ is 4-modal, and we need to check whether the outer face of~\upmu satisfies the conditions of an upward-consistent assignment and whether it is an impossible face or not. These properties can be checked in~$O(1)$ time, as will described below. 

Let~$(\mathcal E_{\nu_1},\lambda_{\nu_1})$ be a~$uw$-external good embedding of~$G_{\nu_1}$ whose descriptor pair is~$(\sigma_1,\omega_1)$ and let~$(\mathcal E_{\nu_2},\lambda_{\nu_2})$ be a~$wv$-external good embedding of~$G_{\nu_2}$ whose descriptor pair is~$(\sigma_2,\omega_2)$. Denote by~$\tau^1_l, \tau^1_r, \lambda^u_1, \lambda^w_1$ the first four labels in~$\sigma_1$ and by~$\tau^2_l, \tau^2_r, \lambda^w_2, \lambda^v_2$ the first four labels in~$\sigma_2$. Also, let \upmu be obtained from~$(\mathcal E_{\nu_1},\lambda_{\nu_1})$ and~$(\mathcal E_{\nu_2},\lambda_{\nu_2})$ as follows. Each of~$G_{\nu_1}$ and~$G_{\nu_2}$ maintains its embedding,~$\mathcal E_{\nu_1}$ and~$\mathcal E_{\nu_2}$ respectively, in~$\mathcal E_\mu$ and is in the outer face of the other. This completely defines~$\mathcal E_\mu$; note that~$u$ and~$v$ are incident to the outer face of~$\mathcal E_\mu$. Every angle, except for the angles incident to~$w$ in the outer face~$f_\mu$ of~$\mathcal E_\mu$, is also an angle in~$\mathcal E_{\nu_1}$ or~$\mathcal E_{\nu_2}$, and then it maintains the same label it is assigned by~$\lambda_{\nu_1}$ or~$\lambda_{\nu_2}$, respectively. Finally, the angle at~$w$ in~$f_\mu$ to the left of the left outer path of~$\mathcal E_{\mu}$ is assigned label~$\beta_w$ and the angle at~$w$ in~$f_\mu$ to the right of the right outer path of~$\mathcal E_{\mu}$ is assigned label~$\gamma_w$. Then the descriptor pair~$(\sigma,\omega)$ of \upmu can be computed in~$O(1)$ time as follows. 

\begin{itemize}
\item Concerning the shape descriptor~$\sigma$, the left-turn-number~$\tau_l$ of \upmu (first label in~$\sigma$) is equal to~$\tau^1_l+\beta_w+\tau^2_l$. Similarly, the right-turn-number~$\tau_r$ of~\upmu (second label in~$\sigma$) is equal to~$\tau^1_r+\gamma_w+\tau^2_r$. The label~$\lambda^u$ for the angle at~$u$ in~$f_\mu$ (third label in~$\sigma$) is equal to~$\lambda^u_1$, while the label~$\lambda^v$ for the angle at~$v$ in~$f_\mu$ (fourth label in~$\sigma$) is equal to~$\lambda^2_v$. The labels~$\rho^u_l$ and~$\rho^u_r$ (fifth and sixth labels in~$\sigma$) have the same values as the corresponding labels in~$\sigma_1$, and the labels~$\rho^v_l$ and~$\rho^v_r$ (seventh and eighth labels in~$\sigma$) have the same values as the corresponding labels in~$\sigma_2$.    
\item Concerning the pbe descriptor~$\omega$, the labels~$p_l^u$ and~$p_r^u$ (first and second labels in~$\omega$) have the same values as the corresponding labels in~$\omega_1$, while the labels~$p_l^v$ and~$p_r^v$ (third and fourth labels in~$\omega$) have the same values as the corresponding labels in~$\omega_2$. The label~$\chi_l$ (fifth label in~$\omega$) has value~$1$ if and only if all the following hold true: (i) the corresponding labels in~$\omega_1$ and~$\omega_2$ both have value~$1$; (ii)~$\beta_w=0$; and (ii) the labels~$p_l^w$ in~$\omega_1$ and~$\omega_2$ both have value~$L$ or both have value~$R$. The value of the label~$\chi_r$ (sixth label in~$\omega$) is computed similarly.  The label~$\alpha^u_l$ (seventh value in~$\omega$) has value~$1$ if and only if at least one of the following holds true: (i) the corresponding label in~$\omega_1$ has value~$1$; (ii) the label~$\chi_l$ in~$\omega_1$ has value~$1$ and~$\beta_w=-1$; or (iii) the label~$\chi_l$ in~$\omega_1$ has value~$1$,~$\beta_w=0$, the label~$\alpha^w_l$ in~$\omega_2$  has value~$1$, and the labels~$p_l^w$ in~$\omega_1$ and~$\omega_2$ both have value~$L$ or both have value~$R$. The values of the labels~$\alpha_r^u$,~$\alpha_l^v$, and~$\alpha_r^v$ (eighth, ninth, and tenth labels in~$\omega$) are computed similarly. 
\end{itemize}

We now show how it can be tested in~$O(1)$ time whether~$(\sigma,\omega)$ actually corresponds to a~$uv$-external good embedding \upmu of~$G_\mu$. 
\begin{enumerate}
\item We first deal with the properties that make~\upmu an upward embedding. We do not address here the bimodality of~\upmu, as later we will show how to check its 4-modality, which implies its bimodality. 

Property C1 is satisfied by~\upmu for all angles different from~$\beta_w$ and~$\gamma_w$ since it is satisfied by~$(\mathcal E_{\nu_1},\lambda_{\nu_1})$ and~$(\mathcal E_{\nu_2},\lambda_{\nu_2})$. Thus, we check whether~$\beta_w=0$ if and only if the label~$\rho^w_l$ in~$\sigma_1$ is equal to~$\texttt{in}$ and the  label~$\rho^w_l$ in~$\sigma_2$ is equal to~$\texttt{out}$, or vice versa, and similar for~$\gamma_w$. 
    
Property C2 is satisfied by~\upmu for all vertices different from~$w$ since it is satisfied by~$(\mathcal E_{\nu_1},\lambda_{\nu_1})$ and~$(\mathcal E_{\nu_2},\lambda_{\nu_2})$. Thus, we check whether~$\beta_w+\gamma_w=\lambda^w_1+\lambda^w_2-2$. Indeed, By Property C2 for~$(\mathcal E_{\nu_1},\lambda_{\nu_1})$, the label~$\lambda^w_1$ is equal to~$2-\deg_1(w)-\sum \lambda_{\nu_1}(a)$, where~$\deg_1(w)$ is the degree of~$w$ in~$G_{\nu_1}$ and the sum is over all the angles~$a$ at~$w$ in internal faces of~$\mathcal E_{\nu_1}$. A similar equation holds true for~$(\mathcal E_{\nu_2},\lambda_{\nu_2})$, and this gives us that~$\lambda^w_1+\lambda^w_2=4-\deg(w)-\sum \lambda_\mu(a)$, where~$\deg(w)$ is the degree of~$w$ in~$G_\mu$ and the sum is over all the angles~$a$ at~$w$ in internal faces of~$\mathcal E_\mu$. Also, Property C2 is satisfied by \upmu for~$w$ if and only if~$\beta_w+\gamma_w=2-\deg(w)-\sum \lambda_\mu(a)$. Combining the last two equations, we get that Property C2 is satisfied by \upmu for~$w$ if and only if~$\beta_w+\gamma_w=\lambda^w_1+\lambda^w_2-2$.    

Property C3 does not require any further check. Indeed, it is satisfied by~\upmu for all the internal faces of~$\mathcal E_\mu$ since it is satisfied by~$(\mathcal E_{\nu_1},\lambda_{\nu_1})$ and~$(\mathcal E_{\nu_2},\lambda_{\nu_2})$. By Property~C3 for~$(\mathcal E_{\nu_1},\lambda_{\nu_1})$, we have~$\tau^1_l+\tau^1_r+\lambda^u_1+\lambda^w_1=2$. Analogously,~$\tau^2_l+\tau^2_r+\lambda^w_2+\lambda^v_2=2$. Thus,~$\tau^1_l+\tau^1_r+\lambda^u_1+\lambda^w_1+\tau^2_l+\tau^2_r+\lambda^w_2+\lambda^v_2=4$. Property C3 is satisfied by \upmu for~$f_\mu$ if and only if~$\tau_l+\tau_r+\lambda^u+\lambda^v=2$. Since~$\tau_l=\tau^1_l+\tau^2_l+\beta_w$,~$\tau_r=\tau^1_r+\tau^2_r+\gamma_w$,~$\lambda^u=\lambda^u_1$, and~$\lambda^v=\lambda^v_2$, by the previous equations we get that Property C3 is satisfied by \upmu for~$f_\mu$ if and only if~$\beta_w+\gamma_w=\lambda^w_1+\lambda^w_2-2$, which is the same equation that was checked for Property~C2. 


\item The 4-modality of \upmu can be checked as follows. First, every vertex of~$G_\mu$ different from~$w$ is 4-modal in \upmu since it is 4-modal in~$(\mathcal E_{\nu_1},\lambda_{\nu_1})$ or~$(\mathcal E_{\nu_2},\lambda_{\nu_2})$. Now, consider the circular sequence~$[LO, RO, RI, LI]$, where~$L$,~$R$,~$O$, and~$I$ stand for left, right, outgoing, and incoming, respectively. The labels~$\rho^w_l$,~$\rho^w_r$, and~$\lambda^w_1$ in~$\sigma_1$, together with the labels~$p^w_l$ and~$p^w_r$ in~$\omega_1$, define a linear sequence~$\pi_1$ from the circular sequence~$[LO, RO, RI, LI]$ that is ``used'' by the edges and faces incident to~$w$ in~$(\mathcal E_{\nu_1},\lambda_{\nu_1})$. For example, if~$\rho^w_r=\texttt{in}$ and~$\rho^w_l=\texttt{out}$ in~$\sigma_1$, and~$p^w_r=R$ and~$p^w_l=L$ in~$\omega_1$, then the first edge incident to~$w$ in~$\mathcal E_{\nu_1}$, in the clockwise order of the edges incident to~$w$ that starts at the outer face of~$\mathcal E_{\nu_1}$, is red incoming, while the last edge is left outgoing, thus~$(\mathcal E_{\nu_1},\lambda_{\nu_1})$ uses the linear sequence~$\pi_1=[RI, LI, LO]$. The label~$\lambda^w_1$ is used in order to compute~$\pi_1$ only when~$\rho^w_r=\rho^w_l$ in~$\sigma_1$ and~$p^w_l=p^w_r$ in~$\omega_1$. Say, for example, that~$\rho^w_r=\rho^w_l=\texttt{in}$  and~$p^w_l=p^w_r=L$; then if~$\lambda^w_1=1$ in~$\sigma_1$ we have~$\pi_1=[LI]$, while if~$\lambda^w_1=-1$ we have~$\pi_1=[LI, LO, RO, RI, LI]$. A linear portion~$\pi_2$ of the circular sequence~$[LO, RO, RI, LI]$ that is used by the edges incident to~$w$ in~$(\mathcal E_{\nu_2},\lambda_{\nu_2})$ is defined similarly.  We then need to perform the following check. If~$\pi_1$ consists of just one element, we check that such an element is not an internal element of~$\pi_2$. Symmetrically, if~$\pi_2$ consists of just one element, we check that such an element is not an internal element of~$\pi_1$. Finally, if both~$\pi_1$ and~$\pi_2$ have more than one element, then we check whether the first element of~$\pi_1$ does not belong to~$\pi_2$ or is the last element of~$\pi_2$, and whether the first element of~$\pi_2$ does not belong to~$\pi_1$ or is the last element of~$\pi_1$. Then~$w$ is 4-modal in \upmu if and only if the check is successful. 


\item Finally, the absence of impossible faces in \upmu can be checked as follows. First, every face of~$\mathcal E_\mu$ different from~$f_{\mu}$ is not impossible in \upmu since it is not impossible in~$(\mathcal E_{\nu_1},\lambda_{\nu_1})$ and~$(\mathcal E_{\nu_2},\lambda_{\nu_2})$.  In order to test whether~$f_{\mu}$ is impossible, we need to check whether a directed path on the boundary of~$f_{\mu}$, composed of left edges with~$f_{\mu}$ to its left or composed of right edges with~$f_{\mu}$ to its right, and with two small angles at its end-vertices, arises by joining~$\mathcal E_{\nu_1}$ and~$\mathcal E_{\nu_2}$. Since the outer faces of~$\mathcal E_{\nu_1}$ and~$\mathcal E_{\nu_2}$ are not impossible, the directed path typically consists of a directed path in the outer face of~$\mathcal E_{\nu_1}$ and of a directed path in the outer face of~$\mathcal E_{\nu_2}$, joined at~$w$; an exception to this situation is the one in which the directed path is entirely in the outer face of~$\mathcal E_{\nu_1}$ or~$\mathcal E_{\nu_2}$, and it is the small angle at one of its end-vertices (in this case, this is necessarily~$w$) that is created by joining~$\mathcal E_{\nu_1}$ and~$\mathcal E_{\nu_2}$. The part of the directed path that is in the outer face of~$\mathcal E_{\nu_1}$ and that starts at~$w$, if any, might belong to the left or to the right outer path of~$\mathcal E_{\nu_1}$; also, the small angle incident to such a path might occur in the left outer path of~$\mathcal E_{\nu_1}$, in the right outer path of~$\mathcal E_{\nu_1}$, or at~$u$. Similar options are possible for the part of the directed path that is in the outer face of~$\mathcal E_{\nu_2}$. All these properties determine which labels have to be checked in order to test whether~$G_\mu$ contains a directed path that causes~$f_\mu$ to be an impossible face. Formally, we check whether:

\begin{figure}[htb]
    \centering
    \begin{subfigure}{.15\textwidth}\centering
    \includegraphics[page=1]{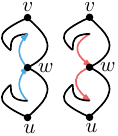}
    \subcaption{\label{fig:check1}}
    \end{subfigure}
    \hfil
    \begin{subfigure}{.15\textwidth}\centering
    \includegraphics[page=2]{figures/Check.pdf}
    \subcaption{\label{fig:check2}}
    \end{subfigure}
    \hfil
    \begin{subfigure}{.15\textwidth}\centering
    \includegraphics[page=3]{figures/Check.pdf}
    \subcaption{\label{fig:check3}}
    \end{subfigure}
    \hfil
    \begin{subfigure}{.15\textwidth}\centering
    \includegraphics[page=4]{figures/Check.pdf}
    \subcaption{\label{fig:check4}}
    \end{subfigure}
    \hfil
    \begin{subfigure}{.15\textwidth}\centering
    \includegraphics[page=5]{figures/Check.pdf}
    \subcaption{\label{fig:check5}}
    \end{subfigure}
    \hfil
    \begin{subfigure}{.15\textwidth}\centering
    \includegraphics[page=6]{figures/Check.pdf}
    \subcaption{\label{fig:check6}}
    \end{subfigure}
    \hfil
    \begin{subfigure}{.15\textwidth}\centering
    \includegraphics[page=7]{figures/Check.pdf}
    \subcaption{\label{fig:check7}}
    \end{subfigure}
    \hfil
    \begin{subfigure}{.15\textwidth}\centering
    \includegraphics[page=8]{figures/Check.pdf}
    \subcaption{\label{fig:check8}}
    \end{subfigure}
    \hfil
    \begin{subfigure}{.15\textwidth}\centering
    \includegraphics[page=9]{figures/Check.pdf}
    \subcaption{\label{fig:check9}}
    \end{subfigure}
    \hfil
    \begin{subfigure}{.15\textwidth}\centering
    \includegraphics[page=10]{figures/Check.pdf}
    \subcaption{\label{fig:check10}}
    \end{subfigure}
    \hfil
    \begin{subfigure}{.15\textwidth}\centering
    \includegraphics[page=11]{figures/Check.pdf}
    \subcaption{\label{fig:check11}}
    \end{subfigure}
    \hfil
    \begin{subfigure}{.15\textwidth}\centering
    \includegraphics[page=12]{figures/Check.pdf}
    \subcaption{\label{fig:check12}}
    \end{subfigure}
    \hfil
    \begin{subfigure}{.15\textwidth}\centering
    \includegraphics[page=13]{figures/Check.pdf}
    \subcaption{\label{fig:check13}}
    \end{subfigure}
    \hfil
    \begin{subfigure}{.15\textwidth}\centering
    \includegraphics[page=14]{figures/Check.pdf}
    \subcaption{\label{fig:check14}}
    \end{subfigure}
    \hfil
    \begin{subfigure}{.15\textwidth}\centering
    \includegraphics[page=15]{figures/Check.pdf}
    \subcaption{\label{fig:check15}}
    \end{subfigure}
    \hfil
    \begin{subfigure}{.15\textwidth}\centering
    \includegraphics[page=16]{figures/Check.pdf}
    \subcaption{\label{fig:check16}}
    \end{subfigure}
    \hfil
    \begin{subfigure}{.15\textwidth}\centering
    \includegraphics[page=17]{figures/Check.pdf}
    \subcaption{\label{fig:check17}}
    \end{subfigure}
    \hfil
    \begin{subfigure}{.15\textwidth}\centering
    \includegraphics[page=18]{figures/Check.pdf}
    \subcaption{\label{fig:check18}}
    \end{subfigure}
    \caption{Configurations that make~$f_\mu$ impossible (first part). Blue and red edges represent directed paths composed only of left and right edges, respectively. Only the vertices~$u$,~$v$, and~$w$ are explicitly shown, while the closed curves represent the boundaries of~$\mathcal E_{\nu_1}$ (below) and~$\mathcal E_{\nu_2}$ (above).}\label{fig:check}
\end{figure}

\begin{figure}[htb]
    \centering
    \begin{subfigure}{.15\textwidth}\centering
    \includegraphics[page=19]{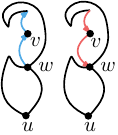}
    \subcaption{\label{fig:check19}}
    \end{subfigure}
    \hfil
    \begin{subfigure}{.15\textwidth}\centering
    \includegraphics[page=20]{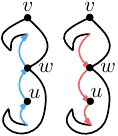}
    \subcaption{\label{fig:check20}}
    \end{subfigure}
    \hfil
    \begin{subfigure}{.15\textwidth}\centering
    \includegraphics[page=21]{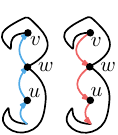}
    \subcaption{\label{fig:check21}}
    \end{subfigure}
    \hfil
    \begin{subfigure}{.15\textwidth}\centering
    \includegraphics[page=22]{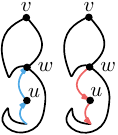}
    \subcaption{\label{fig:check22}}
    \end{subfigure}
    \hfil
    \begin{subfigure}{.15\textwidth}\centering
    \includegraphics[page=23]{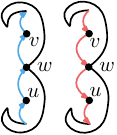}
    \subcaption{\label{fig:check23}}
    \end{subfigure}
    \hfil
    \begin{subfigure}{.15\textwidth}\centering
    \includegraphics[page=24]{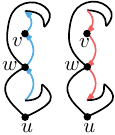}
    \subcaption{\label{fig:check24}}
    \end{subfigure}
    \hfil
    \begin{subfigure}{.15\textwidth}\centering
    \includegraphics[page=25]{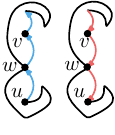}
    \subcaption{\label{fig:check25}}
    \end{subfigure}
    \hfil
    \begin{subfigure}{.15\textwidth}\centering
    \includegraphics[page=26]{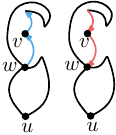}
    \subcaption{\label{fig:check26}}
    \end{subfigure}
    \hfil
    \begin{subfigure}{.15\textwidth}\centering
    \includegraphics[page=27]{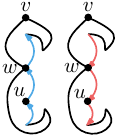}
    \subcaption{\label{fig:check27}}
    \end{subfigure}
    \hfil
    \begin{subfigure}{.15\textwidth}\centering
    \includegraphics[page=28]{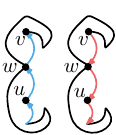}
    \subcaption{\label{fig:check28}}
    \end{subfigure}
    \hfil
    \begin{subfigure}{.15\textwidth}\centering
    \includegraphics[page=29]{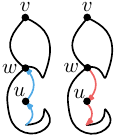}
    \subcaption{\label{fig:check29}}
    \end{subfigure}
    \hfil
    \begin{subfigure}{.15\textwidth}\centering
    \includegraphics[page=30]{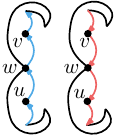}
    \subcaption{\label{fig:check30}}
    \end{subfigure}
    
    \caption{Configurations that make~$f_\mu$ impossible (second part).}\label{fig:check-second}
\end{figure}

\begin{itemize}
    \item The label~$\alpha^w_l$ in~$\omega_1$ is~$1$, the label~$\alpha^w_l$ in~$\omega_2$ is~$1$,~$\beta_w=0$, and the labels~$p_l^w$ in~$\omega_1$ and~$\omega_2$ both have value~$L$ or both have value~$R$ (see \cref{fig:check1}).
    \item The label~$\alpha^w_l$ in~$\omega_1$ is~$1$ and~$\beta_w=-1$ (see \cref{fig:check2}).
    \item The label~$\alpha^w_l$ in~$\omega_2$ is~$1$ and~$\beta_w=-1$ (see \cref{fig:check3}).
    \item The label~$\alpha^w_r$ in~$\omega_1$ is~$1$, the label~$\alpha^w_r$ in~$\omega_2$ is~$1$,~$\gamma_w=0$, and the labels~$p_r^w$ in~$\omega_1$ and~$\omega_2$ both have value~$L$ or both have value~$R$ (see \cref{fig:check4}).    
    \item The label~$\alpha^w_r$ in~$\omega_1$ is~$1$ and~$\gamma_w=-1$ (see \cref{fig:check5}).
    \item The label~$\alpha^w_r$ in~$\omega_2$ is~$1$ and~$\gamma_w=-1$ (see \cref{fig:check6}).
    \item The label~$\alpha^w_l$ in~$\omega_1$ is~$1$,~$\beta_w=0$, the label~$\chi_l$ in~$\omega_2$ is~$1$, the label~$\lambda^v$ in~$\sigma_2$ is~$-1$, and the labels~$p_l^w$ in~$\omega_1$ and~$\omega_2$ both have value~$L$ or both have value~$R$ (see \cref{fig:check7}).
    \item The label~$\alpha^w_l$ in~$\omega_2$ is~$1$,~$\beta_w=0$, the label~$\chi_l$ in~$\omega_1$ is~$1$, the label~$\lambda^u$ in~$\sigma_1$ is~$-1$, and the labels~$p_l^w$ in~$\omega_1$ and~$\omega_2$ both have value~$L$ or both have value~$R$ (see \cref{fig:check8}).    
    \item The label~$\chi_l$ in~$\omega_1$ is~$1$, the label~$\lambda^u$ in~$\sigma_1$ is~$-1$,~$\beta_w=0$, the label~$\chi_l$ in~$\omega_2$ is~$1$, the label~$\lambda^v$ in~$\sigma_2$ is~$-1$, and the labels~$p_l^w$ in~$\omega_1$ and~$\omega_2$ both have value~$L$ or both have value~$R$ (see \cref{fig:check9}).
    \item The label~$\alpha^w_r$ in~$\omega_1$ is~$1$,~$\gamma_w=0$, the label~$\chi_r$ in~$\omega_2$ is~$1$, the label~$\lambda^v$ in~$\sigma_2$ is~$-1$, and the labels~$p_r^w$ in~$\omega_1$ and~$\omega_2$ both have value~$L$ or both have value~$R$ (see \cref{fig:check10}).
    \item The label~$\alpha^w_r$ in~$\omega_2$ is~$1$,~$\gamma_w=0$, the label~$\chi_r$ in~$\omega_1$ is~$1$, the label~$\lambda^u$ in~$\sigma_1$ is~$-1$, and the labels~$p_r^w$ in~$\omega_1$ and~$\omega_2$ both have value~$L$ or both have value~$R$ (see \cref{fig:check11}).
    \item The label~$\chi_r$ in~$\omega_1$ is~$1$, the label~$\lambda^u$ in~$\sigma_1$ is~$-1$,~$\gamma_w=0$, the label~$\chi_r$ in~$\omega_2$ is~$1$, the label~$\lambda^v$ in~$\sigma_2$ is~$-1$, and the labels~$p_r^w$ in~$\omega_1$ and~$\omega_2$ both have value~$L$ or both have value~$R$ (see \cref{fig:check12}).
    \item The label~$\chi_l$ in~$\omega_1$ is~$1$, the label~$\lambda^u$ in~$\sigma_1$ is~$-1$, and~$\beta_w=-1$ (see \cref{fig:check13}). 
    \item The label~$\chi_l$ in~$\omega_2$ is~$1$, the label~$\lambda^v$ in~$\sigma_2$ is~$-1$, and~$\beta_w=-1$ (see \cref{fig:check14}). 
    \item The label~$\chi_r$ in~$\omega_1$ is~$1$, the label~$\lambda^u$ in~$\sigma_1$ is~$-1$, and~$\gamma_w=-1$ (see \cref{fig:check15}). 
    \item The label~$\chi_r$ in~$\omega_2$ is~$1$, the label~$\lambda^v$ in~$\sigma_2$ is~$-1$, and~$\gamma_w=-1$ (see \cref{fig:check16}). 
    \item The label~$\alpha^w_l$ in~$\omega_1$ is~$1$,~$\beta_w=0$, the label~$\chi_l$ in~$\omega_2$ is~$1$, the label~$\lambda^v$ in~$\sigma_2$ is~$0$, the label~$\alpha^v_r$ in~$\omega_2$ is~$1$, and the labels~$p_l^w$ in~$\omega_1$ and~$p_l^w$ and~$p_r^v$ in~$\omega_2$ all have value~$L$ or all have value~$R$ (see \cref{fig:check17}).
    \item The label~$\chi_l$ in~$\omega_1$ is~$1$, the label~$\lambda^u$ in~$\sigma_1$ is~$-1$,~$\beta_w=0$, the label~$\chi_l$ in~$\omega_2$ is~$1$, the label~$\lambda^v$ in~$\sigma_2$ is~$0$, the label~$\alpha^v_r$ in~$\omega_2$ is~$1$, and the labels~$p_l^w$ in~$\omega_1$ and~$p_l^w$ and~$p_r^v$ in~$\omega_2$ all have value~$L$ or all have value~$R$ (see \cref{fig:check18}).
   \item~$\beta_w=-1$, the label~$\chi_l$ in~$\omega_2$ is~$1$, the label~$\lambda^v$ in~$\sigma_2$ is~$0$, the label~$\alpha^v_r$ in~$\omega_2$ is~$1$, and the labels~$p_l^w$ and~$p_r^v$ in~$\omega_2$ both have value~$L$ or both have value~$R$ (see \cref{fig:check19}).
   \item The label~$\alpha^w_l$ in~$\omega_2$ is~$1$,~$\beta_w=0$, the label~$\chi_l$ in~$\omega_1$ is~$1$, the label~$\lambda^u$ in~$\sigma_1$ is~$0$, the label~$\alpha^u_r$ in~$\omega_1$ is~$1$, and the labels~$p_l^w$ in~$\omega_2$ and~$p_l^w$ and~$p_r^u$ in~$\omega_1$ all have value~$L$ or all have value~$R$ (see \cref{fig:check20}).
    \item The label~$\chi_l$ in~$\omega_2$ is~$1$, the label~$\lambda^v$ in~$\sigma_2$ is~$-1$,~$\beta_w=0$, the label~$\chi_l$ in~$\omega_1$ is~$1$, the label~$\lambda^u$ in~$\sigma_1$ is~$0$, the label~$\alpha^u_r$ in~$\omega_1$ is~$1$, and the labels~$p_l^w$ in~$\omega_2$ and~$p_l^w$ and~$p_r^u$ in~$\omega_1$ all have value~$L$ or all have value~$R$ (see \cref{fig:check21}).   
   \item~$\beta_w=-1$, the label~$\chi_l$ in~$\omega_1$ is~$1$, the label~$\lambda^u$ in~$\sigma_1$ is~$0$, the label~$\alpha^u_r$ in~$\omega_1$ is~$1$, and the labels~$p_l^w$ and~$p_r^u$ in~$\omega_1$ both have value~$L$ or both have value~$R$ (see \cref{fig:check22}).
   \item The label~$\chi_l$ in~$\omega_1$ is~$1$, the label~$\lambda^u$ in~$\sigma_1$ is~$0$, the label~$\alpha^u_r$ in~$\omega_1$ is~$1$,~$\beta_w=0$, the label~$\chi_l$ in~$\omega_2$ is~$1$, the label~$\lambda^v$ in~$\sigma_2$ is~$0$, the label~$\alpha^v_r$ in~$\omega_2$ is~$1$, and the labels~$p_l^w$ and~$p_r^u$ in~$\omega_1$ and~$p_l^w$ and~$p_r^v$ in~$\omega_2$ all have value~$L$ or all have value~$R$ (see \cref{fig:check23}).
    \item The label~$\alpha^w_r$ in~$\omega_1$ is~$1$,~$\gamma_w=0$, the label~$\chi_r$ in~$\omega_2$ is~$1$, the label~$\lambda^v$ in~$\sigma_2$ is~$0$, the label~$\alpha^v_l$ in~$\omega_2$ is~$1$, and the labels~$p_r^w$ in~$\omega_1$ and~$p_r^w$ and~$p_l^v$ in~$\omega_2$ all have value~$L$ or all have value~$R$ (see \cref{fig:check24}).
    \item The label~$\chi_r$ in~$\omega_1$ is~$1$, the label~$\lambda^u$ in~$\sigma_1$ is~$-1$,~$\gamma_w=0$, the label~$\chi_r$ in~$\omega_2$ is~$1$, the label~$\lambda^v$ in~$\sigma_2$ is~$0$, the label~$\alpha^v_l$ in~$\omega_2$ is~$1$, and the labels~$p_r^w$ in~$\omega_1$ and~$p_r^w$ and~$p_l^v$ in~$\omega_2$ all have value~$L$ or all have value~$R$ (see \cref{fig:check25}).
   \item~$\gamma_w=-1$, the label~$\chi_r$ in~$\omega_2$ is~$1$, the label~$\lambda^v$ in~$\sigma_2$ is~$0$, the label~$\alpha^v_l$ in~$\omega_2$ is~$1$, and the labels~$p_r^w$ and~$p_l^v$ in~$\omega_2$ both have value~$L$ or both have value~$R$ (see \cref{fig:check26}).
   \item The label~$\alpha^w_r$ in~$\omega_2$ is~$1$,~$\gamma_w=0$, the label~$\chi_r$ in~$\omega_1$ is~$1$, the label~$\lambda^u$ in~$\sigma_1$ is~$0$, the label~$\alpha^u_l$ in~$\omega_1$ is~$1$, and the labels~$p_r^w$ in~$\omega_2$ and~$p_r^w$ and~$p_l^u$ in~$\omega_1$ all have value~$L$ or all have value~$R$ (see \cref{fig:check27}).
    \item The label~$\chi_r$ in~$\omega_2$ is~$1$, the label~$\lambda^v$ in~$\sigma_2$ is~$-1$,~$\gamma_w=0$, the label~$\chi_r$ in~$\omega_1$ is~$1$, the label~$\lambda^u$ in~$\sigma_1$ is~$0$, the label~$\alpha^u_l$ in~$\omega_1$ is~$1$, and the labels~$p_r^w$ in~$\omega_2$ and~$p_r^w$ and~$p_l^u$ in~$\omega_1$ all have value~$L$ or all have value~$R$ (see \cref{fig:check28}).   
   \item~$\gamma_w=-1$, the label~$\chi_r$ in~$\omega_1$ is~$1$, the label~$\lambda^u$ in~$\sigma_1$ is~$0$, the label~$\alpha^u_l$ in~$\omega_1$ is~$1$, and the labels~$p_r^w$ and~$p_l^u$ in~$\omega_1$ both have value~$L$ or both have value~$R$ (see \cref{fig:check29}).
   \item The label~$\chi_r$ in~$\omega_1$ is~$1$, the label~$\lambda^u$ in~$\sigma_1$ is~$0$, the label~$\alpha^u_l$ in~$\omega_1$ is~$1$,~$\gamma_w=0$, the label~$\chi_r$ in~$\omega_2$ is~$1$, the label~$\lambda^v$ in~$\sigma_2$ is~$0$, the label~$\alpha^v_l$ in~$\omega_2$ is~$1$, and the labels~$p_r^w$ and~$p_l^u$ in~$\omega_1$ and~$p_r^w$ and~$p_l^v$ in~$\omega_2$ all have value~$L$ or all have value~$R$ (see \cref{fig:check30}).   
\end{itemize}
If all the above checks fail, then~$f_{\mu}$ is not impossible and hence we conclude that~$(\sigma,\omega)$ corresponds to a~$uv$-external good embedding \upmu.  
\end{enumerate}
This concludes the proof of the lemma.
\end{proof} 

\cref{le:s-node-check} allows us to compute the feasible set~$\mathcal{F}_\mu$ of~$\mu$ in time~$O(|\mathcal{F}_{\nu_1}|\cdot|\mathcal{F}_{\nu_2}|)$. The following lemma, which generalizes a similar statement appearing in the extended version of \cite{DBLP:conf/gd/ChaplickGFGRS22}, proves that this sums up to~$O(m^2)$ time over all S-nodes of~$T$, where~$m$ is the number of edges of~$G$, and thus to~$O(n^2)$ time. A function~$f:\mathbb{N}^+\rightarrow \mathbb{R}_{\geq 0}$ is \emph{super-additive} if~$f(\sum_i x_i)\geq \sum_i f(x_i)$. For a node $\mu$ of $T$ whose pertinent graph $G_\mu$ has $n_{\mu}$ vertices and $m_\mu$ edges, we have $n_\mu \in O(m_\mu)$ and, by \cref{le:shape-data-structure}, we have~$|\mathcal{F}_{\mu}|\in O(n_\mu)$, thus~$|\mathcal{F}_{\mu}|\in O(m_\mu)$. Hence, there exist  positive constants~$p$ and~$q$ with $p\geq q$ such that, for any node $\mu$ of $T$ whose pertinent graph $G_\mu$ has $m_\mu$ edges, we have $|\mathcal{F}_{\mu}|\leq f(m_\mu):=p\cdot m_\mu -q$. Note that the function $f(m_\mu):=p\cdot m_\mu -q$ is indeed super-additive, given that $q>0$. Also, we have that~$f(1)\geq 0$, given that $p\geq q$. We have the following.

\begin{lemma}\label{le:the-most-important-lemma-in-history}
Let~$f:\mathbb{N}^+\rightarrow \mathbb{R}_{\geq 0}$ be a super-additive function such that~$f(1)\geq 0$. Then~$\sum_\mu \left ( f(m_{\nu_1})\cdot f(m_{\nu_2})\right ) \in O\big( (f(m))^2 \big)$, where the sum is taken over all S-nodes~$\mu$ of~$T$, and~$\nu_1$ and~$\nu_2$ are the children of~$\mu$, where~$G_{\nu_1}$ and~$G_{\nu_2}$ have~$m_{\nu_1}$ and~$m_{\nu_2}$ edges, respectively. 
\end{lemma}

\begin{proof}
For a node~$\tau$ of~$T$ (not necessarily an S-node), let $m_\tau$ be the number of edges in $G_\tau$, and let~$S_\tau:=\sum_\mu \left ( f(m_{\nu_1})\cdot f(m_{\nu_2})\right )$, where the sum is taken over all S-nodes~$\mu$ in the subtree of~$T$ rooted at~$\tau$. We prove that~$S_\tau\leq 4\big (f(m_{\tau})\big )^2$. The proof proceeds bottom-up on~$T$. 
\begin{itemize}
\item If~$\tau$ is a leaf, then~$S_\tau =0$ since the subtree of~$T$ rooted at~$\tau$ contains no S-node. Also,~$4\big (f(m_{\tau})\big)^2\geq 0$, given that~$m_{\tau}=1$ and~$f(1)\geq 0$. Hence, the inequality holds. 
\item If~$\tau$ is an S-node with children~$\tau_1$ and~$\tau_2$, where $G_{\tau_1}$ and $G_{\tau_2}$ have $m_{\tau_1}$ and $m_{\tau_2}$ edges, respectively, then~$S_\tau=f(m_{\tau_1})f(m_{\tau_2})+S_{\tau_1}+S_{\tau_2}\leq 4 \big(f(m_{\tau_1})\big)^2+f(m_{\tau_1})f(m_{\tau_2})+4\big(f(m_{\tau_2})\big)^2
\leq 4\big(f(m_{\tau_1})\big)^2+8f(m_{\tau_1})f(m_{\tau_2})+4\big(f(m_{\tau_2})\big)^2
=4(f(m_{\tau_1})+f(m_{\tau_2}))^2\leq 4\big(f(m_{\tau})\big)^2$
 where the last inequality exploits the fact that~$m_{\tau_1}+m_{\tau_2}= m_{\tau}$ and that~$f$ is a super-additive function.
\item Finally, if~$\tau$ is a P-node with~$k\geq 2$ children~$\tau_1$,~$\tau_2$, \dots~$\tau_k$, we have that~$S_{\tau}=\sum_{i=1}^{k}S_{\tau_i}\leq \sum_{i=1}^{k}4\big(f(m_{\tau_i})\big)^2 \leq 4\Big(\sum_{i=1}^{k} f(m_{\tau_i}) \Big)^2\leq 4\big(f(m_{\tau})\big)^2$, where the second inequality is the Cauchy-Schwartz inequality, and last inequality again exploits the fact that~$\sum_{i=1}^{k}m_{\tau_i}=m_\tau$ and that~$f$ is a super-additive function.
 \end{itemize}
The upper bound on~$S_\tau$, applied to the root of~$T$, proves the statement of the lemma.
\end{proof}

We thus get the following.

\begin{lemma}\label{lem:S_node_general}
Let~$\mu$ be an S-node of~$T$ with children~$\nu_1$ and~$\nu_2$, and let~$n_1$ and~$n_2$ be the number of nodes of~$G_{\nu_1}$ and~$G_{\nu_2}$. Given the feasible sets~$\mathcal{F}_{\nu_1}$ and~$\mathcal{F}_{\nu_2}$ of~$\nu_1$ and~$\nu_2$, respectively, the feasible set~$\mathcal{F}_\mu$ of~$\mu$ can be computed in 
$O(n_1 n_2)$ time. This sums up to~$O(n^2)$ time over all S-nodes of~$T$. 
\end{lemma}

\subsection{P-node} \label{sub:P}

Suppose next that~$\mu$ is a P-node. Differently from the case of S-nodes, we cannot just combine the descriptor pairs in the feasible sets of the children~$\nu_1,\dots,\nu_k$ of~$\mu$, as~$k$ might be large, and hence the number of combinations might be super-polynomial. Also, even if a descriptor pair were chosen for each child of~$\mu$, the number of permutations of the children of~$\mu$ might be super-polynomial; the choice of the permutation affects the descriptor pair of the resulting~$uv$-external good embedding of~$G_\mu$. Instead, our algorithm considers every possible descriptor pair that {\em might} describe a~$uv$-external good embedding of~$G_\mu$ and tests whether it belongs to~$\mathcal F_{\mu}$ or not. This is formalized as follows. A set~$\mathcal U_n$ of descriptor pairs is \emph{$n$-universal} if it satisfies the following properties.


\begin{itemize}
\item First, for every~$h$-vertex biconnected partitioned directed partial~$2$-tree~$H$ with $h\leq n$, for every two vertices~$u_H$ and~$v_H$ of~$H$, and for every~$u_Hv_H$-external good embedding~$(\mathcal E_H,\lambda_H)$ of~$H$, the descriptor pair of~$(\mathcal E_H,\lambda_H)$ is in~$\mathcal U_n$. 
\item Second, for every descriptor pair~$(\sigma,\omega)$ in~$\mathcal U_n$, there exists a biconnected partitioned directed  partial~$2$-tree~$H$ that contains two vertices~$u_H$ and~$v_H$, and that admits a~$u_Hv_H$-external good embedding~$(\mathcal E_H,\lambda_H)$ with descriptor pair~$(\sigma,\omega)$.  
\end{itemize}

Note that the feasible set~$\mathcal F_\mu$ of~$\mu$ is a subset of~$\mathcal U_n$. We observe the following:

\begin{lemma} \label{le:universal}
An~$n$-universal set~$\mathcal U_n$ of descriptor pairs with~$|\mathcal U_n|\in O(n)$ can be constructed in~$O(n)$ time. 
\end{lemma}	

\begin{proof}
We describe how to construct~$\mathcal U_n$. We start by considering each~$4$-tuple of values~$\tau_l\in [-n+2,n-2]$,~$\lambda^u,\lambda^v\in \{-1,0,1\}$, and~$\rho_l^u\in \{\texttt{in},\texttt{out}\}$. For each such~$4$-tuple, we construct a shape descriptor~$\sigma=\shapeDesc{\tau_l}{\tau_r}{\lambda^u}{\lambda^v}{\rho_l^u}{\rho_r^u}{\rho_l^v}{\rho_r^v}$, where~$\tau_r=2-\tau_l-\lambda^u-\lambda^v$, where~$\rho_r^u\in \{\texttt{in},\texttt{out}\}$ has the same value as~$\rho_l^u$ if~$\lambda^u\in \{-1,1\}$ and different value if~$\lambda^u=0$, where~$\rho_l^v\in \{\texttt{in},\texttt{out}\}$ has the same value as~$\rho_l^u$ if~$\tau_l$ is odd and different value if~$\tau_l$ is even, and where~$\rho_r^v\in \{\texttt{in},\texttt{out}\}$ has the same value as~$\rho_l^v$ if~$\lambda^v\in \{-1,1\}$ and different value if~$\lambda^v=0$. For each constructed shape descriptor~$\sigma$, we consider each~$10$-tuple~$\omega$ of values~$p^u_l,p^u_r,p^v_l,p^v_r\in \{L, R\}$, and~$\chi_l,\chi_r,\alpha^u_l,\alpha^u_r,\alpha^v_l,\alpha^v_r\in \{0,1\}$. Out of the~$2^{10}$ possible tuples~$\omega$, we discard those  which satisfy at least one of the following checks:
\begin{enumerate}
\item~$\chi_l=1$ and~$\tau_l\neq 0$ (indeed,~$\chi_l=1$ requires the left outer path~$P_l$ to be a directed path, thus the angles in the outer face at the internal vertices of~$P_l$ all have to be assigned label~$0$, whereas~$\tau_l\neq 0$ requires at least one of such angles to be assigned a label different from~$0$), or~$\chi_r=1$ and~$\tau_r\neq 0$; 
\item~$\chi_l=1$ and~$p^u_l\neq p^v_l$ (indeed,~$\chi_l=1$ requires the left outer path~$P_l$ to be entirely composed of left edges or entirely composed of right edges, whereas~$p^u_l\neq p^v_l$ implies that~$P_l$ contains a left and a right edge), or~$\chi_r=1$ and~$p^u_r\neq p^v_r$;
\item~$\chi_l=1$ and~$\rho^u_l=\rho^v_l$ (indeed,~$\chi_l=1$ requires  the left outer path~$P_l$ to be a directed path from~$u$ to~$v$ or from~$v$ to~$u$, whereas~$\rho^u_l=\rho^v_l$ requires the edges of~$P_l$ incident to~$u$ and~$v$ to be both incoming into~$u$ and~$v$ or both outgoing from~$u$ and~$v$), or~$\chi_r=1$ and~$\rho^u_r=\rho^v_r$;
\item~$\chi_l=1$,~$\rho^u_l=\texttt{out}$, and~$p^u_l=R$ (indeed,~$\chi_l=1$ and~$\rho^u_l=\texttt{out}$ require the left outer path~$P_l$ to be entirely composed of left edges, whereas~$p^u_l=R$ implies that~$P_l$ contains a right edge), or~$\chi_l=1$,~$\rho^u_l=\texttt{in}$, and~$p^u_l=L$, or~$\chi_r=1$,~$\rho^u_r=\texttt{out}$, and~$p^u_l=L$, or~$\chi_r=1$,~$\rho^u_r=\texttt{in}$, and~$p^u_l=R$;
\item~$\chi_l=1$ and~$\alpha^u_l=1$ (indeed,~$\chi_l=1$ requires the internal vertices of the left outer path~$P_l$ to be incident to flat angles in the outer face, whereas~$\alpha^u_l=1$ requires~$P_l$ to contain an internal vertex that is incident to a small angle in the outer face), or~$\chi_l=1$ and~$\alpha^v_l=1$, or~$\chi_r=1$ and~$\alpha^u_r=1$, or~$\chi_r=1$ and~$\alpha^v_r=1$;
\begin{figure}[htb]
    \centering
    \begin{subfigure}{.14\textwidth}\centering
    \includegraphics[page=1]{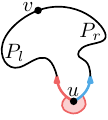}
    \subcaption{\label{fig:Pnodes1}}
    \end{subfigure}
    \hfil
    \begin{subfigure}{.14\textwidth}\centering
    \includegraphics[page=2]{figures/Pnodes.pdf}
    \subcaption{\label{fig:Pnodes2}}
    \end{subfigure}
    \hfil
    \begin{subfigure}{.14\textwidth}\centering
    \includegraphics[page=3]{figures/Pnodes.pdf}
    \subcaption{\label{fig:Pnodes3}}
    \end{subfigure}
    \hfil
    \begin{subfigure}{.14\textwidth}\centering
    \includegraphics[page=4]{figures/Pnodes.pdf}
    \subcaption{\label{fig:Pnodes4}}
    \end{subfigure}
    \hfil
    \begin{subfigure}{.14\textwidth}\centering
    \includegraphics[page=5]{figures/Pnodes.pdf}
    \subcaption{\label{fig:Pnodes5}}
    \end{subfigure}
    \hfil
    \begin{subfigure}{.14\textwidth}\centering
    \includegraphics[page=6]{figures/Pnodes.pdf}
    \subcaption{\label{fig:Pnodes6}}
    \end{subfigure}
    \caption{Forbidden values for the pbe descriptor.}\label{fig:universal}
\end{figure}
\item~$p^u_l=R$,~$p^u_r=L$,~$\rho_l^u=\texttt{out}$, and~$\lambda^u=1$ (indeed, 4-modality, large angle in the outer face, left outer path outgoing with a red edge and right outer path outgoing with a blue edge cannot be all achieved simultaneously, see \cref{fig:Pnodes1}),~$p^u_l=L$,~$p^u_r=R$,~$\rho_l^u=\texttt{in}$, and~$\lambda^u=1$, or~$p^v_r=R$,~$p^v_l=L$,~$\rho_l^v=\texttt{out}$, and~$\lambda^v=1$, or~$p^v_l=R$,~$p^v_r=L$,~$\rho_l^v=\texttt{in}$, and~$\lambda^v=1$;
\item~$\chi_l=1$,~$\lambda^u=-1$, and~$\lambda^v=-1$ (indeed, this results in the outer face to be impossible, see \cref{fig:Pnodes2}), or~$\chi_r=1$,~$\lambda^u=-1$, and~$\lambda^v=-1$;
\item~$\alpha^u_l=1$ and~$\lambda^u=-1$ (indeed, this results in the outer face to be impossible, see \cref{fig:Pnodes3}), or~$\alpha^v_l=1$ and~$\lambda^v=-1$, or~$\alpha^u_r=1$ and~$\lambda^u=-1$, or~$\alpha^v_r=1$ and~$\lambda^v=-1$;
\item~$\alpha^u_l=1$,~$\lambda^u=0$,~$\alpha^u_r=1$, and~$p^u_l=p^u_r$ (indeed, this results in the outer face to be impossible, see \cref{fig:Pnodes4}), or~$\alpha^v_l=1$,~$\lambda^v=0$,~$\alpha^v_r=1$, and~$p^v_l=p^v_r$; 
\item~$\chi_l=1$,~$\lambda^u=0$,~$\lambda^v=-1$,~$\alpha^u_r=1$, and~$p^u_l=p^u_r$ (indeed, this results in the outer face to be impossible, see \cref{fig:Pnodes5}), or~$\chi_l=1$,~$\lambda^v=0$,~$\lambda^u=-1$,~$\alpha^v_r=1$, and~$p^v_l=p^v_r$, or~$\chi_r=1$,~$\lambda^u=0$,~$\lambda^v=-1$,~$\alpha^u_l=1$, and~$p^u_l=p^u_r$, or~$\chi_r=1$,~$\lambda^v=0$,~$\lambda^u=-1$,~$\alpha^v_l=1$, and~$p^v_l=p^v_r$; and
\item~$\chi_l=1$,~$\lambda^u=0$,~$\lambda^v=0$,~$\alpha^u_r=1$,~$\alpha^v_r=1$, and~$p^u_l=p^u_r=p^v_r$ (indeed, this results in the outer face to be impossible, see \cref{fig:Pnodes6}), or~$\chi_r=1$,~$\lambda^u=0$,~$\lambda^v=0$,~$\alpha^u_l=1$,~$\alpha^v_l=1$, and~$p^u_r=p^u_l=p^v_l$.
\end{enumerate}

Each~$10$-tuple that is not discarded is a pbe descriptor~$\omega$ that together with the shape descriptor~$\sigma$ forms a descriptor pair~$(\sigma,\omega)$ that we insert into~$\mathcal U_n$. By construction, we have~$|\mathcal U_n|<(2n-3)\cdot 2^{13}\in O(n)$; also,~$\mathcal U_n$ can clearly be constructed in~$O(n)$ time. We now prove that~$\mathcal U_n$ is indeed~$n$-universal.

\begin{itemize}
    \item First, we prove that, for every~$h$-vertex biconnected partitioned directed partial~$2$-tree~$H$ with $h\leq n$, for every two vertices~$u_H$ and~$v_H$ of~$H$, and for every~$u_Hv_H$-external good embedding~$(\mathcal E_H,\lambda_H)$ of~$H$, the descriptor pair~$(\sigma,\omega)$ of~$(\mathcal E_H,\lambda_H)$ is in~$\mathcal U_n$. By definition, the left-turn-number~$\tau_l$ of~$(\mathcal E_H,\lambda_H)$ is equal to the sum of the labels assigned by~$\lambda_H$ to the angles in the outer face of~$\mathcal E_H$ at the internal vertices of the left outer path~$P_l$ of~$\mathcal E_H$. Since~$H$ has~$h$ vertices, the number of internal vertices of~$P_l$ is at most~$h-2$, hence~$\tau_l\in [-h+2,h-2]$. As mentioned earlier, the values~$\tau_l,\lambda^u,\lambda^v,\rho_l^u$ determine the other values of the labels in~$\sigma$. Since we consider all possible values for~$\lambda^u$,~$\lambda^v$, and~$\rho_l^u$, we indeed generate~$\sigma$. Finally, we generate all~$10$-tuples~$\shapeDescPUBE{p_l^u}{p_r^u}{p_l^v}{p_r^v}{\chi_l}{\chi_r}{\alpha_l^u}{\alpha_r^u}{\alpha_l^v}{\alpha_r^v}$ and only discard tuples that violate the fact that each edge of~$H$ is uniquely oriented, or that each edge of~$H$ is uniquely assigned to a part of the edge set of~$H$, or that~$H$ is acyclic, or that~$(\mathcal E_H,\lambda_H)$ is a~$u_Hv_H$-external good embedding, hence we do not discard~$\omega$, and the descriptor pair~$(\sigma,\omega)$ is indeed added to~$\mathcal U_n$.
    \item Second, we prove that, for every descriptor pair~$(\sigma,\omega)$ in~$\mathcal U_n$, there exists a biconnected partitioned directed partial~$2$-tree~$H$ that contains two vertices~$u$ and~$v$, and that admits a~$uv$-external good embedding~$(\mathcal E_H,\lambda_H)$ with descriptor pair~$(\sigma,\omega)$. 
    
    The underlying graph of our digraph~$H$ is just a cycle, composed of paths~$P_l$ and~$P_r$ connecting~$u$ and~$v$, where~$P_l$ and~$P_r$ are the left and right outer path of~$\mathcal E_H$, respectively. We show how to construct~$P_l$, as the construction of~$P_r$ is analogous. By ``multi-path'' we mean a length-$2$ directed path composed of edges in both parts of the edge set of~$H$. Such a path is assigned label~$0$ by~$\lambda_H$ at the angles incident to its internal vertex.  

    Let~$w_u$ and~$w_v$ be the neighbors of~$u$ and~$v$ in~$P_l$, respectively. We choose the orientation and part for the edge between~$u$ and~$w_u$ according to~$\rho^u_l$ and~$p^u_l$, respectively, and likewise for the edge between~$v$ and~$w_v$, by exploiting~$\rho^v_l$ and~$p^v_l$.
    
    \begin{itemize}
    \item If~$\chi_l=1$, then~$\tau_l=0$, by Check 1, and~$P_l$ is a directed path. We let~$w_u=w_v$ and we let both the angles at~$w_u=w_v$ be assigned label~$0$ by~$\lambda_H$. Note that, by Checks 2, 3, and 4,~$P_l$ is either a directed path~$(u,w_u=w_v,v)$ composed of left edges, or a directed path~$(v,w_u=w_v,u)$ composed of right edges. By Check 5, we have~$\alpha^u_l=\alpha^v_l=0$, hence the definition of~$P_l$ complies with the values of such labels.

\begin{figure}[htb]
    \centering
    \begin{subfigure}{.3\textwidth}\centering
    \includegraphics[page=7]{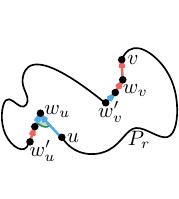}
    \subcaption{\label{fig:Pnodes7}}
    \end{subfigure}
    \hfil
    \begin{subfigure}{.3\textwidth}\centering
    \includegraphics[page=8]{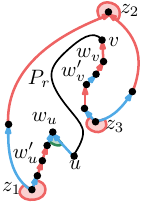}
    \subcaption{\label{fig:Pnodes8}}
    \end{subfigure}
    \begin{subfigure}{.3\textwidth}\centering
    \includegraphics[page=9]{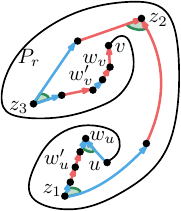}
    \subcaption{\label{fig:Pnodes9}}
    \end{subfigure}
    \caption{Definition of the left outer path~$P_l$ of the planar embedding~$\mathcal E_H$ of the graph~$H$. (a) Multipaths~$P^u_l$ and~$P^v_l$. In this example,~$\alpha^u_l=1$ and~$\alpha^v_l=0$. (b)-(c) Inserting vertices~$z_1,\dots,z_x$ and length-$2$ multi-paths incident to them, with~$S=-1$.  In (b), we have~$\tau_l=2$, hence~$x=|2-(-1)|=3$, while in (c) we have~$\tau_l=-4$, hence~$x=|-4-(-1)|=3$.}\label{fig:universal-works}
\end{figure}
    
    \item If~$\chi_l=0$, we introduce a multi-path~$P^u_l$ between a vertex~$w'_u$ and~$w_u$ and a multi-path~$P^v_l$ between a vertex~$w'_v$ and~$w_v$, see \cref{fig:Pnodes7}. If~$\alpha^u_l=1$, then~$P^u_l$ is oriented so that~$w_u$ is a switch and the partition of its edges is such that~$w_u$ is only incident to edges in one part; we let~$\lambda_H$ assign label~$-1$ to the angle at~$w_u$ in the outer face of~$\mathcal E_H$ and label~$1$ to the angle at~$w_u$ in the internal face of~$\mathcal E_H$. If~$\alpha^u_l=0$, then~$P^u_l$ is oriented so that~$w_u$ is a not a switch; we let~$\lambda_H$ assign label~$0$ to both angles at~$w_u$. Analogously, if~$\alpha^v_l=1$, then~$P^v_l$ is oriented so that~$w_v$ is a switch and the partition of its edges is such that~$w_v$ is only incident to edges in one part; we let~$\lambda_H$ assign label~$-1$ to the angle at~$w_v$ in the outer face of~$\mathcal E_H$ and label~$1$ to the angle at~$w_v$ in the internal face of~$\mathcal E_H$. If~$\alpha^v_l=0$, then~$P^v_l$ is oriented so that~$w_v$ is a not a switch; we let~$\lambda_H$ assign label~$0$ to both angles at~$w_v$. 
    
    Let~$S$ be the sum of the labels assigned to the angles at $w_u$ and $w_v$ in the outer face of~$\mathcal E_H$, note that~$S\in \{-2,-1,0\}$. We introduce~$x=|\tau_l-S|$ vertices~$z_1,\dots,z_x$ in~$P_l$, as in \cref{fig:Pnodes8} and \cref{fig:Pnodes9}. We let~$z_0:=w'_u$ and~$z_{x+1}:=w'_v$. For~$i=0,\dots,x$, we introduce a multi-path between~$z_i$ and~$z_{i+1}$. These multi-paths are oriented so that~$z_0$ and~$z_{x+1}$ are not switches and so that~$z_1,\dots,z_x$ are switches. This orientation is indeed possible, given that~$\tau_l$ is odd if and only if the edges between~$u$ and~$w_u$ and between~$v$ and~$w_v$ are both incoming~$u$ and~$v$ or both outgoing from~$u$ and~$v$, by construction. We let~$\lambda_H$ assign label~$0$ to each angle at~$z_0$ or~$z_{x+1}$; also, for~$i=1,\dots,x$, if~$\tau_l-S>0$, we let~$\lambda_H$ assign label~$1$ to the angle at~$z_i$ in the outer face of~$\mathcal E_H$ and label~$-1$ to the angle at~$z_i$ in the internal face of~$\mathcal E_H$, while if~$\tau_l-S<0$, we let~$\lambda_H$ assign label~$-1$ to the angle at~$z_i$ in the outer face of~$\mathcal E_H$ and label~$1$ to the angle at~$z_i$ in the internal face of~$\mathcal E_H$. Note that, if~$\tau_l-S=0$, then~$x=0$, and a single multi-path between~$w'_u$ and~$w'_v$ is inserted. We assign the edges of the introduced multi-paths to the parts of the edge set of~$H$ so that both edges incident to each vertex among~$z_1,\dots,z_x$ belong the same part.

\end{itemize}
We complete the definition of~$(\mathcal E_H,\lambda_H)$ by letting~$\lambda_H$ assign labels at~$u$ and~$v$ according to~$\lambda^u$ and~$\lambda^v$, respectively: The label of the angle at~$u$ in the outer face of~$\mathcal E_H$ is~$\lambda^u$ and the one at~$u$ in the internal face of~$\mathcal E_H$ is~$0-\lambda^u$, and similar for~$v$.

We prove that~$H$ and~$(\mathcal E_H,\lambda_H)$ satisfy the required properties. 
\begin{itemize}
\item First, note that~$H$ is an acyclic digraph. For the contrary, assume that~$P_l$ is a directed path from~$u$ to~$v$ and~$P_r$ is a directed path from~$v$ to~$u$. By construction,~$P_l$ is a directed path from~$u$ to~$v$ if and only if~$\tau_l=0$, and~$P_r$ is a directed path from~$v$ to~$u$ if and only if~$\tau_r=0$. By construction, since the edge of~$P_l$ incident to~$u$ is outgoing from~$u$, we have~$\rho^u_l=\texttt{out}$; similarly, we have~$\rho^v_l=\texttt{in}$,~$\rho^v_r=\texttt{out}$, and~$\rho^u_r=\texttt{in}$. Furthermore, by construction we have~$\rho^u_l\neq \rho^u_r$ if and only if~$\lambda^u=0$; similarly, we have~$\lambda^v=0$. However, this contradicts the definition of the value~$\tau_r$, which requires~$\tau_r=2-\tau_l-\lambda^u-\lambda^v$. 
\item Second, since~$H$ is a cycle, then clearly~$\mathcal E_H$ is a planar embedding in which~$u$ and~$v$ are incident to the outer face. 
\item Third, we prove that~$(\mathcal E_H,\lambda_H)$ is an upward embedding. Since~$\mathcal E_H$ is a planar embedding, we need to prove that~$\lambda_H$ is an upward-consistent angle assignment, hence it satisfies Properties C1--C3, see also~\cref{th:upward-conditions}. Concerning Properties C1 and C2, note that the degree of every vertex of~$H$ is~$2$, hence the sum of the labels assigned to each vertex has to be~$0$. By construction the internal vertices of the multi-paths, as well as the vertices~$w'_u$ and~$w'_v$ define two flat angles that are both assigned label~$0$ by~$\lambda_H$, while the vertices~$z_1,\dots,z_x$  define two switch angles which are one assigned label~$-1$ and one assigned~$1$ by~$\lambda_H$. Vertex~$w_u$ (the argument for~$w_v$ is analogous) either defines two flat angles that are both assigned label~$0$ by~$\lambda_H$ (if~$\alpha_l^u=0$), or defines two switch angles which are one assigned label~$-1$ and one assigned~$1$ by~$\lambda_H$ (if~$\alpha_l^u=1$). Finally, by construction, vertex~$u$ (the argument for~$v$ is analogous) defines two flat angles which are assigned label~$0$ by~$\lambda_H$ if~$\lambda^u=0$, and  two switch angles which are one assigned label~$-1$ and one assigned~$1$ by~$\lambda_H$ if~$\lambda^u=\pm 1$. Concerning Property C3, by construction the left-turn-number and right-turn-number of~$(\mathcal E_H,\lambda_H)$ are~$\tau_l$ and~$\tau_r$, respectively, while the label of the angles at~$u$ and~$v$ in the outer face of~$\mathcal E_H$ are respectively~$\lambda^u$ and~$\lambda^v$. The fact that the sum of the labels assigned to the angles in the outer face of~$\mathcal E_H$ is~$2$ hence follows from~$\tau_r=2-\tau_l-\lambda^u-\lambda^v$. Since every angle in the internal face of~$\mathcal E_H$ is assigned a label which sums up to~$0$ with the label assigned to the angle at the same vertex in the outer face of~$\mathcal E_H$, it follows that the sum of the labels assigned to the angles in the internal face of~$\mathcal E_H$ is~$-2$.  
\item Fourth, we prove that~$(\mathcal E_H,\lambda_H)$ is 4-modal. Consider any vertex~$z$ of~$P_l$, the argument for the vertices of~$P_r$ is symmetric. If~$z$ is internal to some multi-path, or if~$z=w'_u$, or if~$z=w'_v$, then it has one incoming and one outgoing edge, hence it is 4-modal. If~$z\in\{z_1,\dots,z_x\}$, then it is incident to two edges in the same part of the edge set of~$H$, hence it is 4-modal. If~$z=w_u$ (the argument for the case in which~$z=w_v$ is analogous), then either it has one incoming and one outgoing edge (if~$\alpha_l^u=0$), or it is incident to two edges in the same part of the edge set of~$H$ (if~$\alpha_l^u=1$), hence it is 4-modal. Finally, Check 6 ensures that~$u$ and~$v$ are 4-modal.
\item Finally, we prove that~$(\mathcal E_H,\lambda_H)$ contains no impossible face. By definition, a multi-path contains both left and right edges, hence every maximal directed path in~$H$ that is entirely composed of left edges or entirely composed of right edges contains~$u$ or contains~$v$ (or contains both). Let~$f_I$ and~$f_O$ be the internal and outer face of~$\mathcal E_H$.
\begin{itemize}
\item We prove that~$f_I$ is not impossible. Indeed, we show that the edge connecting~$u$ to its neighbor~$w_u$ is not part of a maximal directed path causing~$f_I$ to be impossible, the proof for the other three edges incident to~$u$ or~$v$ is analogous. If~$\alpha^u_l=1$, then the angle at~$w_u$ in~$f_I$ is labeled~$1$, hence the  maximal directed path containing the edge connecting~$u$ to~$w_u$ does not have a small angle in~$f_I$ at its end-vertex and thus does not cause~$f_I$ to be impossible. If~$\alpha^u_l=0$, then the maximal directed path containing the edge connecting~$u$ to~$w_u$ also contains the multi-path between~$w_u$ and~$w'_u$, hence it is not entirely composed of left edges or entirely composed of right edges and does not cause~$f_I$ to be impossible. 
\item That the edges connecting~$u$ and~$v$ to its neighbors do not belong to any maximal directed path causing~$f_O$ to be impossible is ensured by Checks 7--11. Namely, Check 7 deals with directed paths that have~$u$ and~$v$ as end-vertices, Check 8 deals with directed paths that have one of~$u$ and~$v$ as an end-vertex and do not contain the other one, Check 9 deals with directed paths that have one of~$u$ and~$v$ as an internal vertex and do not contain the other one, Check 10 deals with directed paths that have one of~$u$ and~$v$ as an internal vertex and the other one as an end-vertex, and finally Check 11 deals with directed paths that have~$u$ and~$v$ as internal vertices.
\end{itemize}
\end{itemize}
\end{itemize}
This concludes the proof of the lemma.
\end{proof}

Consider a~$uv$-external good embedding~$(\mathcal E_{\mu},\lambda_\mu)$ of~$G_{\mu}$ with descriptor pair~$(\sigma,\omega)$. For~$i=1,\dots,k$, let~$\mathcal E_{\nu_i}$ be the~$uv$-external good embedding~$(\mathcal E_{\nu_i},\lambda_{\nu_i})$ of~$G_{\nu_i}$ which is the restriction of~$\mathcal E_{\mu}$ to~$G_{\nu_i}$ and let~$(\sigma_i,\omega_i)$ be the descriptor pair of~$(\mathcal E_{\nu_i},\lambda_{\nu_i})$. Assume, without loss of generality up to a change of the indices of the nodes~$\nu_i$, that the clockwise order around~$u$ in~$\mathcal E_{\mu}$ of the pertinent graphs of the children of~$\mu$ is~$G_{\nu_1},\dots,G_{\nu_k}$, where the left outer path of~$\mathcal E_{\nu_1}$ is also the left outer path of~$\mathcal E_{\mu}$ and the right outer path of~$\mathcal E_{\nu_k}$ is also the right outer path of~$\mathcal E_\mu$. The sequence~$\mathcal S_\mu=[(\sigma_1,\omega_1),\dots,(\sigma_k,\omega_k)]$ is then called the \emph{descriptor sequence} of~$(\mathcal E_{\mu},\lambda_\mu)$. The \emph{contracted descriptor sequence} of~$(\mathcal E_{\mu},\lambda_\mu)$ is the sequence of descriptor pairs obtained from~$\mathcal S_\mu$ by identifying consecutive descriptor pairs that are equal; refer to \cref{fig:constracted-sequences} for an example. Finally, the \emph{generating set}~$\mathcal G(\sigma,\omega)$ for the descriptor pair~$(\sigma,\omega)$ is the set of contracted descriptor sequences that a P-node~$\mu$ with poles~$u$ and~$v$ can have in a~$uv$-external good embedding with descriptor pair~$(\sigma,\omega)$. We have the following.

\begin{figure}[tb!]
    \centering
\includegraphics[page=3,width=.7\textwidth]{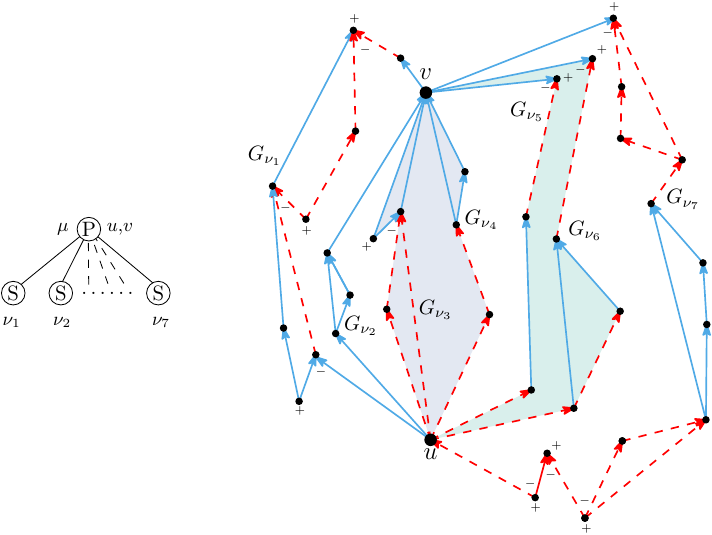}
    \caption{A $uv$-external good embedding $(\mathcal E_\mu, \lambda_\mu)$
    of the pertinent graph $G_\mu$ of a P-node~$\mu$ of the SPQ-tree of a biconnected partitioned directed $2$-tree. To avoid visual cluttering, only the large and small angles on the left and right outer paths of the graphs $G_{\nu_1}, G_{\nu_2},\dots,G_{\nu_{7}}$ are shown; large and small angles are labeled with a $+$ and a $-$ sign, respectively.
    The descriptor pair $(\sigma,\omega)$ of $(\mathcal E_\mu, \lambda_\mu)$ is such that 
    $\sigma = \langle 
    1, 2,0,-1, \texttt{out}, \texttt{in}, \texttt{out}, \texttt{out} \rangle$ and 
    $\omega = \langle L,R,L,L,0,0,1,0,0,0\rangle$.
    The descriptor sequence of $(\mathcal E_\mu, \lambda_\mu)$ is $[(\sigma_1,\omega_1),(\sigma_2,\omega_2),\dots,(\sigma_{7},\omega_{7})]$
    and these are the descriptor pairs of the 
    $uv$-external good embeddings of the graphs $G_{\nu_1}, G_{\nu_2},\dots,G_{\nu_{7}}$ in $(\mathcal E_\mu, \lambda_\mu)$, respectively. The shape descriptors are
    $\sigma_1= \langle1, -1, 1, 1, \texttt{out}, \texttt{out}, \texttt{out}, \texttt{out} \rangle$,
    $\sigma_2= \langle 0, 0, 1, 1, \texttt{out}, \texttt{out}, \texttt{in}, \texttt{in} \rangle$,
    $\sigma_3=\sigma_4=\langle 
    0, 0, 1, 1, \texttt{out}, \texttt{out}, \texttt{in}, \texttt{in} \rangle$,
    $\sigma_5=\sigma_6 = \langle 
    -1,  1, 1, 1, \texttt{out}, \texttt{out}, \texttt{out}, \texttt{out}\rangle$,
    $\sigma_{7} = \langle 
    -2, 2, 1, 1, \texttt{in}, \texttt{in}, \texttt{out}, \texttt{out}\rangle$, and the pbe descriptors are
    $\omega_1 = \langle L,L,L,L,0,0,1,0,0,0 \rangle$,
    $\omega_2 = \langle L,L,L,L,1,0,0,0,0,0 \rangle$,
    $\omega_3=\omega_4 = \langle R,R,L,L,0,0,0,0,0,0 \rangle$,
    $\omega_5=\omega_6=\langle R, R, L, L, 0, 0, 0, 0, 0, 0\rangle $, and
    $\omega_7=\langle R, R, L, L, 0, 0, 1, 0, 0, 0\rangle $. The contracted descriptor sequence of $(\mathcal E_\mu, \lambda_\mu)$ is $[(\sigma_1,\omega_1), (\sigma_2,\omega_2), (\sigma_3,\omega_3), (\sigma_5,\omega_5), (\sigma_7,\omega_7)]$.
    }\label{fig:constracted-sequences}
\end{figure}


\begin{lemma} \label{le:generating}
The generating set~$\mathcal G(\sigma,\omega)$ of a descriptor pair~$(\sigma,\omega)$ has size~$O(1)$ and can be constructed in~$O(1)$ time. Further, each contracted descriptor sequence in~$\mathcal G(\sigma,\omega)$ has length~$O(1)$. 
\end{lemma} 

\begin{proof}
Consider a~$uv$-external good embedding~$(\mathcal E_{\mu},\lambda_\mu)$ of~$G_{\mu}$ with descriptor pair~$(\sigma,\omega)$. Let~$\sigma=\shapeDesc{\tau_l}{\tau_r}{\lambda^{u}}{\lambda^{v}}{\rho_l^{u}}{\rho_r^{u}}{\rho_l^{v}}{\rho_r^{v}}$ and~$\omega=\shapeDescPUBE{p_l^{u}}{p_r^{u}}{p_l^{v}}{p_r^{v}}{\chi_l}{\chi_r}{\alpha_l^{u}}{\alpha_r^{u}}{\alpha_l^{v}}{\alpha_r^{v}}$. For~$i=1,\dots,k$, let~$(\mathcal E_{\nu_i},\lambda_{\nu_i})$ be the~$uv$-external good embedding of~$G_{\nu_i}$ in~$(\mathcal E_{\mu},\lambda_\mu)$. For each pertinent graph~$G_{\nu_i}$ of a child~$\nu_i$ of~$\mu$, let~$(\mathcal E_{\nu_i},\lambda_{\nu_i})$ be its~$uv$-external good embedding in~$(\mathcal E_{\mu},\lambda_\mu)$, let~$(\sigma_i,\omega_i)$ be the descriptor pair of~$(\mathcal E_{\nu_i},\lambda_{\nu_i})$, let~$\sigma_i=\shapeDesc{\tau^i_l}{\tau^i_r}{\lambda^{i,u}}{\lambda^{i,v}}{\rho_l^{i,u}}{\rho_r^{i,u}}{\rho_l^{i,v}}{\rho_r^{i,v}}$ and let~$\omega_i=\shapeDescPUBE{p_l^{i,u}}{p_r^{i,u}}{p_l^{i,v}}{p_r^{i,v}}{\chi^i_l}{\chi^i_r}{\alpha_l^{i,u}}{\alpha_r^{i,u}}{\alpha_l^{i,v}}{\alpha_r^{i,v}}$. 

We first prove that the generating set~$\mathcal G(\sigma,\omega)$ of~$(\sigma,\omega)$ has size~$O(1)$ and that each contracted descriptor sequence in~$\mathcal G(\sigma,\omega)$ has length~$O(1)$. The key point for the proof is that, in~$(\mathcal E_{\mu},\lambda_\mu)$,  there are only constantly-many flat or large angles at~$u$ and~$v$, and constantly-many changes between left and right edges in the circular orders of incident edges around~$u$ and~$v$. The embeddings~$(\mathcal E_{\nu_i},\lambda_{\nu_i})$ whose descriptor pair is not the same as the one of an embedding~$(\mathcal E_{\nu_j},\lambda_{\nu_j})$ next to them in \upmu are found where something ``notable'' happens in~$(\mathcal E_{\mu},\lambda_\mu)$. We deem notable the fact that~$(\mathcal E_{\nu_i},\lambda_{\nu_i})$ is incident to the outer face of~$\mathcal E_{\mu}$, or that the angle at~$u$ or~$v$ in the outer face of~$(\mathcal E_{\nu_i},\lambda_{\nu_i})$ is not large, or that the edges incident to~$u$ on the outer face of~$(\mathcal E_{\nu_i},\lambda_{\nu_i})$ are not in the same part of the edge set of~$G$, or that the edges incident to~$v$ on the outer face of~$(\mathcal E_{\nu_i},\lambda_{\nu_i})$ are not in the same part of the edge set of~$G$, or that an internal face of~$\mathcal E_{\mu}$ that is incident to~$\mathcal E_{\nu_i}$ has an angle at~$u$ or~$v$ that is not small or that is delimited by two edges that are not in the same part of the edge set of~$G$. A contracted descriptor sequence might contain constantly-many additional descriptor pairs ``close to'' where something notable happens. Any such additional descriptor pair represents arbitrarily many graphs, all embedded with that descriptor pair in~\upmu, that can be placed next to each other, without modifying the fact that we overall have a~$uv$-external good embedding with descriptor pair~$(\sigma,\omega)$. 

Formally, a descriptor pair~$(\sigma_i,\omega_i)$ is called \emph{replicable} if it satisfies the following properties. First,~$\lambda^{i,u}=\lambda^{i,u}=1$; second,~$\rho_l^{i,u}=\rho_r^{i,u}$,~$p_l^{i,u}=p_r^{i,u}$; third,~$\rho_l^{i,v}=\rho_r^{i,v}$ and~$p_l^{i,v}=p_r^{i,v}$; and finally,~$\chi^i_l=\chi^i_r=\alpha_l^{i,u}=\alpha_r^{i,u}=\alpha_l^{i,v}=\alpha_r^{i,v}=0$. A descriptor pair that is not replicable is called \emph{special}. First, we note that a contracted descriptor sequence contains $O(1)$ special descriptor pairs. This is immediate for special descriptor pairs that violate~$\lambda^{i,u}=\lambda^{i,u}=1$, or~$\rho_l^{i,u}=\rho_r^{i,u}$, or~$p_l^{i,u}=p_r^{i,u}$, or~$\rho_l^{i,v}=\rho_r^{i,v}$, or~$p_l^{i,v}=p_r^{i,v}$. If a special descriptor pair~$(\sigma_i,\omega_i)$ violates~$\chi^i_l=\alpha_l^{i,u}=\alpha_l^{i,v}=0$, then this forces the face that is incident to the left outer path of~$(\mathcal E_{\nu_i},\lambda_{\nu_i})$ to be either the outer face or to have a non-small angle at $u$ or $v$, otherwise the face would be impossible in \upmu; hence, there are constantly-many of such special descriptor pairs. An analogous argument for the special descriptor pairs that violate~$\chi^i_r=\alpha_r^{i,u}=\alpha_r^{i,v}=0$ shows that altogether there are $O(1)$ special descriptor pairs. Now consider two replicable descriptor pairs~$(\sigma_i,\omega_i)$ and~$(\sigma_j,\omega_j)$ that are consecutive in the contracted shape sequence of \upmu, and let~$f$ be the internal face of~$\mathcal E_\mu$ bounded by the corresponding embeddings. Suppose that~the angles~$\beta^u$ and~$\beta^v$ at~$u$ and~$v$ in~$f$ are small, that the edges incident to $u$ on $f$ are in the same part of the edge set of $G$, and that the edges incident to $v$ on $f$ are in the same part of the edge set of $G$, as these conditions can only be violated constantly-many times. This implies that~$\lambda^{i,u}=\lambda^{i,v}=\lambda^{j,u}=\lambda^{j,v}=1$, that~$\rho_l^{i,u}=\rho_r^{i,u}=\rho_l^{j,u}=\rho_r^{j,u}$, that~$\rho_l^{i,v}=\rho_r^{i,v}=\rho_l^{j,v}=\rho_r^{j,v}$, that~$p_l^{i,u}=p_r^{i,u}=p_l^{j,u}=p_r^{j,u}$, that~$p_l^{i,v}=p_r^{i,v}=p_l^{j,v}=p_r^{j,v}$, and that~$\chi^i_l=\chi^i_r=\alpha_l^{i,u}=\alpha_r^{i,u}=\alpha_l^{i,v}=\alpha_r^{i,v}=\chi^j_l=\chi^j_r=\alpha_l^{j,u}=\alpha_r^{j,u}=\alpha_l^{j,v}=\alpha_r^{j,v}=0$. Also, we have~$\tau^{j}_l=\tau^{i}_l$ and~$\tau^{j}_r=\tau^{i}_r$, where the former comes from~$\tau^{j}_l+\tau^i_r+\beta^u+\beta^v=-2 \Rightarrow \tau^{j}_l=-\tau^i_r$ and from~$\tau^i_l+\tau^i_r+\lambda^{i,u}+\lambda^{i,v}=2 \Rightarrow \tau^i_l=-\tau^i_r$, and the latter comes from~$\tau^{j}_l=-\tau^i_r$ and from~$\tau^{j}_l+\tau^{j}_r+\lambda^{j,u}+\lambda^{j,v}=2\Rightarrow \tau^{j}_l=-\tau^{j}_r$. Hence, the two descriptor pairs~$(\sigma_i,\omega_i)$ and~$(\sigma_j,\omega_j)$ coincide, and thus they would be contracted into a single descriptor pair in the contracted descriptor sequence of \upmu. This concludes the proof that each contracted descriptor sequence in~$\mathcal G(\sigma,\omega)$ has length~$O(1)$. 

In order to prove that~$\mathcal G(\sigma,\omega)$  has size~$O(1)$, note that each element of a descriptor pair can only assume constantly-many different values, with the apparent exception of the left- and right-turn-numbers. However, since the left-turn-number of $(\sigma,\omega)$ is $\tau_l$, and hence so is the left-turn-number of~$(\sigma_i,\omega_i)$, where~$(\mathcal E_{\nu_i},\lambda_{\nu_i})$ contains the left outer path of \upmu, then the left-turn-number of every descriptor pair~$(\sigma_i,\omega_i)$ in the sequence can only be in the interval $[\tau_l-4,\tau_l]$; this is because, for every other descriptor pair~$(\sigma_j,\omega_j)$ in the contracted descriptor sequence, $\tau^j_l=\tau^i_l-\lambda^u_x-\lambda^v_y-2$, where $\lambda^u_x,\lambda^v_y\in \{-1,0,1\}$ are the angles at $u$ and $v$ in the internal face of the graph composed of the left outer path of~$(\mathcal E_{\nu_i},\lambda_{\nu_i})$ and of the left outer path of an embedding~$(\mathcal E_{\nu_j},\lambda_{\nu_j})$ with descriptor pair~$(\sigma_j,\omega_j)$. Similarly, the right-turn-numbers of the descriptor pairs in a contracted descriptor sequence can only assume constantly-many different values. Since each contracted descriptor sequence has length~$O(1)$, this results in constantly-many different contracted descriptor sequences, and hence~$\mathcal G(\sigma,\omega)$ has size~$O(1)$.


We now show how to construct, in~$O(1)$ time, the generating set~$\mathcal G(\sigma,\omega)$. We start by introducing some definitions. 
We define the \emph{edge-type sequence}~$\pi_u$ of~$u$ in~$(\sigma,\omega)$ as follows. Recall that, in clockwise order around~$u$ in a good embedding of~$G_\mu$, the edges incident to~$u$ appear in the clockwise circular order: left outgoing (LO), right outgoing (RO), right incoming (RI), and left incoming (LI). Then~$\pi_u$ defines the linear order in which the four types of edges can be encountered around~$u$ in \upmu, starting at the edge on the left outer path and ending at the edge on the right outer path and moving clockwise. Thus, the first element of~$\pi_u$ is defined by the labels~$\rho_l^{u}$ and~$p_l^{u}$ and then~$\pi_u$ follows the elements in the circular sequence [LO, RO, RI, LI] until the edge type corresponding to the labels~$\rho_r^{u}$ and~$p_r^{u}$ is encountered. This completely defines~$\pi_u$, unless~$\rho_l^{u}=\rho_r^{u}$ and~$p_l^{u}=p_r^{u}$; in this case, if~$\lambda^{u}=1$, then~$\pi_u$ consists of just one element, while if~$\lambda^{u}=-1$, then~$\pi_u$ consists of five elements. The \emph{edge-type sequence}~$\pi_v$ of~$v$ in~$(\sigma,\omega)$ is defined analogously, however the edge types are considered in counter-clockwise order, so that the sequence again starts from the left outer path and ends at the right outer path. 


We construct the generating set~$\mathcal G(\sigma,\omega)$ by repeated augmentations. Throughout the process,~$\mathcal G(\sigma,\omega)$ contains initial subsequences of the contracted descriptor sequences that eventually form the generating set~$\mathcal G(\sigma,\omega)$. An initial subsequence of the edge-type sequence~$\pi_u$ of~$u$ in~$(\sigma,\omega)$ and an initial subsequence of the edge-type sequence~$\pi_v$ of~$v$ in~$(\sigma,\omega)$ are associated to each sequence in~$\mathcal G(\sigma,\omega)$. These represent the edge types that are used by the current subsequence of the contracted descriptor sequence. We call \emph{final} a sequence in~$\mathcal G(\sigma,\omega)$ that belongs to the final generating set~$\mathcal G(\sigma,\omega)$.   


We initialize~$\mathcal G(\sigma,\omega)$ by inserting into it several sequences, each composed of a single descriptor pair~$(\sigma_1,\omega_1)\in \mathcal U_{|V(G_\mu)|}$, where~$\sigma_1=\shapeDesc{\tau^1_l}{\tau^1_r}{\lambda^{1,u}}{\lambda^{1,v}}{\rho_l^{1,u}}{\rho_r^{1,u}}{\rho_l^{1,v}}{\rho_r^{1,v}}$ and~$\omega_1=\shapeDescPUBE{p_l^{1,u}}{p_r^{1,u}}{p_l^{1,v}}{p_r^{1,v}}{\chi^1_l}{\chi^1_r}{\alpha_l^{1,u}}{\alpha_r^{1,u}}{\alpha_l^{1,v}}{\alpha_r^{1,v}}$. The descriptor pairs~$(\sigma_1,\omega_1)$ that compose the sequences to be initially inserted into~$\mathcal G(\sigma,\omega)$ are constructed as follows. 

\begin{itemize}
    \item First, since the left outer path of \upmu coincides with the one of the ``leftmost'' pertinent graph of a child of~$\mu$, the labels~$\tau^1_l$,~$\rho_l^{1,u}$,~$\rho_l^{1,v}$,~$p_l^{1,u}$,~$p_l^{1,v}$,~$\chi^1_l$,~$\alpha_l^{1,u}$, and~$\alpha_l^{1,v}$ are univocally determined by~$(\sigma,\omega)$, namely~$\tau^1_l=\tau_l$,~$\rho_l^{1,u}=\rho_l^{u}$,~$\rho_l^{1,v}=\rho_l^{v}$,~$p_l^{1,u}=p_l^{u}$,~$p_l^{1,v}=p_l^{v}$,~$\chi^1_l=\chi_l$,~$\alpha_l^{1,u}=\alpha_l^{u}$, and~$\alpha_l^{1,v}=\alpha_l^{v}$.
    \item The labels~$\lambda^{1,u}$,~$\rho_r^{1,u}$, and~$p_r^{1,u}$ are defined by the choice of an initial subsequence~$\pi'_u$ (with~$|\pi'_u|\geq 1$) of the edge-type sequence~$\pi_u$ of~$u$ in~$(\sigma,\omega)$. Indeed, the last element of~$\pi'_u$ directly defines~$\rho_r^{1,u}$ and~$p_r^{1,u}$; also, we have~$\lambda^{1,u}=1$ if~$|\pi'_u|\leq 2$ and~$\rho_l^{1,u}=\rho_r^{1,u}$, we have~$\lambda^{1,u}=-1$ if~$|\pi'_u|\geq 4$ and~$\rho_l^{1,u}=\rho_r^{1,u}$, and we have~$\lambda^{1,u}=0$ otherwise; the initial subsequence~$\pi'_u$ is associated to the sequence. 
    \item The labels~$\lambda^{1,v}$,~$\rho_r^{1,v}$, and~$p_r^{1,v}$ are analogously defined by the choice of an initial subsequence~$\pi'_v$ (with~$|\pi'_v|\geq 1$) of~$\pi_v$; also,~$\pi'_v$ is associated to the sequence. 
    \item The label~$\tau^1_r$ is set to~$2-\tau^1_l-\lambda^{1,u}-\lambda^{1,v}$. 
    \item The label~$\chi^1_r$ is set to~$0$ if~$\tau^1_r\neq 0$, or if~$\rho_r^{1,u}=\rho_r^{1,v}$, or if~$p_r^{1,u}\neq p_r^{1,v}$, or if~$\rho^{1,u}_r=\texttt{out}$ and~$p^{1,u}_r=L$, or if~$\rho^{1,u}_r=\texttt{in}$ and~$p^{1,u}_l=R$, or if~$\rho^{1,v}_r=\texttt{in}$ and~$p^{1,v}_r=L$, or if~$\rho^{1,v}_r=\texttt{out}$ and~$p^{1,v}_l=R$, as in these cases the right outer path of a~$uv$-external good embedding with descriptor pair $(\sigma_1,\omega_1)$ is not a directed path composed of left edges with the outer face to its left or composed of right edges with the outer face to its right. Also, the label~$\chi^1_r$ is set to~$0$ if~$\lambda^{1,u}=\lambda^u$,~$\lambda^{1,v}=\lambda^v$, and~$\chi_r=0$. Indeed, if the graph whose~$uv$-external good embedding has descriptor pair~$(\sigma_1,\omega_1)$ is the rightmost in clockwise order around~$u$, this is required by the condition~$\chi_r=0$; if the graph whose~$uv$-external good embedding has descriptor pair~$(\sigma_1,\omega_1)$ has another graph to its right, then having~$\chi^1_r=1$ would imply that the face between such two graphs is impossible, given that the equalities~$\lambda^{1,u}=\lambda^u$ and~$\lambda^{1,v}=\lambda^v$ imply that the angles at~$u$ and~$v$ in such a face have label~$-1$. Finally, if none of the previous conditions applies, we set the label~$\chi^1_r$ in both possible ways.
    \item The label~$\alpha^{1,u}_r$ is set to~$0$ if~$\chi^1_r=1$, or if~$\rho^{1,u}_r=\texttt{out}$ and~$p^{1,u}_r=L$, or if~$\rho^{1,u}_r=\texttt{in}$ and~$p^{1,u}_l=R$, or if~$\lambda^{1,u}=\lambda^u$ and~$\alpha^{u}_r=0$. If none of the previous conditions applies, we set the label~$\alpha^{1,u}_r$ in both possible ways.  The label~$\alpha^{1,v}_r$ is set analogously.
\end{itemize}

This concludes the description of the initialization of~$\mathcal G(\sigma,\omega)$.


Now consider any sequence~$S=[(\sigma_1,\omega_1),\dots,(\sigma_i,\omega_i)]$ currently in~$\mathcal G(\sigma,\omega)$. If the subsequence of $\pi_u$ associated to $S$ coincides with $\pi_u$, if  the subsequence of $\pi_v$ associated to $S$ coincides with $\pi_v$, and if~$\chi^i_r=\chi_r$,~$\alpha_r^{i,u}=\alpha_r^{u}$, and~$\alpha_r^{i,v}=\alpha_r^{v}$, then~$S$ is final. If we also have~$\alpha_r^u=\alpha_r^v=\chi_r=0$ and  $(\sigma_i,\omega_i)$ is not replicable, then we insert in~$\mathcal G(\sigma,\omega)$ also the final sequence~$S'$ obtained by appending to~$S$ a replicable descriptor pair~$(\sigma_{i+1},\omega_{i+1})$ with~$\tau^{i+1}_r=\tau^{i}_r$,~$\rho^{i+1,u}_l=\rho^{i,u}_r$,~$\rho^{i+1,v}_l=\rho^{i,v}_r$,~$p^{i+1,u}_l=p^{i,u}_r$, and~$p^{i+1,v}_l=p^{i,v}_r$ (the other labels are forced by the definition of replicable descriptor pair). 


If~$S$ is not final, then we construct several sequences that replace it in~$\mathcal G(\sigma,\omega)$, each obtained by appending a descriptor pair~$(\sigma_{i+1},\omega_{i+1})$ to~$S$, as described in the following. Let~$\pi_u(S)$ and~$\pi_v(S)$ be the initial subsequences of~$\pi_u$ and~$\pi_v$ associated to~$S$, respectively. We pick, in all possible ways according to the rules described below, two elements in~$\pi_u$ that are the edge types of the edges incident to~$u$ and on the outer face of an embedding with descriptor pair~$(\sigma_{i+1},\omega_{i+1})$; the first picked element either coincides with the last element of~$\pi_u(S)$ or follows it in~$\pi_u$, while the second picked element either coincides with the first picked element or follows it in~$\pi_u$. These elements determine the labels~$\rho^{i+1,u}_l$,~$p^{i+1,u}_l$,~$\rho^{i+1,u}_r$, and~$p^{i+1,u}_r$, for example if the first picked element is LI, then~$\rho^{i+1,u}_l=\texttt{in}$ and~$p^{i+1,u}_l=L$. The elements also determine~$\lambda^{i+1,u}$, namely~$\lambda^{i+1,u}=1$ if the length of the subsequence of~$\pi_u$ composed of the elements between the first and the second picked element is at most~$2$ and~$\rho^{i+1,u}_l=\rho^{i+1,u}_r$,~$\lambda^{i+1,u}=-1$ if such length is at least~$4$ and~$\rho^{i+1,u}_l=\rho^{i+1,u}_r$, and~$\lambda^{i+1,u}=0$ otherwise. Two elements in~$\pi_v$ are picked with analogous rules, and this determines the five labels~$\rho^{i+1,v}_l$,~$p^{i+1,v}_l$,~$\rho^{i+1,v}_r$,~$p^{i+1,u}_r$, and~$\lambda^{i+1,v}$. The labels~$\tau^{i+1}_l$ and~$\tau^{i+1}_r$ are also determined by the choice made. Namely, consider the face~$f$ between the embeddings with descriptor pairs~$(\sigma_i,\omega_i)$ and~$(\sigma_{i+1},\omega_{i+1})$. The angle~$\beta_u$ at~$u$ in~$f$ has label~$-1$ if the length of the subsequence of~$\pi_u$ composed of the elements between the last element in~$\pi_u(S)$ and the first picked element is at most~$2$ and~$\rho^{i,u}_r=\rho^{i+1,u}_l$, has label~$1$ if such length is at least~$4$ and~$\rho^{i,u}_r=\rho^{i+1,u}_l$, and it has label~$0$ otherwise. The label of the  angle~$\beta_v$ at~$v$ in~$f$ is computed analogously. Then we have~$\tau^{i+1}_l=-2-\tau^{i}_r-\beta_u-\beta_v$ and we have~$\tau^{i+1}_r=2-\tau^{i+1}_l-\lambda^{i+1,u}-\lambda^{i+1,v}$. It remains to deal with the labels~$\chi^{i+1}_l$,~$\alpha^{i+1,u}_l$,~$\alpha^{i+1,v}_l$,~$\chi^{i+1}_r$,~$\alpha^{i+1,u}_r$,~$\alpha^{i+1,v}_r$. For these labels we perform all possible choices, completing the definition of~$(\sigma_{i+1},\omega_{i+1})$ in several, constantly-many, ways. We discard those descriptor pairs~$(\sigma_{i+1},\omega_{i+1})$ in which one of the following conditions is satisfied, as each condition implies a contradiction to the meaning of the labels, or that~$f$ is an impossible face:
\begin{itemize}
    \item~$(\chi^{i+1}_l=1 \land (\alpha^{i+1,u}_l=1 \lor \alpha^{i+1,v}_l=1)) \lor (\chi^{i+1}_r=1\land (\alpha^{i+1,u}_r=1 \lor \alpha^{i+1,v}_r=1))$;
    \item~$(\beta_u=-1 \land (\alpha^{i,u}_r=1 \lor \alpha^{i+1,u}_l=1)) \lor (\beta_v=-1\land(\alpha^{i,v}_r=1 \lor \alpha^{i+1,v}_l=1)) \lor (\beta_u=-1\land \beta_v=-1 \land \chi^{i}_r=1) \lor (\beta_u=-1 \land \beta_v=-1 \land \chi^{i+1}_l=1)$;
    \item~$(\beta_u=0 \land \alpha^{i,u}_r=1 \land \alpha^{i+1,u}_l=1 \land p^{i,u}_r=p^{i+1,u}_l) \lor (\beta_v=0\land \alpha^{i,v}_r=1\land \alpha^{i+1,v}_l=1\land p^{i,v}_r=p^{i+1,v}_l)$;
    \item~$(\alpha^{i+1,u}_l=1\land \beta_u=0\land \chi^i_r=1\land \beta_v=-1\land p^{i,u}_r=p^{i+1,u}_l)\lor (\alpha^{i+1,v}_l=1\land \beta_v=0\land \chi^i_r=1\land \beta_u=-1\land p^{i,v}_r=p^{i+1,v}_l) \lor (\alpha^{i,u}_r=1\land \beta_u=0 \land \chi^{i+1}_l=1 \land \beta_v=-1 \land p^{i,u}_r=p^{i+1,u}_l)\lor (\alpha^{i,v}_r=1\land \beta_v=0\land \chi^{i+1}_l=1\land \beta_u=-1\land p^{i,v}_r=p^{i+1,u}_l)$;
    \item~$(\alpha^{i+1,u}_l=1\land \beta_u=0\land \chi^i_r=1\land \beta_v=0\land \alpha^{i+1,v}_l=1 \land p^{i,u}_r=p^{i+1,u}_l=p^{i+1,v}_l)\lor (\alpha^{i,u}_r=1\land \beta_u=0\land \chi^{i+1}_l=1\land \beta_v=0\land \alpha^{i,v}_r=1\land p^{i,u}_r=p^{i+1,u}_l=p^{i,v}_r)$.
\end{itemize}

We also discard those descriptor pairs~$(\sigma_{i+1},\omega_{i+1})$ in which one of the following conditions is satisfied, as each condition implies that the sequence~$[(\sigma_1,\omega_1),\dots,(\sigma_i,\omega_i),(\sigma_{i+1},\omega_{i+1})]$ cannot be completed to a contracted descriptor sequence of a $uv$-external good embedding with descriptor pair~$(\sigma,\omega)$, as the desired values for the labels~$\alpha^u_r$,~$\alpha^v_r$, and~$\chi_r$ cannot be all achieved without creating an impossible face:
\begin{itemize}
    \item~$\alpha^{i+1,u}_r=1$,~$\alpha^u_r=0$,~$\rho^{i+1,u}_r=\rho^{u}_r$, and the final subsequence of~$\pi_u$ starting from the second picked element in~$\pi_u$ contains at most two elements;
    \item~$\alpha^{i+1,v}_r=1$,~$\alpha^v_r=0$,~$\rho^{i+1,v}_r=\rho^{v}_r$, and the final subsequence of~$\pi_v$ starting from the second picked element in~$\pi_v$ contains at most two elements;
    \item~$\chi^{i+1}_r=1$,~$\chi_r=0$,~$\rho^{i+1,u}_r=\rho^{u}_r$,~$\rho^{i+1,v}_r=\rho^{v}_r$, the final subsequence of~$\pi_u$ starting from the second picked element in~$\pi_u$ contains at most two elements, and the final subsequence of~$\pi_v$ starting from the second picked element in~$\pi_v$ contains at most two elements.
\end{itemize}

Finally, we discard a descriptor pair~$(\sigma_{i+1},\omega_{i+1})$ if it is replicable, if it coincides with~$(\sigma_{i},\omega_{i})$, if the first picked element in~$\pi_u$ is the last element of~$\pi_u(S)$, and if the first picked element in~$\pi_v$ is the last element of~$\pi_v(S)$. For every descriptor pair~$(\sigma_{i+1},\omega_{i+1})$ which was not discarded, we insert in~$\mathcal G(\sigma,\omega)$ the sequence $[(\sigma_1,\omega_1),\dots,(\sigma_i,\omega_i),(\sigma_{i+1},\omega_{i+1})]$. 

This concludes the description of the construction of~$\mathcal G(\sigma,\omega)$. Since the length of~$\mathcal G(\sigma,\omega)$ is~$O(1)$ and since at each step of the described construction constantly-many choices and checks are performed, the construction takes overall~$O(1)$ time.  
\end{proof}

Our algorithm to compute the feasible set~$\mathcal F_\mu$ of~$\mu$ is the following. First, we generate an~$n$-universal set~$\mathcal U_{n}$ of descriptor pairs. By definition, we have~$\mathcal F_\mu\subseteq \mathcal U_{n}$. Second, for each descriptor pair~$(\sigma,\omega)$ in~$\mathcal U_{n}$, we construct the generating set~$\mathcal G(\sigma,\omega)$ of~$(\sigma,\omega)$. Third, for each contracted descriptor sequence~$\mathcal C$ in~$\mathcal G(\sigma,\omega)$, we test whether~$\mathcal C$ is \emph{realizable} by~$G_\mu$, i.e., whether there exists a~$uv$-external good embedding of~$G_\mu$ whose contracted descriptor sequence is a subsequence of~$\mathcal C$ containing the first and the last elements of~$\mathcal C$; here ``subsequence'' means a sequence that can be obtained from~$\mathcal C$ by deleting elements. In the positive case, we add~$(\sigma,\omega)$ to~$\mathcal F_\mu$. The running time of the algorithm is as follows. First, the construction of~$\mathcal U_{n}$ takes~$O(n)$ time, by \cref{le:universal}. Second,~$\mathcal U_{n}$ contains~$O(n)$ descriptor pairs, again by \cref{le:universal}. For each descriptor pair~$(\sigma,\omega)$ in~$\mathcal U_{n}$, the construction of the generating set~$\mathcal G(\sigma,\omega)$ takes~$O(1)$ time, by \cref{le:generating}. For each contracted descriptor sequence~$\mathcal C$ in~$\mathcal G(\sigma,\omega)$, the decision on whether~$\mathcal C$ is realizable by~$G_\mu$ takes~$O(k)$ time, by the upcoming \cref{le:realizable}. Since~$\mathcal G(\sigma,\omega)$ contains $O(1)$ contracted descriptor sequences, again by \cref{le:generating}, we have that the construction of~$\mathcal F_\mu$ takes~$O(n k)$ time. 

Before presenting \cref{le:realizable}, we motivate our definition of realizable contracted descriptor sequence. Consider a contracted descriptor sequence~$\mathcal C$. Deciding whether there exists a~$uv$-external good embedding of~$G_\mu$ whose contracted descriptor sequence {\em is}~$\mathcal C$ is not an easy task, from an algorithmic point of view. However, deciding whether there exists a~$uv$-external good embedding of~$G_\mu$ whose contracted descriptor sequence is a subsequence of~$\mathcal C$ containing the first and the last elements of~$\mathcal C$ is algorithmically easier, and equivalent, in order to decide whether~$(\sigma,\omega)$ belongs to~$\mathcal F_\mu$ or not. The last statement is justified by the following lemma. 

\begin{lemma}\label{le:realizability_why}
Let~$\mathcal C$ be a contracted descriptor sequence in~$\mathcal G(\sigma,\omega)$. Any subsequence~$\mathcal C'$ of~$\mathcal C$ containing the first and last elements of~$\mathcal C$ also belongs to~$\mathcal G(\sigma,\omega)$.
\end{lemma}

\begin{proof}
Let~$(\mathcal E_{\mu},\lambda_\mu)$ be a~$uv$-external good embedding of a P-node~$\mu$ with poles~$u$ and~$v$ such that the contracted descriptor sequence of~$(\mathcal E_{\mu},\lambda_\mu)$ is~$\mathcal C$. Let~$(\mathcal E'_{\mu'},\lambda'_{\mu'})$  be the restriction of~$(\mathcal E_{\mu},\lambda_\mu)$ to the graph~$H$ composed of the pertinent graphs~$G_{\nu_i}$ whose~$uv$-external good embedding in~$(\mathcal E_{\mu},\lambda_\mu)$ has a descriptor pair that belongs to~$\mathcal C'$. Note that there are at least two such pertinent graphs~$G_{\nu_i}$. This is obvious if~$\mathcal C$ contains more than one descriptor pair, as in this case the first and last elements of~$\mathcal C$ are distinct and belong to~$\mathcal C'$, by assumption. Otherwise,~$\mathcal C$ contains a single element, which is necessarily~$(\sigma,\omega)$, thus the embedding of every pertinent graph~$G_{\nu_i}$ in~$(\mathcal E_{\mu},\lambda_\mu)$ has descriptor pair~$(\sigma,\omega)$. Also,~$\mu$ has at least two children, since it is a P-node. Hence,~$H$ is the pertinent graph~$G'_{\mu'}$ of a P-node~$\mu'$ of the SPQ-tree of a graph~$G'$, and~$\mathcal C'$ is indeed a contracted descriptor sequence. Also,~$\mathcal C'$ belongs to~$\mathcal G(\sigma,\omega)$, since the descriptor pair of a~$uv$-external good embedding only depends on the first and last elements of its contracted descriptor sequence and these elements are the same in~$\mathcal C$ and~$\mathcal C'$, hence the descriptor pair of a~$uv$-external good embedding of~$G'_{\mu'}$ with~$\mathcal C'$ as contracted descriptor sequence is~$(\sigma,\omega)$.
\end{proof}

We now present the following.

\begin{lemma} \label{le:realizable}
It is possible to test in~$O(k)$ time whether a contracted descriptor sequence~$\mathcal C$ is realizable by~$G_\mu$. 
\end{lemma}

\begin{proof}
Let~$\mathcal C=[(\sigma_1,\omega_1),\dots,(\sigma_\ell,\omega_\ell)]$. By \cref{le:generating}, we have~$\ell \in O(1)$. We create a bipartite graph~$A_{\mathcal C}$ with vertex set~$(\{\nu_1,\dots,\nu_k\},\{(\sigma_1,\omega_1),\dots,(\sigma_\ell,\omega_\ell)\})$ and with an edge between a vertex~$\nu_i$ and a vertex~$(\sigma_j,\omega_j)$ if~$(\sigma_j,\omega_j)$ belongs to~$\mathcal F_{\nu_i}$. This construction takes~$O(k)$ time. Indeed, for each of the~$k$ children~$\nu_i$ of~$\mu$, by \cref{le:shape-data-structure} it can be tested in~$O(1)$ time whether each of the~$\ell \in O(1)$ descriptor pairs~$(\sigma_j,\omega_j)$ belongs to~$\mathcal F_{\nu_i}$. 

We now test whether~$\mathcal C$ is realizable by~$G_\mu$ as follows. First, if a vertex~$\nu_i$ has degree~$0$ in~$A_{\mathcal C}$, then~$\mathcal C$ is not realizable by~$G_\mu$; indeed,~$\mathcal F_{\nu_i}$ contains no descriptor pair among those in~$\mathcal C$, hence any~$uv$-external good embedding of~$G_\mu$ does not yield a contracted descriptor sequence which is a subsequence of~$\mathcal C$. Second, if~$(\sigma_1,\omega_1)$ or~$(\sigma_\ell,\omega_\ell)$ has degree~$0$ in~$A_{\mathcal C}$, or if~$(\sigma_1,\omega_1)$ and~$(\sigma_\ell,\omega_\ell)$ have degree~$1$ and have the same unique neighbor in~$A_{\mathcal C}$, then~$\mathcal C$ is not realizable by~$G_\mu$; indeed, in both cases the children of~$\mu$ cannot be ordered and assigned with a descriptor pair in their feasible sets so that the first child in the ordering is assigned with~$(\sigma_1,\omega_1)$ and the last child with~$(\sigma_\ell,\omega_\ell)$. If we did not conclude that~$\mathcal C$ is not realizable by~$G_\mu$, then~$\mathcal C$ is realizable by~$G_\mu$. Indeed, let~$(\sigma_m,\omega_m)$ be the descriptor pair between~$(\sigma_1,\omega_1)$ and~$(\sigma_\ell,\omega_\ell)$ that has smaller degree in~$A_{\mathcal C}$ and let~$(\sigma_M,\omega_M)$ be the other descriptor pair, with a possible tie broken arbitrarily. We can assign~$(\sigma_m,\omega_m)$ to any neighbor~$\nu_x$ in~$A_{\mathcal C}$; this neighbor exists since the degree of~$(\sigma_1,\omega_1)$ and~$(\sigma_\ell,\omega_\ell)$ in~$A_{\mathcal C}$ is at least one. Then we can assign~$(\sigma_M,\omega_M)$ to a neighbor~$\nu_y\neq \nu_x$; this neighbor obviously exists if the degree of~$(\sigma_M,\omega_M)$ in~$A_{\mathcal C}$ is at least two, and it exists even if the degree of~$(\sigma_M,\omega_M)$ in~$A_{\mathcal C}$ is one, as in this case the unique neighbors of~$(\sigma_1,\omega_1)$ and~$(\sigma_\ell,\omega_\ell)$ are different. We assign each remaining node~$\nu_i$ with any descriptor pair that is a neighbor of~$\nu_i$ in~$A_{\mathcal C}$. Finally, we order the children of~$\mu$ as dictated by~$\mathcal C$: First the children that have been assigned with the descriptor pair~$(\sigma_1,\omega_1)$, then the children that have been assigned with the descriptor pair~$(\sigma_2,\omega_2)$, and so on. Using a~$uv$-external good embedding of~$G_{\nu_i}$ with the assigned descriptor pair, for each child~$\nu_i$ of~$\mu$, results in a~$uv$-external good embedding of~$G_{\mu}$ whose contracted descriptor sequence is a subsequence of~$\mathcal C$ containing~$(\sigma_1,\omega_1)$ and~$(\sigma_\ell,\omega_\ell)$. Since the describe test can be performed in~$O(k+\ell)\in O(k)$ time, the lemma follows.
\end{proof}

We thus get the following.

\begin{lemma}\label{lem:P_node_general}
Let~$\mu$ be a P-node of~$T$ with children~$\nu_1,\dots,\nu_k$. Given the feasible sets~$\mathcal{F}_{\nu_1}$,~$\dots$,~$\mathcal{F}_{\nu_k}$ of~$\nu_1,\dots,\nu_k$, respectively, the feasible set~$\mathcal{F}_\mu$ of~$\mu$ can be computed in~$O(n k)$ time. This sums up to~$O(n^2)$ time over all P-nodes of~$T$. 
\end{lemma}

\subsection{Root} \label{sub:root}

The root~$\rho^*$ of~$T$ corresponds to the entire graph~$G$ and can be treated as a P-node with two children, whose pertinent graphs are the edge~$e^*$ and the pertinent graph of the child~$\sigma^*$ of~$\rho^*$ in~$T$. By \cref{lem:P_node_general}, the feasible set of the root can hence be computed in~$O(n)$ time from the feasible sets of~$\sigma^*$ and of a node representing~$e^*$; the latter can be computed in~$O(1)$ time as in \cref{lem:Q_node_general}. Once the feasible set~$\mathcal F_{\rho^*}$ of~$\rho^*$ has been computed, we have that~$G$ admits a good embedding with~$e^*$ on the outer face if and only if~$\mathcal F_{\rho^*}$ is non-empty. 

By \cref{lem:Q_node_general,lem:S_node_general,lem:P_node_general}, the entire processing of~$T$ takes overall~$O(n^2)$ time. Repeating the test for every possible choice of~$e^*$ leads to the following.

\begin{theorem} \label{th:series-parallel}
Let~$G$ be an~$n$-vertex biconnected partitioned directed partial~$2$-tree. It is possible to test in~$O(n^3)$ time whether~$G$ admits an upward book embedding. 
\end{theorem}

\section{Conclusions} \label{se:conclusions}

In this paper, we considered the problem of computing upward book embeddings of partitioned digraphs. Our research focused on the previously unsolved case of two pages, one of the ``ultimate'' algorithmic book embedding problems still open (see \cref{table:complexity}), and closed the complexity gap for the problem.  We also conceived a characterization of the upward embeddings that support such layouts, and leveraged such a characterization in combination with a number of algorithmic tools, such as flow techniques, SPQ-decompositions, and concise embedding encodings to obtain efficient testing algorithms for digraphs with a prescribed planar embedding and for biconnected directed partial~$2$-trees with a variable embedding. 

Our results coud be enhanced in two ways. First, the multiplicative linear overhead on the running time of our algorithm for biconnected directed partial~$2$-trees caused by rerouting the SPQ-tree at every Q-node might be avoided by using the techniques designed by Didimo et al.\ \cite{DBLP:conf/soda/DidimoLOP20}, see also \cite{DBLP:conf/gd/ChaplickGFGRS22,DBLP:journals/comgeo/Frati22}. Second, such an algorithm might be generalized to handle arbitrary, simply connected, partial $2$-trees; this appears to be a non-trivial task, as it requires handling the possible nestings of biconnected components into one another, as prescribed by the structure of the block-cut-vertex tree, while keeping cut-vertices $4$-modal and avoiding impossible faces. 

Additional interesting research directions are the following: 
\begin{itemize}
    \item  studying the complexity of the problem for single-source partitioned digraphs and, more generally, for digraphs with a bounded number of sources;
    \item determining whether the problem is \NP-complete for instances of bounded treewidth; and
    \item devising FPT algorithms with respect to parameters that are more restrictive than the treewidth.
\end{itemize} 

\medskip
{\bf Acknowledgements.} This research started at the Summer Workshop on Graph Drawing 2024 (SWGD 2024). The authors thank the other participants for useful discussions. 

\bibliographystyle{plainurl}
\bibliography{bibliography}
\end{document}